%% file: main.tex
\PassOptionsToPackage{dvipsnames}{xcolor}

\pdfoutput=1
\documentclass[a4paper,UKenglish,cleveref,autoref,thm-restate]{lipics-v2021}
\hideLIPIcs

\input{packages}
\input{macros}

\title{An MSO Framework for Weak-Memory Verification and Robustness}

\author{Giovanna Kobus Conrado}{Aarhus University, Denmark}{kobusgiovanna@gmail.com}{https://orcid.org/0000-0001-9474-6505}{}

\author{Andreas Pavlogiannis}{Aarhus University, Denmark}{pavlogiannis@cs.au.dk}{https://orcid.org/0000-0002-8943-0722}{}

\authorrunning{G. K. Conrado and A. Pavlogiannis}

\ccsdesc[500]{Theory of computation~Concurrency}
\ccsdesc[300]{Theory of computation~Formal languages and automata theory}
\ccsdesc[300]{Software and its engineering~Formal software verification}

\keywords{treewidth, monadic second order logic, reads-from robustness}

\begin{document}
\nolinenumbers
\maketitle

\input{abstract}

\input{introduction}
\input{preliminaries}

\input{tw_mm}
\input{mso_definability}

\input{rf_robustness}

\input{related_work}

\input{conclusion}


\bibliographystyle{plainurl}
\bibliography{bibliography}

\pagebreak

\appendix
\input{app_tw_mm}
\input{app_weak}
\input{app_experiments}
\input{app_mso}
\input{app_robustness}
\input{app_rc20}
\end{document}

%% file: packages.tex
\usepackage[utf8]{inputenc} 
\usepackage[dvipsnames]{xcolor}
\usepackage{amsmath, amsthm}
\usepackage{thmtools}
\usepackage{mathtools}
\usepackage{graphicx}
\usepackage{makecell} 
\usepackage{parskip}
\usepackage{array}
\usepackage{tabularx}
\usepackage{tikz} 
\usepackage{lineno}
\usepackage{csquotes}
\usepackage{thm-restate}
\usepackage{stackengine}
\usepackage{bm}
\usepackage{diagbox}
\usepackage{verbatim}
\usepackage{multirow}
\usepackage{stmaryrd}
\usepackage[inline]{enumitem}
\usepackage{mdframed}
\usepackage{listings}
\usepackage{subcaption}
\usepackage{hyperref}
\usepackage{wrapfig}

\usepackage{pgfplots}

\usepackage[ruled,linesnumbered]{algorithm2e}

\usepackage[capitalise]{cleveref}
\usepackage{todonotes}
\usepackage{booktabs}
\usepackage{pgfplotstable}
\usepackage{colortbl}
\usepackage[group-separator={,}, group-minimum-digits=4, text-series-to-math, propagate-math-font]{siunitx}
\usetikzlibrary {arrows.meta}
\usetikzlibrary{shapes,snakes}

\pgfdeclarelayer{background layer}
\pgfdeclarelayer{foreground layer}
\pgfsetlayers{background layer,main,foreground layer}

%% file: macros.tex
\newcommand\numberthis{\addtocounter{equation}{1}\tag{\theequation}}

\usetikzlibrary{arrows,positioning} 
\tikzset{
>=stealth',
punkt/.style={
       rectangle,
       rounded corners,
       draw=black, very thick,
       text width=6.5em,
       minimum height=2em,
       text centered},
pil/.style={
       ->,
       thick,
       shorten <=2pt,
       shorten >=2pt,},
main/.style={
    draw, rectangle,inner sep=4pt,
    }
}

\usetikzlibrary{automata,positioning,arrows}

\graphicspath{ {./} }

\renewcommand*\descriptionlabel[1]{\hspace\labelsep
  \normalfont{#1}}

\definecolor{mGreen}{rgb}{0,0.6,0}
\definecolor{mGray}{rgb}{0.5,0.5,0.5}
\definecolor{mPurple}{rgb}{0.58,0,0.82}
\definecolor{backgroundColour}{rgb}{0.95,0.95,0.92}
\definecolor{mybluecolor}{rgb}{0.2549019607843137, 0.2117647058823529, 0.9823529411764706}

\lstdefinestyle{CStyle}{
tabsize = 2, 
showstringspaces = false, 
backgroundcolor=\color{backgroundColour},   
numbers = left, 
commentstyle = \color{green}, 
keywordstyle = \color{blue}, 
stringstyle = \color{red}, 
rulecolor = \color{black}, 
basicstyle = \small \ttfamily , 
breaklines = true, 
numberstyle = \tiny,
language=Java
}

\newlist{compactitem}{itemize}{3} 
\setlist[compactitem]{label={\Large \textbullet},nosep,leftmargin=*}
\crefname{compactitemi}{Item}{Items}

\newlist{compactenum}{enumerate}{3} 
\setlist[compactenum]{label=(\arabic*), nosep,leftmargin=*}
\crefname{compactenumi}{Item}{Items}

\makeatletter
\newcommand{\customlabel}[2]{%
   \protected@write \@auxout {}{\string \newlabel {#1}{{#2}{\thepage}{#2}{#1}{}} }%
   \hypertarget{#1}{#2}
}
\makeatother

\newcommand{\mso}{\mathsf{MSO}}
\newcommand{\MSOSemantics}[1]{\llbracket{#1}\rrbracket}

\newcommand{\Prog}{\mathsf{Pseq}}
\newcommand{\ProgAug}{\mathsf{Pseq}^*}
\newcommand{\ConcProg}{\mathsf{P}}
\newcommand{\ConcProgAug}{\mathsf{P}^*}

\newcommand{\var}{\mathsf{v}}

\newcommand{\Events}{E}

\newcommand{\ModDom}{\mathsf{Mod}}
\newcommand{\FenceDom}{\mathsf{F}}
\newcommand{\LabelDom}{\mathsf{Lab}}
\newcommand{\ValueDom}{\mathsf{Val}}
\newcommand{\LocationDom}{\mathsf{Reg}}
\newcommand{\ThreadDom}{\mathsf{Tid}}
\newcommand{\EventDom}{\mathsf{E}}
\newcommand{\ReadDom}{\mathsf{R}}
\newcommand{\UnusedReadDom}{\mathsf{UnR}}
\newcommand{\UnusedReadDomX}{\mathsf{UnR}_x}
\newcommand{\WriteDom}{\mathsf{W}}

\newcommand{\RMWDom}{\mathsf{RMW}}

\newcommand{\mra}{\mathsf{acq}\text{-}\mathsf{rel}}
\newcommand{\mrel}{\mathsf{rel}}
\newcommand{\macq}{\mathsf{acq}}
\newcommand{\mrlx}{\mathsf{rlx}}

\newcommand{\List}{\mathcal{L}}

\newcommand{\Reduces}{\rhd}

\newcommand{\tw}{\mathsf{tw}}

\newcommand{\grid}[1]{\boxplus_{#1}}

\newcommand{\Vertices}{V}
\newcommand{\Edges}{L}
\newcommand{\Bag}{\mathcal{X}}
\newcommand{\Tree}{\mathsf{T}}

\newcommand{\True}{\mathsf{true}}

\newcommand{\LTS}{\mathcal{A}}

\newcommand{\States}{\mathcal{Q}}
\newcommand{\Alphabet}{\Sigma}
\newcommand{\Trans}{\Delta}
\newcommand{\InitState}{\mathsf{init}}

\newcommand{\tuple}[1]{\langle #1 \rangle}
\newcommand{\LTo}[1]{\xrightarrow{\raisebox{-0.5ex}[0ex][0ex]{\footnotesize $#1$} } }
\newcommand{\LTransTo}[1]{\xhookrightarrow{\raisebox{-0.5ex}[0ex][0ex]{\footnotesize $#1$} } }

\newcommand{\Paragraph}[1]{\smallskip\noindent{\bf #1}}
\newcommand{\SubParagraph}[1]{\smallskip\noindent{\em #1}}
\newcommand{\Nats}{\mathbb{N}}

\newcommand{\IntSet}[1]{[#1]}

\newcommand{\op}{\mathsf{op}}
\newcommand{\tid}{\mathsf{tid}}
\newcommand{\event}{e}
\newcommand{\id}{\mathsf{id}}
\newcommand{\llab}{\mathsf{lab}}
\newcommand{\acc}{\mathsf{acc}}
\newcommand{\lloc}{\mathsf{loc}}

\newcommand{\val}{\mathsf{val}}
\newcommand{\rd}{\mathtt{r}}

\newcommand{\wt}{\mathtt{w}}
\newcommand{\fc}{\mathtt{f}}
\newcommand{\rmw}{\mathtt{rmw}}
\newcommand{\ck}{\mathtt{ck}}
\newcommand{\tm}{\mathtt{error}}

\colorlet{colorPO}{brown}
\colorlet{colorSB}{darkgray}
\colorlet{colorRF}{blue}
\colorlet{colorCO}{red!80!black}
\colorlet{colorMO}{red!80!black}
\colorlet{colorOB}{orange}
\colorlet{colorFR}{Brown}
\colorlet{colorECO}{orange}
\colorlet{colorCOM}{magenta}
\colorlet{colorSW}{teal}
\colorlet{colorHB}{green!40!black}
\colorlet{colorPPO}{magenta}
\colorlet{colorRSEQ}{green!40!black}
\colorlet{colorSC}{violet}
\colorlet{colorGW}{brown}
\colorlet{colorPSC}{violet}
\colorlet{colorREL}{olive}
\colorlet{colorSO}{violet}
\colorlet{colorWB}{olive}
\colorlet{colorDOB}{violet}
\colorlet{colorRMW}{brown}
\colorlet{colorVARa}{Fuchsia}
\colorlet{colorVARb}{RedOrange}
\colorlet{colorVARc}{WildStrawberry}
\colorlet{colorDP}{Fuchsia}

\newcommand{\ProjectGraph}[3]{{#1}{\downharpoonright}_{#2}^{{#3}} }

\newcommand{\Lang}[1]{\mathcal{L}(#1)}

\newcommand{\Langrf}[1]{\mathcal{L}^\rf(#1)}

\newcommand{\ExecutionGraphsOf}[2]{\llbracket{#1}\rrbracket_{#2}}
\newcommand{\RFGraphsOf}[2]{\llbracket{#1}\rrbracket_{#2}^{\rf}}
\newcommand{\RedGraphsOf}[2]{\llbracket{#1}\rrbracket_{#2}^{\rf, \text{red}}}
\newcommand{\RFGraphsOfTw}[3]{\llbracket{#1}\rrbracket_{#2}^{\rf, #3}}
\newcommand{\RFGraphsMM}[1]{\llbracket{#1}\rrbracket^{\rf}}
\newcommand{\GraphsMM}[1]{\llbracket{#1}\rrbracket}

\newcommand{\identity}[1]{[#1]}
\newcommand{\seqbef}{\mathsf{\color{colorSB}sb}}
\newcommand{\po}{\mathsf{\color{colorPO}po}}
\newcommand{\ppo}{\mathsf{\color{colorPPO}ppo}}
\newcommand{\pppo}{\mathsf{\color{colorPPO}pppo}}

\newcommand{\rf}{\mathsf{\color{colorRF}rf}}

\newcommand{\mo}{\mathsf{\color{colorMO}mo}}

\newcommand{\moe}{\mathsf{\color{colorMO}moe}}

\newcommand{\fr}{\mathsf{\color{colorFR}fr}}

\newcommand{\hb}{\mathsf{\color{colorHB}hb}}
\newcommand{\hbrc}{\mathsf{\color{colorHB}hb_{\mathsf{RC}}}}

\newcommand{\rfe}{\mathsf{\color{colorRF}rfe}}

\newcommand{\fre}{\mathsf{\color{colorFR}fre}}

\newcommand{\sw}{\mathsf{\color{colorSW}sw}}

\newcommand\locx[1]{{{#1}_{x}}}
\newcommand{\acyclic}{\mathsf{acy}}
\newcommand{\irreflexive}{\mathsf{irr}}

\newcommand{\mm}{\mathsf{MM}}
\newcommand{\scmm}{\mathsf{SC}}
\newcommand{\wramm}{\mathsf{WRA}}
\newcommand{\ramm}{\mathsf{RA}}
\newcommand{\rlxmm}{\mathsf{Relaxed}}

\newcommand{\sramm}{\mathsf{SRA}}

\newcommand{\rcmm}{\mathsf{RC20}}
\newcommand{\tsomm}{\mathsf{TSO}}
\newcommand{\psomm}{\mathsf{PSO}}

\newcommand{\mmorder}{\preccurlyeq}
\newcommand{\rmwfree}[1]{#1^{\setminus\rmw}}

\newcommand{\msosw}{\Phi_{\sw}}
\newcommand{\msoprogsb}{\Phi_{\ConcProg}}
\newcommand{\msosb}{\Phi_{\seqbef}}
\newcommand{\msorf}{\Phi_{\rf}}
\newcommand{\msorfaug}{\Phi_{\rf}^*}
\newcommand{\msohb}{\varphi_{\hb}}
\newcommand{\msohbrc}{\varphi_{\hbrc}}
\newcommand{\msohbirr}{\Phi_{\ref{eq:hbirr}}}
\newcommand{\msohbirrx}{\Phi_{\ref{eq:hbirr}}^x}
\newcommand{\msopmosra}{\varphi_{\mo}^{\sramm}}
\newcommand{\msopmora}{\varphi_{\mo}^{\ramm}}
\newcommand{\msopmorc}{\varphi_{\mo}^{\rcmm}}
\newcommand{\msotriplet}{\varphi_{\text{conf}}}

\newcommand{\msorcwcra}{\Phi_{\text{(coherence)}}^{\ramm}}
\newcommand{\msorcwcrc}{\Phi_{\text{(coherence)}}^{\rcmm}}
\newcommand{\msorcwcrlx}{\Phi_{\text{(coherence)}}^{\rlxmm}}
\newcommand{\msorcswcra}{\Phi_{\text{(coherence)}}^{\rmwfree{\sramm}}}

\newcommand{\msora}{\Phi_{\ramm}}
\newcommand{\msorc}{\Phi_{\rcmm}}
\newcommand{\msowra}{\Phi_{\wramm}}

\newcommand{\msowrc}{\Phi_{\ref{eq:wrc}}}

\newcommand{\msowat}{\Phi_{\ref{eq:wat}}}
\newcommand{\msosranormw}{\Phi_{\rmwfree{\sramm}}}

\newcommand{\msorlx}{\Phi_{\rlxmm}}

\newcommand{\OVInstance}{\mathcal{I}_{\text{OV}}}

\tikzset{
   every path/.style={>=stealth},
   po/.style={->,color=colorPO, thick},
   seqbef/.style={->,color=colorSB, thick},
   ppo/.style={->,color=colorPPO, thick},
   sw/.style={->,color=colorSW, thick},
   rf/.style={->,color=colorRF,dashed, thick},
   mrf/.style={->,color=colorRF,dashed, thick},   
   urf/.style={->,color=colorRF, thick},
   fr/.style={->,color=colorFR,dashed, thick},
   hb/.style={->,color=colorHB,thick, thick},
   mo/.style={->,color=colorMO,dotted, very thick},
   dob/.style={->,color=colorDOB, very thick},
   ob/.style={->,color=colorOB, very thick},
   rmw/.style={->,color=colorRMW,thick, thick},
   rseq/.style={->,color=colorRSEQ,thick,dotted, thick},
   com/.style={->,color=colorCOM,thick, thick},
}

\tikzset{
    ncbar angle/.initial=90,
    ncbar/.style={
        to path=(\tikztostart)
        -- ($(\tikztostart)!#1!\pgfkeysvalueof{/tikz/ncbar angle}:(\tikztotarget)$)
        -- ($(\tikztotarget)!($(\tikztostart)!#1!\pgfkeysvalueof{/tikz/ncbar angle}:(\tikztotarget)$)!\pgfkeysvalueof{/tikz/ncbar angle}:(\tikztostart)$)
        -- (\tikztotarget)
    },
    ncbar/.default=0.5cm,
}

\tikzset{round left paren/.style={ncbar=0.5cm,out=110,in=-110}}
\tikzset{round right paren/.style={ncbar=0.5cm,out=70,in=-70}}

%% file: abstract.tex
\begin{abstract}

Memory models are formal specifications of concurrent-program executions, accounting for weak behaviors introduced by compiler and architectural optimizations.
The increase of their number and complexity has spawned efforts for uniform verification across whole classes of models, by axiomatizing the models in an adequate metatheory that admits a uniform treatment.
In this work, we formally study Monadic Second-Order logic (MSO) as a metatheory
for weak memory, by proving results on the treewidth and MSO-expressibility of various popular weak-memory models, as this combination allows us to uniformly tackle several verification problems.
In summary, our results are as follows.

First, we prove that executions under Sequential Consistency ($\scmm$) have bounded treewidth, while already those under Total Store Order ($\tsomm$) do not.
Second, we prove that a broad range of models, including Release/Acquire and the full RC20, are MSO-axiomatizable, while others, such as Strong Release/Acquire and $\tsomm$, are not, unless the Orthogonal Vectors problem -- which requires quadratic time under SETH -- can be solved in linear time.
Finally, we introduce the notion of \emph{reads-from robustness}, as an extension to recent work on coarse robustness criteria.
We show that our treewidth bounds (both upper and lower) have far-reaching algorithmic implications for any of our MSO-axiomatizable models $\mm$:~there is an algorithm that, for every program $\ConcProg$, either verifies $\ConcProg$ under $\mm$ or reports that $\ConcProg$ is not reads-from robust against $\mm$.
Overall, our results establish a rich and versatile theoretical framework for weak-memory verification and robustness.
\end{abstract}

%% file: introduction.tex
\section{Introduction}\label{SEC:INTRO}

The traditional sequential view ($\scmm$) of concurrent systems, introduced by Lamport in~\cite{Lamport1979}, does not reflect how programs execute on modern software and hardware platforms.
Compiler optimizations, cache hierarchies, instruction prefetching and speculative execution are only some of the means that a concurrent program may deviate from its sequential behaviors, introducing weak data consistency between the executing threads.
\emph{Weak memory models} are formal specifications of all subtle behavior that a program may exhibit in such settings, and are becoming a standard approach to rigorous concurrent programming, both for software~\cite{Batty2011,Vafeiadis2015,Lahav2016,Lahav2017,Kang2017,Dolan2018,Alglave2018,Lahav2022,Moiseenko2025} and hardware~\cite{Sewell2010,Alglave2012,Pulte2017,Podkopaev2019}.

Since weak behaviors can be tricky to predict, or even understand for non-experts,
there has been considerable work on developing program-analysis methods to support debugging, verification and testing of concurrent programs under weak memory, along various directions.
These include automated verification~\cite{Atig2010,Abdulla19,Abdulla2021,Lahav2022},
various notions of program robustness (e.g., data-race freedom, execution-graph robustness)~\cite{Gharachorloo1992,Owens2010,Bouajjani2011,Lahav2019},
stateless model checking~\cite{Abdulla2015,Abdulla2018,Bui2021,Kokologiannakis21},
testing~\cite{Luo2021,Gao2023,Tunc2023,Chakraborty2024,Margalit2025}, and
program logics~\cite{Vafeiadis2013,Lahav2015,Kaiser2017,Hammond2024}.
Since the behavior of a concurrent program is generally memory-model dependent,
virtually each such method is specific to a memory model.

The abundance and complexity of memory models has spawned an interest for unified theories that treat whole classes of memory models in a uniform way.
Recent developments enable bounded model checking and memory-model comparison for a range of memory models~\cite{Wickerson2017,Haas2022,Kokologiannakis2022,Kokologiannakis2023}.
Although promising, the full potential of unified theories for weak memory is still under exploration.
For example, there is no unified approach to reasoning about unbounded executions, performing program differentiation, or deciding program robustness against a range of memory models.

In the setting of message-passing concurrency, a unified theory based on Monadic Second Order Logic (MSO)~\cite{Courcelle1997} has been highly effective.
Indeed, most common message-passing protocols have been proven expressible in MSO~\cite{DiGiusto2023}, paving the road for utilizing Courcelle's celebrated metatheorem~\cite{Courcelle1990} for a variety of tasks, such as under-approximate verification over unbounded executions (by bounding some width parameter instead)~\cite{Cyriac2012,Cyriac2014} and deciding synchronizability~\cite{Bollig2021a}.
Perhaps surprisingly, however, the connections of MSO to weak memory have thus far remained unexplored.
\emph{Which memory models are axiomatizable in MSO and what common verification tasks can this be used for?}
\emph{What is the treewidth of weak-memory executions?}
\emph{Are there MSO-definable robustness criteria that are less sensitive than existing ones based on execution-graphs?}
In this work, we make the first steps to utilizing MSO as a unifying theory for weak memory, proving powerful theoretical capabilities, but also limitations.

\input{contributions}

%% file: contributions.tex
\subsection{Our Contributions}\label{SUBSEC:CONTRIBUTIONS}

Throughout this work, we represent program executions using \emph{reads-from ($\rf$-) graphs} of the form  $G=\tuple{\Events, \seqbef, \rf}$, which are analogous to Message Sequence Charts (MSCs) in message-passing systems~\cite{DiGiusto2023}.
Here, the events in $\Events$ act as vertices of $G$, while the \emph{sequenced-before} relation $\seqbef$ and the reads-from relation $\rf$ are binary relations over $\Events$.
In contrast to the more common \emph{execution graphs} $\tuple{\Events, \seqbef, \rf, \mo}$, $\rf$-graphs represent only program-observable behavior ---in particular, how each thread executes $(\seqbef)$ and which write a read obtains its value from $(\rf)$--- and abstract away the low-level architectural detail captured by the \emph{modification order} $\mo$, i.e., the order in which writes appear in the shared memory.

We make the following contributions.

\Paragraph{1.~The treewidth of weak memory.}
First, we study the treewidth of executions under various memory models.
We prove that executions under Sequential Consistency ($\scmm$), i.e., those that do not exhibit weak behavior, have bounded treewidth, but executions under any model as weak as Total Store Order ($\tsomm$) have unbounded treewidth.
We show that weak executions may still enjoy bounded treewidth, and present experiments that indicate that they tend to stay within the treewidth they would have under $\scmm$.

\Paragraph{2.~The MSO of weak memory.}
We then ask the question:~\emph{which popular weak-memory models can be axiomatized in MSO?}
We prove that this is the case for
Release/Acquire ($\ramm$)~\cite{Batty2011},
$\rlxmm$~\cite{Batty2011},
Weak Release/Acquire ($\wramm$, aka Causal Consistency~\cite{Bouajjani2017})~\cite{Lahav2022}, as well as 
the Read-Modify-Write (RMW)-free fragment of Strong Release/Acquire ($\sramm$, aka Causal Convergence~\cite{Bouajjani2017})~\cite{Lahav2016}.
We further remark that our MSO axiomatization extends to the full $\rcmm$ memory model for C/C++~\cite{Margalit2021}, although we defer its treatment to \cref{SEC:APP_RC20} due to space constraints.
The MSO axiomatization of each of the aforementioned memory models enables us to equip Courcelle's theorem~\cite{Courcelle1990} to tackle a range of fundamental algorithmic questions, in a \emph{unified way}.
We also prove that $\scmm$, $\tsomm$, Partial Store Order ($\psomm$), and the full $\sramm$ (with RMWs) are \emph{not axiomatizable in MSO}, unless the Orthogonal Vectors problem -- which requires quadratic time under SETH -- can be solved in linear time.

\Paragraph{3.~Reads-from robustness.}
We introduce the notion of \emph{reads-from ($\rf$-) robustness}.
This is a relaxation of standard execution-graph robustness~\cite{Bouajjani2011,Lahav2019}, and a fitting next step in the recent line of more permissive robustness criteria~\cite{Lahav2019,Margalit2021,Nagar2024}.
A program $\ConcProg$ is $\rf$-robust against a memory model $\mm$ if the set of $\rf$-graphs of $\ConcProg$ is the same under $\mm$ and $\scmm$.
We show a fruitful algorithmic interplay between $\rf$-robustness and verification:~for every MSO-definable memory model $\mm$ that satisfies a natural contiguity property (satisfied by all models in (2) except $\rlxmm$), there is an algorithm that either solves verification for $\ConcProg$ under $\mm$, or reports that $\ConcProg$ is not $\rf$-robust against $\mm$.
For memory models which have undecidable reachability, like $\ramm$, this result implies that the hard instances can be algorithmically classified as non-$\rf$-robust.
To our knowledge, $\rf$-robustness is the first computationally useful robustness notion that reasons purely about program-observable behaviors, disentangling them from lower-level architectural details involving the order in which writes appear on the shared memory.
We further extend the above to \emph{observational $\rf$-robustness}, which is a relaxation in a similar spirit to observational execution-graph robustness~\cite{Margalit2021}, that effectively filters out robustness violations due to read operations that are unused by the program.

%% file: preliminaries.tex
\section{Concurrent Programs and Memory Models}\label{SEC:PRELIM}

In this section we develop general notation, introduce the memory models we consider in this work, and define program semantics in an automata-theoretic way.

\Paragraph{General notation.} 
Given an integer $i$, we define $\IntSet{i}=\{1,\cdots, i\}$. 
For a binary relation $R$, we denote the reflexive, transitive, reflexive-transitive closures and inverse relations of $R$ as $R^?$, $R^+$, $R^*$, and $R^{-1}$, respectively. 
We occasionally call $(\event_1, \event_2)\in R$ an edge, and write it as $\event_1\LTo{R}\event_2$.
The composition of $R_1$ and $R_2$ is denoted by $R_1;R_2$. 
We write $\acyclic(R)$ and $\irreflexive(R)$ to denote that a relation is acyclic and irreflexive, respectively.
Given a set $S$, we write $\identity{S}$ for the identity relation on $S$.
We use $x \triangleq y$ to denote that $x$ is defined to be equal to $y$.

\Paragraph{Program domains.}
We consider concurrent programs consisting of threads over a finite domain $\ThreadDom$,
communicating over a finite set of shared registers $\LocationDom$, 
which store values over a finite value domain $\ValueDom$.

\input{executions}
\input{memory_models}

\input{programs}

%% file: executions.tex
\subsection{Executions}\label{SUBSEC:PRELIM_EG}

Here we set up standard notation for representing concurrent executions.

\Paragraph{Labels and events.}
A \emph{(event) label}\footnote{In related literature, events sometimes also carry an \emph{access mode}, but this will be clear from the context in our setting.} is either
a \emph{read label} $\rd(t,  x, v)$,
a \emph{write label} $\wt(t, x, v)$, or
a \emph{read-modify-write (RMW) label} $\rmw(t, x, v_{\rd}, v_{\wt})$, where 
$t\in \ThreadDom$ is a thread id,
$x\in\LocationDom$ is a shared register, and
$v, v_{\rd}, v_{\wt}\in \ValueDom$ are values.
We let $\LabelDom$ be the domain of labels.
An \emph{event} $e$ represents an execution step of the program, and is defined as $\event=\tuple{\id, \llab}$, where $\id\in \Nats$ is an identifier and $\llab\in \LabelDom$ is a label.
We write $\llab(e)$ for the label of $e$, $\tid(\event)$ and $\lloc(\event)$ for the thread id and register of $e$, while $\op(e)\in \{\rd, \wt, \rmw\}$ returns the operation of $e$.
If $\op(\event)\in \{ \rd, \rmw \}$, we write $\val_{\rd}(\event)$ for the value read by $\event$, and if $\op(\event)\in \{ \wt, \rmw \}$, we write $\val_{\wt}(\event)$ for the value written by $\event$.
We often identify events by their label when their id is not important, e.g. we may refer to events $\wt(t,x,v)$.
We occasionally ignore the value of an event when it is clear from the context or not important, e.g., we may refer to events $\wt(t,x)$, or even $\wt(x)$.
We let $\EventDom$ be the domain of events, and further distinguish the sets of 
read events $\ReadDom=\{ \event\in \EventDom\colon \op(\event)=\rd \}$, 
write events $\WriteDom=\{ \event\in \EventDom\colon \op(\event)=\wt \}$, and
RMW events $\RMWDom=\{ \event\in \EventDom\colon \op(\event)=\rmw \}$.
Given a set of events $X$ and a register $x\in \LocationDom$, we let $\locx{X}=\{e \in X\colon \lloc(\event)=x\}$, e.g., $\locx{\WriteDom}$ denotes all writes on register $x$.
Similarly, given a thread $t\in \ThreadDom$, we let $X^t=\{\event\in X\colon \tid(\event)=t \}$, e.g., $\ReadDom^t$ denotes all reads of thread $t$.
We extend the sub/super script notation to binary relations $R$ over $\EventDom$, i.e., $\locx{R}=[\locx{\EventDom}]; R ; [\locx{\EventDom}]$ and $R^t= [\EventDom^t]; R; [\EventDom^t]$.

\Paragraph{Execution graphs.}
In the context of weak memory, (concrete) program executions are often represented as execution graphs~\cite{Lahav2017} (aka candidate executions~\cite{Batty2011}).
An \emph{execution graph} is a tuple $G=\tuple{\Events, \seqbef, \rf, \mo}$, where 
$\Events\subseteq \EventDom$ is a finite set of (distinct) events,
and $\seqbef$, $\rf$, and $\mo$ are binary relations over $\Events$, satisfying the following conditions.
\begin{compactitem}
\item The \emph{sequenced-before order} $\seqbef$ records the order of events executed in each thread.
The $\seqbef$ relation is \emph{not} transitive, but instead relates each event to its immediate successor.
In particular, letting the \emph{program order} $\po \triangleq \seqbef^+$, 
we require that
\begin{enumerate*}[label=(\roman*)]
\item $\po^t$ is a total order for each $t\in \ThreadDom$, and
\item $\seqbef$ coincides with the (unique) transitive reduction of $\po$.
\end{enumerate*}
\item The \emph{reads-from relation} $\rf \subseteq \bigcup_{x\in \LocationDom}(\locx{\WriteDom} \cup \locx{\RMWDom}) \times (\locx{\ReadDom} \cup \locx{\RMWDom})$ relates a write/RMW event to a read/RMW event, and indicates that the latter reads its value from the former.
The values of the related events must match, i.e., for each 
$\event_1\LTo{\rf}\event_2$, we have $\val_{\wt}(e_1)=\val_{\rd}(e_2)$, 
and $\rf^{-1}$ must be a function, i.e.,
\begin{enumerate*}[label=(\roman*)]
\item if $\event_1\LTo{\rf}\event$ and $\event_2\LTo{\rf}\event$ then $e_1=e_2$, and
\item for all $e\in \Events\cap (\ReadDom\cup \RMWDom)$, there exists some $e'\in \Events$ such that $\event'\LTo{\rf}\event$.
\end{enumerate*}
\item The \emph{modification order} $\mo\subseteq \bigcup_{x\in \LocationDom} (\locx{\WriteDom} \cup \locx{\RMWDom}) \times (\locx{\WriteDom} \cup \locx{\RMWDom})$ intuitively captures the order in which the write events on each register take effect, i.e., $\mo_x$ is total for each $x\in \LocationDom$.
\end{compactitem}
We often write $G.\Events$, $G.\seqbef$, $G.\rf$ and $G.\mo$ for the respective components of $G$, or simply $\Events$, $\seqbef$, $\rf$ and $\mo$ when $G$ is clear from the context.
The requirement that every read event reads from some write event in $\Events$ reflects the assumption that registers hold no initial value.

\Paragraph{Reads-from and sequenced-before graphs.}
In our setting, we primarily represent program executions using abstractions of execution graphs.
In particular, a \emph{reads-from graph} (aka a partial execution~\cite{Tunc2023}) is a tuple $G=\tuple{\Events, \seqbef, \rf}$, i.e., it lacks the $\mo$ component of an execution graph.
Going one step further, a \emph{sequenced-before graph} (aka an abstract execution~\cite{Chakraborty2024}) is a tuple $G=\tuple{\Events, \seqbef}$, i.e., it lacks both the $\mo$ and $\rf$ components of an execution graph.

%% file: memory_models.tex
\subsection{Memory Models}\label{SUBSEC:PRELIM_MM}

As per standard in literature, we define memory models using a declarative approach. 
Each memory model is defined as a collection of axioms that the underlying execution graph $G$ must satisfy.
These axioms are phrased in terms of the relations of $G$, as well as some derived relations, defined below.

\Paragraph{Derived relations.}
Besides $\po$, the various consistency axioms use the following derived relations.
{
\footnotesize
\begin{center}
$\begin{array}{c}
\fr \triangleq \rf^{-1};\mo \qquad \hb \triangleq (\po \cup \rf)^{+} \qquad \rfe \triangleq \rf \setminus \po \qquad \moe \triangleq \mo \setminus \po\\[4pt]
\fre \triangleq \fr \setminus \po \qquad \ppo \triangleq \po \setminus (\WriteDom \times \ReadDom) \qquad \pppo \triangleq \ppo \setminus  (\WriteDom_x \times \WriteDom_y)_{x\neq y}
\end{array}$
\end{center}
}
The \emph{from-reads} relation $\fr$ orders a read/RMW before all the writes/RMWs that executed after its own writer.
The \emph{happens-before} relation $\hb$ captures causality chains in program execution.
The \emph{external relations} $\rfe$, $\moe$ and $\fre$ exclude same-thread events from the respective relations.
Finally the \emph{preserved program order} $\ppo$ and \emph{partially preserved program order} $\pppo$ exclude write-read and write-read/write-write orderings from $\po$.
For simplicity, we disregard RMWs in $\ppo$ and $\pppo$, as we only reason about the RMW-free fragments of the models they appear in ($\tsomm$ and $\psomm$).

\input{tables/axioms}
\input{tables/models}

\Paragraph{Memory models.}
In the scope of this work, a memory model $\mm$ is defined declaratively as a collection of consistency axioms, the full set of which is shown in \cref{tab:axioms}.
We use the notation $\mm^{\setminus\rmw}$ to refer to the RMW-free fragment of $\mm$.
At a high level, each axiom declares that a certain relation is acyclic or irreflexive, except for \ref{eq:wat} which requires that no two RMWs read from the same event.
\cref{tab:main_models} defines the memory models considered in this work.
$\scmm$ is the standard model of Sequential Consistency~\cite{Lamport1979},
while $\tsomm$ is the Total Store Order model, made popular by the x86 and SPARC architectures~\cite{Sewell2010},
and $\psomm$ is a further weakening known as Partial Store Order.
We only define their RMW-free fragments here as we only establish hardness results for them that already hold without RMWs.
$\ramm$ and $\rlxmm$ refer to the fragments of C11 using release/acquire and relaxed atomics, respectively~\cite{Lahav2019},
while $\sramm$ (Strong RA) and $\wramm$ (Weak RA) are stronger and weaker variants of $\ramm$, respectively~\cite{Lahav2022}, aka Causal Convergence and Causal Consistency~\cite{Bouajjani2017}.

\input{figures/mm_examples}

\Paragraph{Consistency.}
Given a memory model $\mm$ and an execution graph $G=\tuple{\Events, \seqbef, \rf, \mo}$, we say that $G$ is \emph{consistent in $\mm$}, denoted $G\models \mm$, if $G$ satisfies the consistency axioms of $\mm$.
This notion naturally extends to $\rf$-graphs, i.e., for an $\rf$-graph $G=\tuple{\Events, \seqbef, \rf}$, we have $G\models \mm$ if there exists a modification order $\mo$ such that resulting execution graph $G'=\tuple{\Events, \seqbef, \rf, \mo}$ yields $G'\models \mm$.
We let $\GraphsMM{\mm}=\{ G=\tuple{\Events, \seqbef, \rf, \mo}\colon G$ is an execution graph and  $G\models \mm \}$ and $\RFGraphsMM{\mm}=\{ G=\tuple{\Events, \seqbef, \rf}\colon  G \text{ is an $\rf$-graph and } G\models \mm \}$ be the sets of $\mm$-consistent execution graphs and $\rf$-graphs, respectively.
Note that, although $\ThreadDom$, $\LocationDom$ and $\ValueDom$ are fixed, $\RFGraphsMM{\mm}$ is infinite in general.
Given two memory models $\mm_1$, $\mm_2$, we say that $\mm_1$ is \emph{stronger than} $\mm_2$, written $\mm_1\mmorder\mm_2$, if $\GraphsMM{\mm_1}\subseteq\GraphsMM{\mm_2}$.
For example, $\scmm \mmorder \tsomm \mmorder \sramm \mmorder \ramm \mmorder \{\rlxmm, \wramm \}$.
\cref{fig:mm_examples} shows examples of execution graphs consistent in different memory models.

%% file: tables/axioms.tex
\begin{table}
\caption{\label{tab:axioms}
Consistency axioms.
}
\footnotesize
\begin{tabular}{lr|lr}
$\acyclic(\po \cup \rf \cup \mo \cup \fr)$ & \customlabel{eq:sc}{(sc)}  & 
$\irreflexive((\hb\cup \mo)^+)$ & \customlabel{eq:swc}{(strong-write-coherence)} \\
$\acyclic(\po_x \cup \rf_x \cup \mo_x \cup \fr_x)$ & \customlabel{eq:scloc}{(sc-per-loc)} & 
$\irreflexive(\hb)$ & \customlabel{eq:hbirr}{(irr-hb)} \\
$\acyclic(\ppo \cup \rfe \cup \moe \cup \fre)$ & \customlabel{eq:sce}{(external-hb)} & 
$\irreflexive(\mo;\hb)$ &  \customlabel{eq:wc}{(write-coherence)}\\
$\acyclic(\pppo \cup \rfe \cup \moe \cup \fre)$ &  \customlabel{eq:psce}{(pso-external-hb)} &
$\irreflexive(\mo;\hb;\rf^{-1})$ & \customlabel{eq:rc}{(read-coherence)}\\
$\irreflexive(\hb_l;[\WriteDom];\hb;\rf^{-1})$ & \customlabel{eq:wrc}{(weak-read-coherence)} &
$\irreflexive(\mo;\mo;\rf^{-1})$ & \customlabel{eq:at}{(atomicity)}\\
\multicolumn{4}{c}{
$\identity{\RMWDom};\rf^{-1};\rf;\identity{\RMWDom} = \identity{\RMWDom} $ \qquad \customlabel{eq:wat}{(weak-atomicity)}
}\\
\end{tabular}
\end{table}

%% file: tables/models.tex
\begin{table}
\caption{\label{tab:main_models}
Weak-memory models we consider in this work.
$\rmwfree{\mm}$ denotes the RMW-free fragment of $\mm$.
}
\centering
\footnotesize
\setlength\tabcolsep{10pt}
\begin{tabular}{lr|lr}
$\scmm$ & \ref{eq:sc} & $\wramm$ & \makecell{\ref{eq:hbirr}, \ref{eq:wrc},\\ \ref{eq:wat}} \\
\hline
$\rmwfree{\tsomm}$ & \ref{eq:scloc}, \ref{eq:sce}  & $\ramm$ & \makecell{\ref{eq:hbirr}, \ref{eq:wc},\\ \ref{eq:rc}, \ref{eq:at}}\\
\hline
$\rmwfree{\psomm}$ & \ref{eq:scloc}, \ref{eq:psce} & $\sramm$ & \makecell{\ref{eq:swc},\\ \ref{eq:rc}, \ref{eq:at}} \\
\hline
$\rlxmm$ & \ref{eq:scloc} && \\
\end{tabular}
\end{table}

%% file: figures/mm_examples.tex
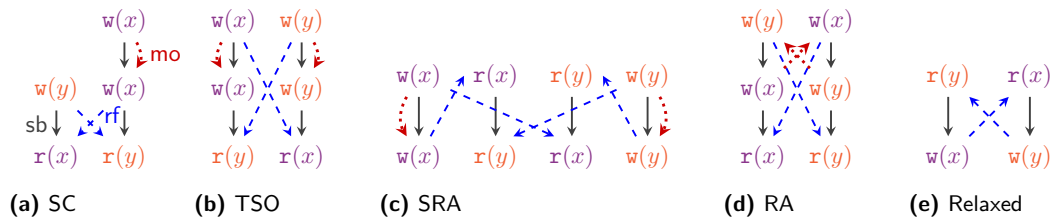
\begin{figure}

\def\xstep{0.4}
\def\ystep{0.4}
\def\scaleboxvalue{0.9}
\begin{subfigure}[b]{0.15\textwidth}
\scalebox{\scaleboxvalue}{
\begin{tikzpicture}[ line width=1pt, node distance = 1cm and 1.2cm] 

            \node (wy1) {\textcolor{colorVARb}{$\wt(y)$}};
            \node[below of = wy1] (rx0){\textcolor{colorVARa}{$\rd(x)$}};

            \node[right of = wy1] (wx1) {\textcolor{colorVARa}{$\wt(x)$}};
            \node[above of = wx1] (wx0) {\textcolor{colorVARa}{$\wt(x)$}};
            \node[below of = wx1] (ry0) {\textcolor{colorVARb}{$\rd(y)$}};

            \draw [seqbef] (wx0) to (wx1);
            \draw [seqbef] (wy1) to node [left=0pt]{$\seqbef$} (rx0);
            \draw [seqbef] (wx1) to (ry0);
            \draw [rf] (wx1) to (rx0);
            \draw [rf] (wy1) to node [near end, above right=-1pt]{$\rf$}(ry0);
            \draw [mo] (wx0) to[bend left] node [right=0pt]{$\mo$} (wx1);
\end{tikzpicture} 
}  
\caption{$\scmm$}
\label{subfig:aa}
\end{subfigure}
\hfill
\begin{subfigure}[b]{0.15\textwidth}

\scalebox{\scaleboxvalue}{
\begin{tikzpicture}[ line width=1pt, node distance = 1cm and 1.2cm] 
            \node[] (wy0) at (1,1.25){\textcolor{colorVARa}{$\wt(x)$}};
            \node[right of = wy0] (wx0) {\textcolor{colorVARb}{$\wt(y)$}}; 

            \node[below of = wx0] (wx1) {\textcolor{colorVARb}{$\wt(y)$}};
            \node[below of = wx1] (ry0) {\textcolor{colorVARa}{$\rd(x)$}}; 

            \node[below of = wy0] (wy1) {\textcolor{colorVARa}{$\wt(x)$}};
            \node[below of = wy1] (rx0){\textcolor{colorVARb}{$\rd(y)$}};

            \draw [seqbef] (wy0) to (wy1);
            \draw [seqbef] (wx0) to (wx1);
            \draw [seqbef] (wy1) to (rx0);
            \draw [seqbef] (wx1) to (ry0);
            \draw [rf] (wy0) to (ry0);
            \draw [rf] (wx0) to (rx0);
            \draw [mo] (wy0) to[bend right] (wy1);
            \draw [mo] (wx0) to[bend left] (wx1);

\end{tikzpicture} 
}  
\caption{$\tsomm$}
\label{subfig:bb}
\end{subfigure}
\hfill
\begin{subfigure}[b]{0.3\textwidth}

\scalebox{\scaleboxvalue}{
\begin{tikzpicture}[ line width=1pt, node distance = 1.2cm and 0.2cm] 

            \node(wx0) {\textcolor{colorVARa}{$\wt(x)$}};
            \node[below of = wx0] (wx1) {\textcolor{colorVARa}{$\wt(x)$}};

            \node(rx1)[right = of wx0] {\textcolor{colorVARa}{$\rd(x)$}};
            \node[below of = rx1] (ry0) {\textcolor{colorVARb}{$\rd(y)$}};

            \node(ry1)[right = of rx1] {\textcolor{colorVARb}{$\rd(y)$}};
            \node[below of = ry1] (rx0) {\textcolor{colorVARa}{$\rd(x)$}};

            \node(wy0)[right = of ry1] {\textcolor{colorVARb}{$\wt(y)$}};
            \node[below of = wy0] (wy1) {\textcolor{colorVARb}{$\wt(y)$}};
            
            \draw [seqbef] (wx0) to (wx1);
            \draw [seqbef] (rx1) to (ry0);
            \draw [seqbef] (ry1) to (rx0);
            \draw [seqbef] (wy0) to (wy1);
            \draw [rf] (wx1) to (rx1.180);
            \draw [rf] (wx0) to (rx0.130);
            \draw [rf] (wy1) to (ry1.0);
            \draw [rf] (wy0) to (ry0.50);
            \draw [mo] (wx0) to[bend right] (wx1);
            \draw [mo] (wy0) to[bend left] (wy1);
\end{tikzpicture}
}
\caption{$\sramm$}
\label{subfig:cc}
\end{subfigure}
\hfill
\begin{subfigure}[b]{0.15\textwidth}

\scalebox{\scaleboxvalue}{
\begin{tikzpicture}[ line width=1pt, node distance = 1cm and 1.2cm] 

            \node(wy1) {\textcolor{colorVARb}{$\wt(y)$}};
            \node[below of = wy1] (wx2) {\textcolor{colorVARa}{$\wt(x)$}};
            \node[below of = wx2] (rx0){\textcolor{colorVARa}{$\rd(x)$}};

            \node[right of = wy1] (wx1) {\textcolor{colorVARa}{$\wt(x)$}};
            \node[below of = wx1] (wy2) {\textcolor{colorVARb}{$\wt(y)$}};
            \node[below of = wy2] (ry0) {\textcolor{colorVARb}{$\rd(y)$}}; 

            \draw [seqbef] (wy1) to (wx2);
            \draw [seqbef] (wx2) to (rx0);
            \draw [seqbef] (wx1) to (wy2);
            \draw [seqbef] (wy2) to (ry0);
            \draw [rf] (wx1) to (rx0);
            \draw [rf] (wy1) to (ry0);
            \draw [mo] (wx2) to (wx1);
            \draw [mo] (wy2) to (wy1);
\end{tikzpicture} 
} 
\caption{$\ramm$}
\label{subfig:dd}
\end{subfigure}
\hfill
\begin{subfigure}[b]{0.15\textwidth}

\scalebox{\scaleboxvalue}{
\begin{tikzpicture}[ line width=1pt, node distance = 1.2cm] 

            \node (wy1) {\textcolor{colorVARb}{$\rd(y)$}};
            \node[below of = wy1] (rx0){\textcolor{colorVARa}{$\wt(x)$}};

            \node[right of = wy1] (wx1) {\textcolor{colorVARa}{$\rd(x)$}};
            \node[below of = wx1] (ry0) {\textcolor{colorVARb}{$\wt(y)$}}; 

            \draw [seqbef] (wy1) to (rx0);
            \draw [seqbef] (wx1) to (ry0);
            \draw [rf] (ry0) to (wy1);
            \draw [rf] (rx0) to (wx1);

\end{tikzpicture}    
}
\caption{$\rlxmm$}
\label{subfig:ee}
\end{subfigure}
\caption{
A sequence of increasingly weak execution graphs.
Each graph is consistent in its respective memory model and the models to its right, but inconsistent in the models to its left.
}
\label{fig:mm_examples}
\end{figure}

%% file: programs.tex
\subsection{A Model for Programs}\label{SUBSEC:PRELIM_PROG}

\Paragraph{Labeled Transition Systems.}
A \emph{Labeled Transition System} (LTS) is a tuple $\LTS=\tuple{\States, \Alphabet, \Trans, \InitState}$, where
\begin{enumerate*}[label=(\roman*)]
\item $\States$ is a finite set of states,
\item $\Alphabet$ is a finite alphabet,
\item $\Trans$ is a (labeled) transition relation over $\States\times (\Alphabet \ \cup \ \{\epsilon\}) \times \States$, and 
\item $\InitState\in \States$ is an initial state.
\end{enumerate*}
We often write $p\LTransTo{\alpha}q$ to denote $(p,\alpha, q)\in \Trans$.
The language of $\LTS$ is $\Lang{\LTS}=\{\alpha_1\alpha_2\dots\alpha_n\in \Alphabet^*\colon \InitState\LTransTo{\alpha_1}p_1\LTransTo{\alpha_2}\dots\LTransTo{\alpha_n}p_{n} \} $.
For clarity, we may represent a word as a sequence of comma-separated elements of $\Alphabet$, e.g, $\alpha_1,\alpha_2,\dots,\alpha_n\in \Lang{\LTS}$.

\Paragraph{Sequential programs.}
Given some thread $t\in \ThreadDom$, the \emph{sequential program} of $t$ is a transition system $\Prog_{t}=\tuple{\States_{t}, \Alphabet_{t}, \Trans_{t}, \InitState_{t}}$, where $\Alphabet_{t}=\LabelDom^{t}$.
Intuitively, each transition of $\Prog_{t}$ emits the label of an event executed by that thread.

\input{figures/prog_example}

\Paragraph{Concurrent programs and their semantics.}
A \emph{concurrent program} is a collection of sequential programs $\ConcProg=\tuple{\Prog_{t_1},\dots, \Prog_{t_{|\ThreadDom|}}}$.
Given an $\seqbef$-graph $G$ and a thread $t\in \ThreadDom$, let $\ProjectGraph{G}{\seqbef}{t}=\llab(\event_1), \llab(\event_2),\dots, \llab(\event_n)$, where $\event_1,\event_2,\dots, \event_n$ is the linear sequence of $G.\Events^t$ according to $\seqbef$.
The \emph{language} of $\ConcProg$ is defined as $\Lang{\ConcProg}=\{G=\tuple{\Events, \seqbef}\colon \text{ $G$ is an $\seqbef$-graph and } \forall t\in\ThreadDom.~ \ProjectGraph{G}{\seqbef}{t}\in \Lang{\Prog_{t}}\}$.
In words, $\Lang{\ConcProg}$ contains all $\seqbef$-graphs in which each thread $t$ executes according to its own sequential program $\Prog_{t}$.
Given a memory model $\mm$, we define 
the \emph{$\rf$-semantics} of $\ConcProg$ as $\RFGraphsOf{\ConcProg}{\mm}=\{ G=\tuple{\Events, \seqbef,\rf} \colon \tuple{\Events, \seqbef}\in\Lang{\ConcProg} \text{ and } G\models \mm  \}$, and
the \emph{execution semantics} of $\ConcProg$ as $\ExecutionGraphsOf{\ConcProg}{\mm}=\{ G=\tuple{\Events, \seqbef,\rf,\mo} \colon \tuple{\Events, \seqbef}\in\Lang{\ConcProg} \text{ and } G\models \mm  \}$.
See \cref{FIG:PROGEXAMPLE} for an example.
Note that $\RFGraphsOf{\ConcProg}{\mm}$ and $\ExecutionGraphsOf{\ConcProg}{\mm}$ are infinite in general, even though $\ConcProg$ is finite.

%% file: figures/prog_example.tex
\begin{figure}
\begin{subfigure}[b]{.5\textwidth}
    \begin{tikzpicture}[ line width=1pt, main/.style = {draw, circle, inner sep = 2pt}]

            \node[main, fill=white, draw=black, text=black, initial, initial text={}, initial distance=1em, initial where = above] (t11) at (0,0){$s_{1}$};
            \node[main, fill=white, draw=black, text=black] (t12) at (1.5,0){$s_{2}$};
            \node[main, fill=white, draw=black, text=black] (t13) at (3,0){$s_{3}$};

            \draw [->] (t11) tonode [below=2pt,colorVARa]{\footnotesize\textcolor{colorVARa}{\footnotesize$\rd(x,1)$}} (t12);
            \draw [->] (t12) to[bend left] node [above=2pt]{\footnotesize\textcolor{colorVARb}{\footnotesize$\wt(y,1)$}} (t13);
            \draw [->] (t13) to[bend left] node [below=2pt]{\footnotesize\textcolor{colorVARb}{\footnotesize$\wt(y,0)$}} (t12);

            \node[main, fill=white, draw=black, text=black, initial, initial text={}, initial distance=1em, initial where = above] (t21) at (4,0){$s_1$};
            \node[main, fill=white, draw=black, text=black] (t22) at (5.5,0){$s_2$};
            
            \draw [->] (t21) to[bend left] node [above=0pt]{\footnotesize\textcolor{colorVARa}{\footnotesize$\wt(x,1)$}} (t22);
            \draw [->] (t22) to[bend left] node [below=0pt]{\footnotesize\textcolor{colorVARb}{\footnotesize$\rd(y,1)$}} (t21);
        
        \end{tikzpicture}
\caption{$\Prog_{t_1}$ and $\Prog_{t_2}$.}
\label{SUBFIG:PROGEXAMPLE1}
\end{subfigure}%
\begin{subfigure}[b]{.23\textwidth}
\begin{tikzpicture}[ line width=1pt, main/.style = {draw, white, rectangle, inner sep = 2pt}]

            \node[main] (rx1) at (9.5,-1.35){\textcolor{colorVARa}{\footnotesize$\rd(x,1)$}};
            \node[main] (wy1) at (9.5,-2){\textcolor{colorVARb}{\footnotesize$\wt(y,1)$}};
            \node[main] (wy0) at (9.5,-2.65){\textcolor{colorVARb}{\footnotesize$\wt(y,0)$}};

            \node[main] (wx1) at (11,-1.35){\textcolor{colorVARa}{\footnotesize$\wt(x,1)$}};
            \node[main] (ry1) at (11,-2){\textcolor{colorVARb}{\footnotesize$\rd(y,1)$}}; 
            \node[main] (wx12) at (11,-2.65){\textcolor{colorVARa}{\footnotesize$\wt(x,1)$}};

            \draw [seqbef] (rx1) to (wy1);
            \draw [seqbef] (wx1) to (ry1);
            \draw [seqbef] (ry1) to (wx12);
            \draw [seqbef] (wy1) to (wy0);
        
        \end{tikzpicture}
\caption{$G\in \Lang{\ConcProg}$.}
\label{SUBFIG:PROGEXAMPLE2}
\end{subfigure}%
\begin{subfigure}[b]{.23\textwidth}
\begin{tikzpicture}[ line width=1pt, main/.style = {draw, white, rectangle, inner sep = 2pt}]

            \node[main] (rx1) at (9.5,-1.35){\textcolor{colorVARa}{\footnotesize$\rd(x,1)$}};
            \node[main] (wy1) at (9.5,-2){\textcolor{colorVARb}{\footnotesize$\wt(y,1)$}};
            \node[main] (wy0) at (9.5,-2.65){\textcolor{colorVARb}{\footnotesize$\wt(y,0)$}};

            \node[main] (wx1) at (11,-1.35){\textcolor{colorVARa}{\footnotesize$\wt(x,1)$}};
            \node[main] (ry1) at (11,-2){\textcolor{colorVARb}{\footnotesize$\rd(y,1)$}}; 
            \node[main] (wx12) at (11,-2.65){\textcolor{colorVARa}{\footnotesize$\wt(x,1)$}};

            \draw [seqbef] (rx1) to (wy1);
            \draw [seqbef] (wx1) to (ry1);
            \draw [seqbef] (ry1) to (wx12);
            \draw [seqbef] (wy1) to (wy0);

            \draw [rf] (wx1) to (rx1);
            \draw [rf] (wy1) to (ry1);
        
        \end{tikzpicture}
\caption{$G\in \RFGraphsOf{\ConcProg}{\scmm}$.}
\label{SUBFIG:PROGEXAMPLE3}
\end{subfigure}%
\caption{A concurrent program $\ConcProg=\tuple{\Prog_{t_1}, \Prog_{t_2}}$ (\subref{SUBFIG:PROGEXAMPLE1}), an $\seqbef$-graph in $ \Lang{\ConcProg}$ (\subref{SUBFIG:PROGEXAMPLE2}), and an $\rf$-graph in $\RFGraphsOf{\ConcProg}{\scmm}$(\subref{SUBFIG:PROGEXAMPLE3}).}
\label{FIG:PROGEXAMPLE}
\end{figure}

%% file: tw_mm.tex
\section{The Treewidth of Memory Models}\label{SEC:TWMM}

In this section we investigate the treewidth of consistent $\rf$-graphs under different memory models. 
We prove that $\scmm$ enjoys bounded-treewidth $\rf$-graphs but $\tsomm$ already does not.
Later, we use treewidth for the verification of concurrent programs under weak memory.

\input{treewidth}

\input{tw_sc_tso}

\input{tw_weak_behaviors}
\input{experiments}

%% file: treewidth.tex
\subsection{Treewidth}\label{SUBSEC:PRELIM_TW}

Treewidth is an elegant graph-theoretic notion that measures how similar a graph is to a tree~\cite{Robertson1986b}. Below, we introduce this notion formally.

\Paragraph{Tree decompositions and treewidth.}
Given a (di)graph $G=(\Vertices,\Edges)$, where $\Vertices$ is a set of vertices and $\Edges$ is a set of edges,
a \emph{tree decomposition} of $G$ is a tuple $\tuple{\{ \Bag_i\colon i\in I  \}, \Tree=(I, F)}$, where $\Tree$ is a tree and each $\Bag_i\subseteq \Vertices$ is called a \emph{bag},
such that the following conditions are met.
\begin{compactitem}
\item Every vertex appears in some bag, i.e., $\bigcup_{i \in I} \Bag_i = \Vertices$.
\item Every edge appears in some bag, i.e., for all $(u,v)\in \Edges$ there exists some $i\in I$ with $u,v\in \Bag_i$.
\item Every vertex appears in a connected subgraph of $\Tree$, i.e., for all $i,j,k\in I$, if $j$ is on the path from $i$ to $k$ in $\Tree$, we have $\Bag_i\cap \Bag_k\subseteq \Bag_j$.
\end{compactitem}
The \emph{width} of a tree decomposition is $\max_{i\in I} |\Bag_i|-1$.
The \emph{treewidth} of $G$, written $\tw(G)$, is the smallest width among the tree decompositions of $G$.
A tree decomposition of $G$ is \emph{optimal} if its width equals $\tw(G)$.
See \cref{FIG:DECEXAMPLE} in \cref{SEC:APP_TWMM} for an example. 

\Paragraph{The treewidth of $\rf$-graphs.}
The notion of tree decompositions and treewidth naturally extends to $\rf$-graphs $G=\tuple{\Events, \seqbef, \rf}$, where the vertex set is $\Events$ and the edge set is $\seqbef\cup \rf$.
We will thus be using the above notation for (plain) graphs and $\rf$-graphs interchangeably, and we will write $\tw(G)$ for the treewidth of an $\rf$-graph $G$.
Given a memory model $\mm$, we say that $\mm$ \emph{has bounded treewidth} if there exists some $k\in \Nats^+$ such that for every $G\in \RFGraphsMM{\mm}$, we have $\tw(G)\leq k$.

%% file: tw_sc_tso.tex
\subsection{On The Treewidth of Memory Models}\label{SUBSEC:TW_SC_TSO}

We now study the treewidth of different memory models.
Succinctly, our main question is \emph{what is the strongest memory model that gives rise to unbounded treewidth?}

\Paragraph{The treewidth of $\scmm$.}
We first show that $\scmm$ has bounded treewidth.

\smallskip
\begin{restatable}{lemma}{lemsctw}\label{lem:sc-tw}
For every $\rf$-graph $G$, if $G\models \scmm$ then $\tw(G)\leq |\ThreadDom|+|\LocationDom|$.
\end{restatable}

We remark that this bound is tight up to an additive term $+1$, in the sense that for all values of $|\ThreadDom|$ and $|\LocationDom|$, there exists an $\scmm$-consistent $\rf$-graph of treewidth $\geq |\ThreadDom|+|\LocationDom|-1$ (see \cref{lem:sc_tw_lowerbound} in \cref{SEC:APP_TWMM}).

\Paragraph{The treewidth of weak-memory.}
The $|\ThreadDom|+|\LocationDom|$ bound on the treewidth of $\scmm$ makes it natural to seek extensions for weak memory models.
Clearly, weaker models must admit executions of larger treewidth, that may become unbounded when the model becomes very weak.
It turns out that unbounded treewidth arises already in $\tsomm$, which might be somewhat surprising given how conceptually similar $\tsomm$ is to $\scmm$~\cite{Lahav2016a}.

\smallskip
\begin{restatable}{lemma}{lemtsotw}\label{lem:tso-tw}
For every $n>0$, there exists an $\rf$-graph $G_n\in \RFGraphsMM{\tsomm}$ of $O(n^2)$ events, $2$ threads and $2$ registers, such that $\tw(G_n)\geq n$.
\end{restatable}

Our proof constructs such a $G_n$ that has an $n\times n$ grid as a minor.
\cref{fig:grid_4} illustrates the $4\times 4$ grid and $G_4$.
Since the treewidth of an $n\times n$ grid is $n$ and treewidth is closed under minors~\cite{Robertson1986a}, this gives $\tw(G_n)\geq n$.

\input{figures/grid_4}

\smallskip
\begin{restatable}{corollary}{corweakmmtw}\label{cor:weak_mm-tw}
All memory models of \cref{tab:main_models} except $\scmm$ have unbounded treewidth.
\end{restatable}

%% file: figures/grid_4.tex
\begin{figure}[t]
\begin{subfigure}[b]{.4\textwidth}
\centering
\begin{tikzpicture}[ line width=1pt, node distance = 1cm] 

            \node[fill=blue!10, circle, line width=0.5pt] (n11) {\footnotesize$1,1$};
            \node[fill=blue!10,circle, line width=0.5pt, right = 0.5cm of n11] (n12){\footnotesize$1,2$};
            \node[fill=blue!10,circle, line width=0.5pt, right = 0.5cm of n12] (n13){\footnotesize$1,3$};
            \node[fill=blue!10,circle, right = 0.5cm of n13] (n14){\footnotesize$1,4$};
            
            \node[fill=blue!10, circle, line width=0.5pt, below = 0.5cm of n11] (n21) {\footnotesize$2,1$};
            \node[fill=blue!10,circle, line width=0.5pt, right = 0.5cm of n21] (n22){\footnotesize$2,2$};
            \node[fill=blue!10,circle, line width=0.5pt, right = 0.5cm of n22] (n23){\footnotesize$2,3$};
            \node[fill=blue!10,circle, line width=0.5pt, right = 0.5cm of n23] (n24){\footnotesize$2,4$};
            
            \node[fill=blue!10, circle, line width=0.5pt, below = 0.5cm of n21] (n31) {\footnotesize$3,1$};
            \node[fill=blue!10,circle, line width=0.5pt, right = 0.5cm of n31] (n32){\footnotesize$3,2$};
            \node[fill=blue!10,circle, line width=0.5pt, right = 0.5cm of n32] (n33){\footnotesize$3,3$};
            \node[fill=blue!10,circle, line width=0.5pt, right = 0.5cm of n33] (n34){\footnotesize$3,4$};
            
            \node[fill=blue!10, circle, line width=0.5pt, below = 0.5cm of n31] (n41) {\footnotesize$4,1$};
            \node[fill=blue!10,circle, line width=0.5pt, right = 0.5cm of n41] (n42){\footnotesize$4,2$};
            \node[fill=blue!10,circle, line width=0.5pt, right = 0.5cm of n42] (n43){\footnotesize$4,3$};
            \node[fill=blue!10,circle, line width=0.5pt, right = 0.5cm of n43] (n44){\footnotesize$4,4$};

            \node[fill=white, below = 1.5cm of n41] (n51) {};

            \draw (n11) -- (n12);
            \draw (n12) -- (n13);
            \draw (n13) -- (n14);
            \draw (n21) -- (n22);
            \draw (n22) -- (n23);
            \draw (n23) -- (n24);
            \draw (n31) -- (n32);
            \draw (n32) -- (n33);
            \draw (n33) -- (n34);
            \draw (n41) -- (n42);
            \draw (n42) -- (n43);
            \draw (n43) -- (n44);
            \draw[dashed] (n11) -- (n21);
            \draw[dashed] (n21) -- (n31);
            \draw[dashed] (n31) -- (n41);
            \draw[dashed] (n12) -- (n22);
            \draw[dashed] (n22) -- (n32);
            \draw[dashed] (n32) -- (n42);
            \draw[dashed] (n13) -- (n23);
            \draw[dashed] (n23) -- (n33);
            \draw[dashed] (n33) -- (n43);
            \draw[dashed] (n14) -- (n24);
            \draw[dashed] (n24) -- (n34);
            \draw[dashed] (n34) -- (n44);

\end{tikzpicture}
\caption{The $4\times 4$ grid.}
\label{subfig:grid_4_grid}
\end{subfigure}%
\hspace{1em}
\begin{subfigure}[b]{0.55\textwidth}
\centering
\begin{tikzpicture}[ line width=1pt,scale=1] 

        
        \node[colorVARa] (wx11) at (0,0){\footnotesize$\wt_{\tuple{1,1}}(x)$};
        \node[colorVARa, below = -0.2 cm of wx11] (wx12)
        {\footnotesize$\wt_{\tuple{1,2}}(x)$};
        \node[colorVARa, below = -0.2 cm of wx12] (wx13){\footnotesize$\wt_{\tuple{1,3}}(x)$}; 
        \node[colorVARa, below = -0.2 cm of wx13] (wx14){\footnotesize$\wt_{\tuple{1,4}}(x)$}; 
        \node[colorVARb, below = -0.2 cm of wx14] (ry21){\footnotesize$\rd_{\tuple{2,1}}(y)$}; 
        \node[colorVARa, below = -0.2 cm of ry21] (wx21){\footnotesize$\wt_{\tuple{2,1}}(x)$};
        \node[colorVARb, below = -0.2 cm of wx21] (ry22){\footnotesize$\rd_{\tuple{2,2}}(y)$}; 
        \node[colorVARa, below = -0.2 cm of ry22] (wx22){\footnotesize$\wt_{\tuple{2,2}}(x)$};
        \node[colorVARb, below = -0.2 cm of wx22] (ry23){\footnotesize$\rd_{\tuple{2,3}}(y)$}; 
        \node[colorVARa, below = -0.2 cm of ry23] (wx23){\footnotesize$\wt_{\tuple{2,3}}(x)$};
        \node[colorVARb, below = -0.2 cm of wx23] (ry24){\footnotesize$\rd_{\tuple{2,4}}(y)$}; 
        \node[colorVARa, below = -0.2 cm of ry24] (wx24){\footnotesize$\wt_{\tuple{2,4}}(x)$};
        \node[colorVARb, below = -0.2 cm of wx24] (ry31){\footnotesize$\rd_{\tuple{3,1}}(y)$}; 
        \node[colorVARa, below = -0.2 cm of ry31] (wx31){\footnotesize$\wt_{\tuple{3,1}}(x)$};
        \node[colorVARb, below = -0.2 cm of wx31] (ry32){\footnotesize$\rd_{\tuple{3,2}}(y)$}; 
        \node[colorVARa, below = -0.2 cm of ry32] (wx32){\footnotesize$\wt_{\tuple{3,2}}(x)$};
        \node[colorVARb, below = -0.2 cm of wx32] (ry33){\footnotesize$\rd_{\tuple{3,3}}(y)$}; 
        \node[colorVARa, below = -0.2 cm of ry33] (wx33){\footnotesize$\wt_{\tuple{3,3}}(x)$};
        \node[colorVARb, below = -0.2 cm of wx33] (ry34){\footnotesize$\rd_{\tuple{3,4}}(y)$}; 
        \node[colorVARa, below = -0.2 cm of ry34] (wx34){\footnotesize$\wt_{\tuple{3,4}}(x)$};
        \node[colorVARb, below = -0.2 cm of wx34] (ry41){\footnotesize$\rd_{\tuple{4,1}}(y)$}; 
        \node[colorVARb, below = -0.2 cm of ry41] (ry42){\footnotesize$\rd_{\tuple{4,2}}(y)$}; 
        \node[colorVARb, below = -0.2 cm of ry42] (ry43){\footnotesize$\rd_{\tuple{4,3}}(y)$}; 
        \node[colorVARb, below = -0.2 cm of ry43] (ry44){\footnotesize$\rd_{\tuple{4,4}}(y)$}; 

        \draw [decorate,decoration={brace,amplitude=2},thin] (wx11.190) -- (wx11.170) node (b11) [black,midway] {};
        \node[left = -0.1cm of b11] {\footnotesize$(1,1)$};
        
        \draw [decorate,decoration={brace,amplitude=2},thin] (wx12.190) -- (wx12.170) node (b12) [black,midway] {};
        \node[left = -0.1cm of b12] {\footnotesize$(1,2)$};

        \draw [decorate,decoration={brace,amplitude=2},thin] (wx13.190) -- (wx13.170) node (b13) [black,midway] {};
        \node[left = -0.1cm of b13] {\footnotesize$(1,3)$};

        \draw [decorate,decoration={brace,amplitude=2},thin] (wx14.190) -- (wx14.170) node (b14) [black,midway] {};
        \node[left = -0.1cm of b14] {\footnotesize$(1,4)$};

        \draw [decorate,decoration={brace,amplitude=2},thin] (wx21.190) -- (ry21.170) node (b21) [black,midway] {};
        \node[left = -0.1cm of b21] {\footnotesize$(2,1)$};

        \draw [decorate,decoration={brace,amplitude=2},thin] (wx22.190) -- (ry22.170) node (b22) [black,midway] {};
        \node[left = -0.1cm of b22] {\footnotesize$(2,2)$};

        \draw [decorate,decoration={brace,amplitude=2},thin] (wx23.190) -- (ry23.170) node (b23) [black,midway] {};
        \node[left = -0.1cm of b23] {\footnotesize$(2,3)$};

        \draw [decorate,decoration={brace,amplitude=2},thin] (wx24.190) -- (ry24.170) node (b24) [black,midway] {};
        \node[left = -0.1cm of b24] {\footnotesize$(2,4)$};

        \draw [decorate,decoration={brace,amplitude=2},thin] (wx31.190) -- (ry31.170) node (b31) [black,midway] {};
        \node[left = -0.1cm of b31] {\footnotesize$(3,1)$};

        \draw [decorate,decoration={brace,amplitude=2},thin] (wx32.190) -- (ry32.170) node (b32) [black,midway] {};
        \node[left = -0.1cm of b32] {\footnotesize$(3,2)$};

        \draw [decorate,decoration={brace,amplitude=2},thin] (wx33.190) -- (ry33.170) node (b33) [black,midway] {};
        \node[left = -0.1cm of b33] {\footnotesize$(3,3)$};

        \draw [decorate,decoration={brace,amplitude=2},thin] (wx34.190) -- (ry34.170) node (b34) [black,midway] {};
        \node[left = -0.1cm of b34] {\footnotesize$(3,4)$};

        \draw [decorate,decoration={brace,amplitude=2},thin] (ry41.190) -- (ry41.170) node (b41) [black,midway] {};
        \node[left = -0.1cm of b41] {\footnotesize$(4,1)$};
        \draw [decorate,decoration={brace,amplitude=2},thin] (ry42.190) -- (ry42.170) node (b42) [black,midway] {};
        \node[left = -0.1cm of b42] {\footnotesize$(4,2)$};
        \draw [decorate,decoration={brace,amplitude=2},thin] (ry43.190) -- (ry43.170) node (b43) [black,midway] {};
        \node[left = -0.1cm of b43] {\footnotesize$(4,3)$};
        \draw [decorate,decoration={brace,amplitude=2},thin] (ry44.190) -- (ry44.170) node (b44) [black,midway] {};
        \node[left = -0.1cm of b44] {\footnotesize$(4,4)$};

        \node[colorVARa] (rx11) at (2.5,0){\footnotesize$\rd_{\tuple{1,1}}(x)$};
        \node[colorVARb, below = -0.2 cm of rx11] (wy21)
        {\footnotesize$\wt_{\tuple{2,1}}(y)$};
        \node[colorVARa, below = -0.2 cm of wy21] (rx12)
        {\footnotesize$\rd_{\tuple{1,2}}(x)$};
        \node[colorVARb, below = -0.2 cm of rx12] (wy22)
        {\footnotesize$\wt_{\tuple{2,2}}(y)$};
        \node[colorVARa, below = -0.2 cm of wy22] (rx13)
        {\footnotesize$\rd_{\tuple{1,3}}(x)$};
        \node[colorVARb, below = -0.2 cm of rx13] (wy23)
        {\footnotesize$\wt_{\tuple{2,3}}(y)$};
        \node[colorVARa, below = -0.2 cm of wy23] (rx14)
        {\footnotesize$\rd_{\tuple{1,4}}(x)$};
        \node[colorVARb, below = -0.2 cm of rx14] (wy24)
        {\footnotesize$\wt_{\tuple{2,4}}(y)$};
        \node[colorVARa, below = -0.2 cm of wy24] (rx21)
        {\footnotesize$\rd_{\tuple{2,1}}(x)$};
        \node[colorVARb, below = -0.2 cm of rx21] (wy31)
        {\footnotesize$\wt_{\tuple{3,1}}(y)$};
        \node[colorVARa, below = -0.2 cm of wy31] (rx22)
        {\footnotesize$\rd_{\tuple{2,2}}(x)$};
        \node[colorVARb, below = -0.2 cm of rx22] (wy32)
        {\footnotesize$\wt_{\tuple{3,2}}(y)$};
        \node[colorVARa, below = -0.2 cm of wy32] (rx23)
        {\footnotesize$\rd_{\tuple{2,3}}(x)$};
        \node[colorVARb, below = -0.2 cm of rx23] (wy33)
        {\footnotesize$\wt_{\tuple{3,3}}(y)$};
        \node[colorVARa, below = -0.2 cm of wy33] (rx24)
        {\footnotesize$\rd_{\tuple{2,4}}(x)$};
        \node[colorVARb, below = -0.2 cm of rx24] (wy34)
        {\footnotesize$\wt_{\tuple{3,4}}(y)$};
        \node[colorVARa, below = -0.2 cm of wy34] (rx31)
        {\footnotesize$\rd_{\tuple{3,1}}(x)$};
        \node[colorVARb, below = -0.2 cm of rx31] (wy41)
        {\footnotesize$\wt_{\tuple{4,1}}(y)$};
        \node[colorVARa, below = -0.2 cm of wy41] (rx32)
        {\footnotesize$\rd_{\tuple{3,2}}(x)$};
        \node[colorVARb, below = -0.2 cm of rx32] (wy42)
        {\footnotesize$\wt_{\tuple{4,2}}(y)$};
        \node[colorVARa, below = -0.2 cm of wy42] (rx33)
        {\footnotesize$\rd_{\tuple{3,3}}(x)$};
        \node[colorVARb, below = -0.2 cm of rx33] (wy43)
        {\footnotesize$\wt_{\tuple{4,3}}(y)$};
        \node[colorVARa, below = -0.2 cm of wy43] (rx34)
        {\footnotesize$\rd_{\tuple{3,4}}(x)$};
        \node[colorVARb, below = -0.2 cm of rx34] (wy44)
        {\footnotesize$\wt_{\tuple{4,4}}(y)$};

        \draw [decorate,decoration={brace,amplitude=2},thin] (rx11.10) -- (wy21.-10) node (c11) [black,midway] {};
        \node[right = -0.1cm of c11] {\footnotesize$(1,1)\leftrightarrow(2,1)$};
        \draw [decorate,decoration={brace,amplitude=2},thin] (rx12.10) -- (wy22.-10) node (c11) [black,midway] {};
        \node[right = -0.1cm of c11] {\footnotesize$(1,2)\leftrightarrow(2,2)$};
        \draw [decorate,decoration={brace,amplitude=2},thin] (rx13.10) -- (wy23.-10) node (c11) [black,midway] {};
        \node[right = -0.1cm of c11] {\footnotesize$(1,3)\leftrightarrow(2,3)$};
        \draw [decorate,decoration={brace,amplitude=2},thin] (rx14.10) -- (wy24.-10) node (c11) [black,midway] {};
        \node[right = -0.1cm of c11] {\footnotesize$(1,4)\leftrightarrow(2,4)$};
        \draw [decorate,decoration={brace,amplitude=2},thin] (rx21.10) -- (wy31.-10) node (c11) [black,midway] {};
        \node[right = -0.1cm of c11] {\footnotesize$(2,1)\leftrightarrow(3,1)$};
        \draw [decorate,decoration={brace,amplitude=2},thin] (rx22.10) -- (wy32.-10) node (c11) [black,midway] {};
        \node[right = -0.1cm of c11] {\footnotesize$(2,2)\leftrightarrow(3,2)$};
        \draw [decorate,decoration={brace,amplitude=2},thin] (rx23.10) -- (wy33.-10) node (c11) [black,midway] {};
        \node[right = -0.1cm of c11] {\footnotesize$(2,3)\leftrightarrow(3,3)$};
        \draw [decorate,decoration={brace,amplitude=2},thin] (rx24.10) -- (wy34.-10) node (c11) [black,midway] {};
        \node[right = -0.1cm of c11] {\footnotesize$(2,4)\leftrightarrow(3,4)$};
        \draw [decorate,decoration={brace,amplitude=2},thin] (rx31.10) -- (wy41.-10) node (c11) [black,midway] {};
        \node[right = -0.1cm of c11] {\footnotesize$(3,1)\leftrightarrow(4,1)$};
        \draw [decorate,decoration={brace,amplitude=2},thin] (rx32.10) -- (wy42.-10) node (c11) [black,midway] {};
        \node[right = -0.1cm of c11] {\footnotesize$(3,2)\leftrightarrow(4,2)$};
        \draw [decorate,decoration={brace,amplitude=2},thin] (rx33.10) -- (wy43.-10) node (c11) [black,midway] {};
        \node[right = -0.1cm of c11] {\footnotesize$(3,3)\leftrightarrow(4,3)$};
        \draw [decorate,decoration={brace,amplitude=2},thin] (rx34.10) -- (wy44.-10) node (c11) [black,midway] {};
        \node[right = -0.1cm of c11] {\footnotesize$(3,4)\leftrightarrow(4,4)$};

        \draw  [rf] (wx11.east) -- (rx11.west);
        \draw  [rf] (wx12.east) -- (rx12.west); 
        \draw  [rf] (wx13.east) -- (rx13.west); 
        \draw  [rf] (wx14.east) -- (rx14.west);
        \draw  [rf] (wx21.east) -- (rx21.west);
        \draw  [rf] (wx22.east) -- (rx22.west); 
        \draw  [rf] (wx23.east) -- (rx23.west); 
        \draw  [rf] (wx24.east) -- (rx24.west); 
        \draw  [rf] (wx31.east) -- (rx31.west);
        \draw  [rf] (wx32.east) -- (rx32.west); 
        \draw  [rf] (wx33.east) -- (rx33.west); 
        \draw  [rf] (wx34.east) -- (rx34.west); 

        \draw  [rf] (wy21.west) -- (ry21.east);
        \draw  [rf] (wy22.west) -- (ry22.east);
        \draw  [rf] (wy23.west) -- (ry23.east);
        \draw  [rf] (wy24.west) -- (ry24.east);
        \draw  [rf] (wy31.west) -- (ry31.east);
        \draw  [rf] (wy32.west) -- (ry32.east);
        \draw  [rf] (wy33.west) -- (ry33.east);
        \draw  [rf] (wy34.west) -- (ry34.east);
        \draw  [rf] (wy41.west) -- (ry41.east);
        \draw  [rf] (wy42.west) -- (ry42.east);
        \draw  [rf] (wy43.west) -- (ry43.east);
        \draw  [rf] (wy44.west) -- (ry44.east);

        \draw  [fr] (rx13.west) -- (wx14.east);
        \draw  [fr] (ry21.east) -- (wy22.west);

    \end{tikzpicture}
\caption{The $\rf$-graph $G_4$ that has the $4\times 4$ grid as a minor.}
\label{subfig:grid_4_graph}
\end{subfigure}
\caption{
The $4\times 4$ grid is obtained from $G_4$ by contracting the events that share a subscript into the single vertex carrying that label, and then deleting every remaining edge except those joining vertices adjacent in \cref{subfig:grid_4_grid}.
Horizontal (filled) edges are encoded by the $\seqbef$ of the left thread, and vertical (dashed) edges are encoded by the $\seqbef$ of the right thread.
Note that $G_4$ is not $\scmm$-consistent, due to the cycle $\wt_{\tuple{1,4}}(x) \LTo{\po} \rd_{\tuple{2,1}}(y) \LTo{\fr} \wt_{\tuple{2,2}}(y)\LTo{\po} \rd_{\tuple{1,3}}(x)\LTo{\fr}\wt_{\tuple{1,4}}(x)$, with the highlighted $\fr$ edges.
}
\label{fig:grid_4}
\end{figure}
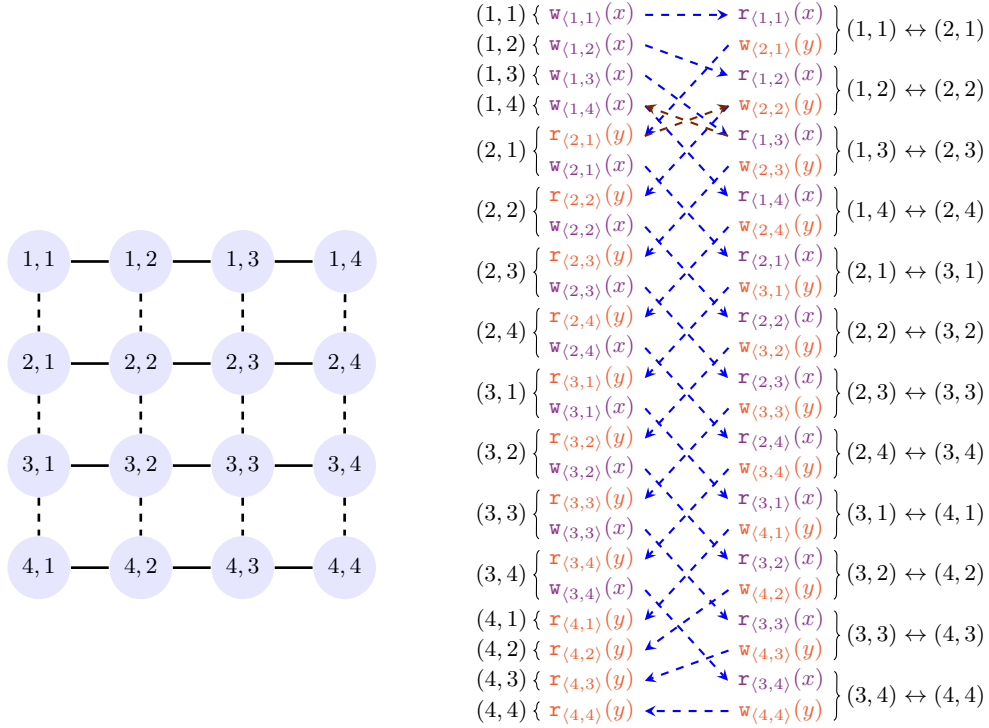

%% file: tw_weak_behaviors.tex
\Paragraph{Weak behaviors of bounded treewidth.}
In light of \cref{lem:tso-tw}, it is natural to ask how large treewidth relates to weak behavior at a conceptual level.
By generalizing our construction behind \cref{lem:sc-tw}, it can be shown that large treewidth requires very weak behavior, in the following sense.

\smallskip
\begin{restatable}{lemma}{lemhightwweak}\label{lem:high_tw_weak}
    For any $k\in \Nats$, if $\tw(G)\geq |\ThreadDom| + k \cdot |\LocationDom|$, then \emph{every} total extension $\sigma$ of $(G.\seqbef\cup G.\rf)^+$ contains an $\rf$-edge $\event_1\LTo{\rf}\event_2$ such that there are $\geq k$ writes/RMWs on $\lloc(\event_1)$ between $\event_1$ and $\event_2$ in $\sigma$.
\end{restatable}

It is worth highlighting, however, that the opposite is \emph{not} true: there are executions with very weak behavior that, nevertheless, enjoy small treewidth.
In \cref{FIG:WEAKTW}, we show 3 examples of $\rf$-graphs that display weak behaviors but have treewidth bounded by 3.
A minimum-width tree decomposition for \cref{SUBFIG:WEAKTW2} is shown in \cref{FIG:PATHDECESB}.
A complete discussion of these examples appears in \cref{SEC:APP_WEAK}.

\input{figures/weak_bounded_tw}
\input{figures/pathdec_esb}

%% file: figures/weak_bounded_tw.tex
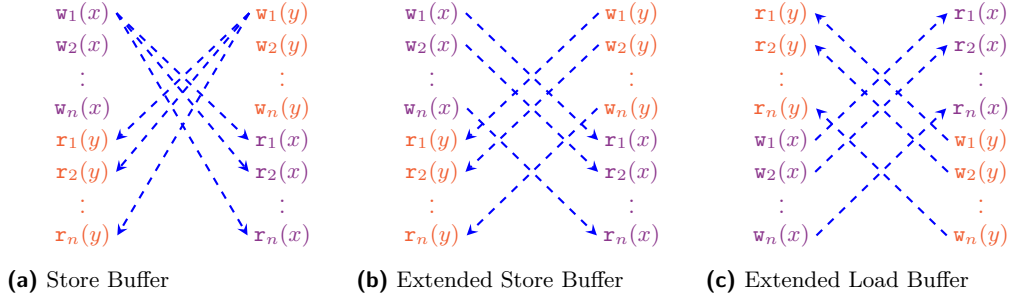
\begin{figure}[t]
\begin{subfigure}{.33\textwidth}
\centering
    \begin{tikzpicture}[ line width=1pt, main/.style = {draw,inner sep=2pt},scale=0.6,yscale=0.7, xscale=2.2] 

            \node[main, draw=white, fill=white, text=colorVARa] (wy1) at (-1,-2){\footnotesize$\wt_1(x)$};
            \node[main, draw=white, fill=white, text=colorVARa] (wy2) at (-1,-3){\footnotesize$\wt_2(x)$};
            \node[main, draw=none, fill=none, text=colorVARa] (wyd) at (-1,-4){\footnotesize$\vdots$};
            \node[main, draw=white, fill=white, text=colorVARa] (wyn) at (-1,-5){\footnotesize$\wt_n(x)$};

            \node[main, draw=white, fill=white, text=colorVARb] (wx1) at (1,-2){\footnotesize$\wt_1(y)$};
            \node[main, draw=white, fill=white, text=colorVARb] (wx2) at (1,-3){\footnotesize$\wt_2(y)$};
            \node[main, draw=none, fill=none, text=colorVARb] (wxd) at (1,-4){\footnotesize$\vdots$};
            \node[main, draw=white, fill=white, text=colorVARb] (wxn) at (1,-5){\footnotesize$\wt_n(y)$};

            \node[main, draw=white, fill=white, text=colorVARa] (ry1) at (1,-6){\footnotesize$\rd_1(x)$};
            \node[main, draw=white, fill=white, text=colorVARa] (ry2) at (1,-7){\footnotesize$\rd_2(x)$};
            \node[main, draw=none, fill=none, text=colorVARa] (ryd) at (1,-8){\footnotesize$\vdots$};
            \node[main, draw=white, fill=white, text=colorVARa] (ryn) at (1,-9){\footnotesize$\rd_n(x)$};

            \node[main, draw=white, fill=white, text=colorVARb] (rx1) at (-1,-6){\footnotesize$\rd_1(y)$};
            \node[main, draw=white, fill=white, text=colorVARb] (rx2) at (-1,-7){\footnotesize$\rd_2(y)$};
            \node[main, draw=none, fill=none, text=colorVARb] (rxd) at (-1,-8){\footnotesize$\vdots$};
            \node[main, draw=white, fill=white, text=colorVARb] (rxn) at (-1,-9){\footnotesize$\rd_n(y)$};

            \draw [rf] (wy1.east) -- (ry1.west) node [midway, above right = -2pt] {}; 
            \draw [rf] (wy1.east) -- (ry2.west) node [midway, above right = -2pt] {};
            \draw [rf] (wy1.east) -- (ryn.west) node [midway, above right = -2pt] {};

            \draw [rf] (wx1.west) -- (rx1.east) node [midway, above right = -2pt] {};
            \draw [rf] (wx1.west) -- (rx2.east) node [midway, above right = -2pt] {}; 
            \draw [rf] (wx1.west) -- (rxn.east) node [midway, above right = -2pt] {}; 
        
        \end{tikzpicture}
\caption{Store Buffer}
\label{SUBFIG:WEAKTW1}
\end{subfigure}%
\begin{subfigure}{.33\textwidth}
\centering
    \begin{tikzpicture}[ line width=1pt, main/.style = {draw,inner sep=2pt},scale=0.6,yscale=0.7, xscale=2.2] 

            \node[main, draw=white, fill=white, text=colorVARa] (wy1) at (-1,-2){\footnotesize$\wt_1(x)$};
            \node[main, draw=white, fill=white, text=colorVARa] (wy2) at (-1,-3){\footnotesize$\wt_2(x)$};
            \node[main, draw=none, fill=none, text=colorVARa] (wyd) at (-1,-4){\footnotesize$\vdots$};
            \node[main, draw=white, fill=white, text=colorVARa] (wyn) at (-1,-5){\footnotesize$\wt_n(x)$};

            \node[main, draw=white, fill=white, text=colorVARb] (wx1) at (1,-2){\footnotesize$\wt_1(y)$};
            \node[main, draw=white, fill=white, text=colorVARb] (wx2) at (1,-3){\footnotesize$\wt_2(y)$};
            \node[main, draw=none, fill=none, text=colorVARb] (wxd) at (1,-4){\footnotesize$\vdots$};
            \node[main, draw=white, fill=white, text=colorVARb] (wxn) at (1,-5){\footnotesize$\wt_n(y)$};

            \node[main, draw=white, fill=white, text=colorVARa] (ry1) at (1,-6){\footnotesize$\rd_1(x)$};
            \node[main, draw=white, fill=white, text=colorVARa] (ry2) at (1,-7){\footnotesize$\rd_2(x)$};
            \node[main, draw=none, fill=none, text=colorVARa] (ryd) at (1,-8){\footnotesize$\vdots$};
            \node[main, draw=white, fill=white, text=colorVARa] (ryn) at (1,-9){\footnotesize$\rd_n(x)$};

            \node[main, draw=white, fill=white, text=colorVARb] (rx1) at (-1,-6){\footnotesize$\rd_1(y)$};
            \node[main, draw=white, fill=white, text=colorVARb] (rx2) at (-1,-7){\footnotesize$\rd_2(y)$};
            \node[main, draw=none, fill=none, text=colorVARb] (rxd) at (-1,-8){\footnotesize$\vdots$};
            \node[main, draw=white, fill=white, text=colorVARb] (rxn) at (-1,-9){\footnotesize$\rd_n(y)$};

            \draw [rf] (wy1.east) -- (ry1.west) node [midway, above right = -2pt] {}; 
            \draw [rf] (wy2.east) -- (ry2.west) node [midway, above right = -2pt] {};
            \draw [rf] (wyn.east) -- (ryn.west) node [midway, above right = -2pt] {};

            \draw [rf] (wx1.west) -- (rx1.east) node [midway, above right = -2pt] {};
            \draw [rf] (wx2.west) -- (rx2.east) node [midway, above right = -2pt] {}; 
            \draw [rf] (wxn.west) -- (rxn.east) node [midway, above right = -2pt] {}; 
        
        \end{tikzpicture}
\caption{Extended Store Buffer}
\label{SUBFIG:WEAKTW2}
\end{subfigure}%
\begin{subfigure}{.33\textwidth}
\centering
\begin{tikzpicture}[ line width=1pt, main/.style = {draw,inner sep=2pt},scale=0.6,yscale=0.7, xscale=2.2]

            \node[main, draw=white, fill=white, text=colorVARa] (wy1) at (-1,-6){\footnotesize$\wt_1(x)$};
            \node[main, draw=white, fill=white, text=colorVARa] (wy2) at (-1,-7){\footnotesize$\wt_2(x)$};
            \node[main, draw=none, fill=none, text=colorVARa] (wyd) at (-1,-8){\footnotesize$\vdots$};
            \node[main, draw=white, fill=white, text=colorVARa] (wyn) at (-1,-9){\footnotesize$\wt_n(x)$};

            \node[main, draw=white, fill=white, text=colorVARb] (wx1) at (1,-6){\footnotesize$\wt_1(y)$};
            \node[main, draw=white, fill=white, text=colorVARb] (wx2) at (1,-7){\footnotesize$\wt_2(y)$};
            \node[main, draw=none, fill=none, text=colorVARb] (wxd) at (1,-8){\footnotesize$\vdots$};
            \node[main, draw=white, fill=white, text=colorVARb] (wxn) at (1,-9){\footnotesize$\wt_n(y)$};

            \node[main, draw=white, fill=white, text=colorVARa] (ry1) at (1,-2){\footnotesize$\rd_1(x)$};
            \node[main, draw=white, fill=white, text=colorVARa] (ry2) at (1,-3){\footnotesize$\rd_2(x)$};
            \node[main, draw=none, fill=none, text=colorVARa] (ryd) at (1,-4){\footnotesize$\vdots$};
            \node[main, draw=white, fill=white, text=colorVARa] (ryn) at (1,-5){\footnotesize$\rd_n(x)$};

            \node[main, draw=white, fill=white, text=colorVARb] (rx1) at (-1,-2){\footnotesize$\rd_1(y)$};
            \node[main, draw=white, fill=white, text=colorVARb] (rx2) at (-1,-3){\footnotesize$\rd_2(y)$};
            \node[main, draw=none, fill=none, text=colorVARb] (rxd) at (-1,-4){\footnotesize$\vdots$};
            \node[main, draw=white, fill=white, text=colorVARb] (rxn) at (-1,-5){\footnotesize$\rd_n(y)$};

            \draw [rf] (wy1.east) -- (ry1.west) node [midway, above right = -2pt] {}; 
            \draw [rf] (wy2.east) -- (ry2.west) node [midway, above right = -2pt] {};
            \draw [rf] (wyn.east) -- (ryn.west) node [midway, above right = -2pt] {};

            \draw [rf] (wx1.west) -- (rx1.east) node [midway, above right = -2pt] {};
            \draw [rf] (wx2.west) -- (rx2.east) node [midway, above right = -2pt] {}; 
            \draw [rf] (wxn.west) -- (rxn.east) node [midway, above right = -2pt] {}; 
        
        \end{tikzpicture}
\caption{Extended Load Buffer}
\label{SUBFIG:WEAKTW3}
\end{subfigure}%
\caption{
Examples of $\rf$-graphs with very weak behavior but low treewidth.
}
\label{FIG:WEAKTW}
\end{figure}

%% file: figures/pathdec_esb.tex
\begin{figure}[t]
\centering
\begin{tikzpicture}[ line width=1pt,scale=1] 

\node[draw, align=left, rounded corners] (a2) at (2.5,0){\textcolor{white}{\footnotesize$\rd_1(x)$}\\
\textcolor{colorVARa}{\footnotesize$\wt_{1}(x)$}\\
\textcolor{colorVARa}{\footnotesize$\rd_{1}(x)$}\\
\textcolor{colorVARa}{\footnotesize$\wt_{2}(x)$}};

\node[draw, align=left, rounded corners, right = 0.2cm of a2] (a3){\textcolor{white}{\footnotesize$\rd_1(x)$}\\
\textcolor{colorVARa}{\footnotesize$\rd_{1}(x)$}\\
\textcolor{colorVARa}{\footnotesize$\wt_{2}(x)$}\\
\textcolor{colorVARa}{\footnotesize$\rd_{2}(x)$}};

\node[draw, align=left, rounded corners, right = 0.2cm of a3] (a4){\textcolor{colorVARa!50}{\footnotesize$\rd_1(x)$}\\
\textcolor{colorVARa}{\footnotesize$\wt_{2}(x)$}\\
\textcolor{colorVARa}{\footnotesize$\rd_{2}(x)$}\\
\textcolor{colorVARa}{\footnotesize$\wt_{3}(x)$}};

\node[draw, align=left, rounded corners, right = 0.2cm of a4] (a5){\textcolor{colorVARa!50}{\footnotesize$\rd_1(x)$}\\
\textcolor{colorVARa}{\footnotesize$\rd_{2}(x)$}\\
\textcolor{colorVARa}{\footnotesize$\wt_{3}(x)$}\\
\textcolor{colorVARa}{\footnotesize$\rd_{3}(x)$}};

\node[draw, align=left, rounded corners, right = 0.5cm of a5] (a6){\textcolor{colorVARa!50}{\footnotesize$\rd_1(x)$}\\
\textcolor{colorVARa}{\footnotesize$\rd_{n-1}(x)$}\\
\textcolor{colorVARa}{\footnotesize$\wt_{n}(x)$}\\
\textcolor{colorVARa}{\footnotesize$\rd_{n}(x)$}};




\node[draw, align=left, rounded corners, right = 0.2cm of a6] (b2) {\textcolor{colorVARa!50}{\footnotesize$\rd_1(x)$}\\
\textcolor{colorVARa}{\footnotesize$\wt_{n}(x)$}\\
\textcolor{colorVARb}{\footnotesize$\wt_{1}(y)$}\\
\textcolor{colorVARb}{\footnotesize$\rd_{1}(y)$}};

\node[draw, align=left, rounded corners, right = 0.2cm of b2] (b3) {\textcolor{colorVARa!50}{\footnotesize$\rd_1(x)$}\\
\textcolor{colorVARb}{\footnotesize$\wt_{1}(y)$}\\
\textcolor{colorVARb}{\footnotesize$\rd_{1}(y)$}\\
\textcolor{colorVARb}{\footnotesize$\wt_{2}(y)$}};

\node[draw, align=left, rounded corners, right = 0.2cm of b3] (b4) {\textcolor{colorVARa!50}{\footnotesize$\rd_1(x)$}\\
\textcolor{colorVARb}{\footnotesize$\rd_{1}(y)$}\\
\textcolor{colorVARb}{\footnotesize$\wt_{2}(y)$}\\
\textcolor{colorVARb}{\footnotesize$\rd_{2}(y)$}};

\node[draw, align=left, rounded corners, right = 0.5cm of b4] (b5) {\textcolor{colorVARa!50}{\footnotesize$\rd_1(x)$}\\
\textcolor{colorVARb}{\footnotesize$\rd_{n-1}(y)$}\\
\textcolor{colorVARb}{\footnotesize$\wt_{n}(y)$}\\
\textcolor{colorVARb}{\footnotesize$\rd_{n}(y)$}};

\draw (a2) -- (a3);
\draw (a3) -- (a4);
\draw (a4) -- (a5);
\draw[dotted] (a5) -- (a6);
\draw (a6) -- (b2);
\draw (b2) -- (b3);
\draw (b3) -- (b4);
\draw[dotted] (b4) -- (b5);

\end{tikzpicture}
\caption{
A tree decomposition of the $\rf$-graph of \cref{SUBFIG:WEAKTW2}.
}
\label{FIG:PATHDECESB}
\end{figure}
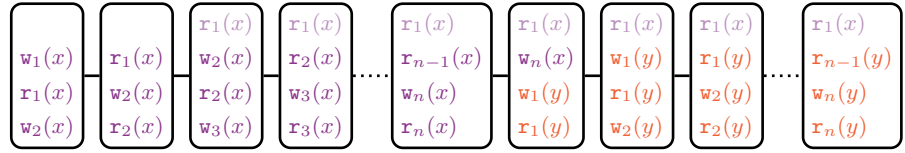

%% file: experiments.tex
\Paragraph{Experiments.}
In order to obtain an early assessment on the applicability of MSO for weak memory in practice, we gathered statistics on the treewidth of $\rf$-graphs of typical benchmarks in the weak-memory literature, taken from the GenMC repository~\cite{Kokologiannakis21}.
We find that even under weak memory, the executions of all benchmarks analyzed never exceed their treewidth bound under $\scmm$ that we prove in \cref{lem:sc-tw}.
A full description of the experimental analysis is presented in \cref{SEC:APP_EXPERIMENTS}.

%% file: mso_definability.tex
\section{The MSO Axiomatization of Memory Models}\label{SEC:MSO}

In this section we recall some basics on MSO (\cref{SUBSEC:MSO}), prove the MSO-definability of common weak-memory models (\cref{SUBSEC:MSO_DEFINABLE_MM}), highlight interesting algorithmic implications of this fact (\cref{SUBSEC:BOUNDED_TW_REACHABILITY}), and show the non-MSO-definability of other common models (\cref{SUBSEC:NON_MSO_DEFINABLE_MM}).
\looseness=-1

\input{mso}
\input{mso_definable_mm}
\input{reachability_implications}

\input{non_mso_definable_mm}

%% file: mso.tex
\subsection{The Monadic Second-Order Logic of Concurrent Executions}\label{SUBSEC:MSO}

We begin with MSO adapted to $\rf$-graphs, and recall Courcelle's theorem.

\Paragraph{Monadic Second Order Logic.}
The set of MSO formulas over $\rf$-graphs is given by the grammar
\[
\varphi ::= \True \ |\  \event_1=\event_2\ |\ \llab(\event)=\alpha \ |\ \seqbef(\event_1,\event_2)\ |\ \rf(\event_1,\event_2) \ |\ \event\in X\ |\ \varphi \lor \varphi \ |\  \neg\varphi \ | \ \exists \event. \varphi\ |\ \exists X. \varphi\
\]
where $\alpha\in \LabelDom$, $\event$, $\event_1,$ and $\event_2$ are first-order variables interpreted as events, and $X$ is a second order variable interpreted as a set of events.
We also use common shorthand notation such as $\neq$ and $\land$, as well as label predicates of the form $\op(x)\in \{\wt, \rmw\}$.
The semantics of MSO is rather intuitive -- we refer to~\cite{Courcelle1997} for a formal definition.
For instance, if $R$ is some relation expressible in MSO (e.g., $\seqbef$ or $\rf$), the transitive and reflexive-transitive closures of $R$ can be expressed by the following formulas.
\begin{align*}
\varphi_{R^+}(\event_1, \event_2) \triangleq \ &\neg \exists X. \left(\event_2\not \in X \land \forall \event_3, \event_4. \left((\event_1=\event_3 \lor \event_3\in X) \land  R(\event_3, \event_4) \implies \event_4\in X\right)   \right)
\numberthis\label{eq:mso_tc}\\
\varphi_{R^*}(\event_1, \event_2) \triangleq \ &\event_1 = \event_2 \lor \varphi_{R^+}(\event_1, \event_2)
\numberthis\label{eq:mso_rtc}
\end{align*}
We call an MSO formula \emph{closed} if it has no free variables.
Given a closed MSO formula $\varphi$ and an $\rf$-graph $G$, we write $G\models \varphi$ to denote that $G$ satisfies $\varphi$, and write $\MSOSemantics{\varphi}=\{ G\colon G\models \varphi \}$.
A class $\mathcal{C}$ of $\rf$-graphs is \emph{MSO-definable} if there exists an MSO formula $\varphi$ such that $\MSOSemantics{\varphi}=\mathcal{C}$.
Courcelle's celebrated metatheorem~\cite{Courcelle1990} gives a fixed-parameter tractable procedure for MSO model checking, which in our setting has the following form.

\smallskip
\begin{restatable}[Courcelle's Theorem]{theorem}{thmcourcelle}\label{thm:courcelle}
There exists an algorithm that, given a closed MSO formula $\varphi$, a natural number $k\in \Nats^+$ and a graph $G=\tuple{\Events, \seqbef, \rf}$, verifies whether $\tw(G)\leq k$ and $G\models \varphi$ in time $O(f(|\varphi|,k)\cdot |G|)$, for some computable function $f$.
\end{restatable}

The proof of \cref{thm:courcelle} also yields the following corollary regarding satisfiability in MSO~\cite{Courcelle2012}.

\smallskip
\begin{restatable}[Courcelle's Theorem]{corollary}{cormsosatisfiability}\label{cor:mso_satisfiability}
Given a closed MSO formula $\varphi$ and some $k\in \Nats^+$, the problem of whether there exists an $\rf$-graph $G=\tuple{\Events, \seqbef, \rf}$ with $\tw(G)\leq k$ and $G\models\varphi$ is decidable.
\end{restatable}

\Paragraph{The set of $\rf$-graphs.}
Towards our MSO treatment of program semantics under various memory models,
we start by establishing that the set of (well-formed) $\rf$-graphs is MSO-definable.

\smallskip
\begin{restatable}{lemma}{lemrfgraphsmsodefinable}\label{lem:rf_graphs_mso_definable}
The set $\{ G=\tuple{\Events, \seqbef, \rf}\colon G\text{ is an $\rf$-graph} \}$ is MSO-definable.
\end{restatable}

\Paragraph{The free $\rf$-graphs of a program.}
Given a concurrent program $\ConcProg$, we define its language of \emph{free $\rf$-graphs} as $\Langrf{\ConcProg}=\{ G=\tuple{\Events, \seqbef, \rf}\colon G \text{ is an $\rf$-graph and }\tuple{\Events, \seqbef}\in \Lang{\ConcProg}\}$, i.e., it is the set of  $\rf$-graphs which, if we ignore the $\rf$ component, are generated by $\ConcProg$. 
Note that here $G$ is not necessarily consistent in some memory model, hence the term ``free''.
The following lemma states that this language is MSO-definable.
It is established by conjuncting the MSO formula behind \cref{lem:rf_graphs_mso_definable} with an MSO formula $\msoprogsb$ which expresses that $\tuple{E, \seqbef} \in \Lang{\ConcProg}$, in spirit similar to the B\"uchi–Elgot–Trakhtenbrot theorem on the MSO definability of regular languages~\cite{Buchi1960}.

\smallskip
\begin{restatable}{lemma}{lemfreerfgraphsdefinable}\label{lem:free_rf_graphs_definable}
For any concurrent program $\ConcProg$, the language $\Langrf{\ConcProg}$ is MSO-definable.
\end{restatable}

In light of \cref{lem:free_rf_graphs_definable}, for any memory model $\mm$, the $\rf$-semantics $\RFGraphsOf{\ConcProg}{\mm}$ are MSO-definable as long as the set $\RFGraphsMM{\mm}$ (i.e., the memory model itself) is MSO-definable.
Towards this, in the coming sections we study the MSO definability of the memory models of \cref{tab:main_models}.

%% file: mso_definable_mm.tex
\subsection{The MSO-Definability of $\ramm$, $\rmwfree{\sramm}$, $\rlxmm$, and $\wramm$  }\label{SUBSEC:MSO_DEFINABLE_MM}

Here we prove the following theorem, by establishing the MSO-definability of the consistency axioms of the corresponding memory models.

\smallskip
\begin{restatable}{theorem}{thmmsodefinablemms}\label{thm:mso_definable_mms}
For any memory model $\mm\in \{\ramm, \rmwfree{\sramm}, \rlxmm, \wramm \}$ and  concurrent program $\ConcProg$,
the set $\RFGraphsOf{\ConcProg}{\mm}$ is MSO-definable.
\end{restatable}

\Paragraph{The case of $\ramm$.}
The MSO formula for $\ramm$ is a conjunction of three formulas: $\msohbirr$, enforcing \ref{eq:hbirr}; $\msowat$, enforcing \ref{eq:wat}; and $\msorcwcra$, enforcing \ref{eq:rc}, \ref{eq:wc}, and \ref{eq:at} under the assumption that \ref{eq:wat} already holds.

\SubParagraph{Expressing $\msohbirr$ and $\msowat$.}
The first two formulas are somewhat straightforward, as the corresponding axioms do not involve an $\mo$, and can thus be phrased directly on the components of an $\rf$-graph.
To define $\msohbirr$, we first express that two events $\event_1, \event_2$ are such that $\event_1\LTo{\hb}\event_2$ in the formula $\msohb=\varphi_{R^+}$, where $\varphi_{R^+}$ is the transitive-closure formula (\cref{eq:mso_tc}) and  $R(\event_1,\event_2)= \rf(\event_1,\event_2) \lor \seqbef(\event_1,\event_2)$.
Then
\[
\msohbirr \triangleq \lnot \exists \event. \msohb(\event,\event)
\numberthis\label{eq:mso_hbirr}
\]
To define $\msowat$, we observe that \ref{eq:wat} is even first-order definable.
\[
\msowat \triangleq \lnot \exists \event_1,\event_2,\event_3. (\event_2\neq \event_3 \land \rmw(\event_2) \land \rmw(\event_3) \land \rf(\event_1, \event_2) \land \rf(\event_1, \event_3))
\numberthis\label{eq:mso_wat}
\]
\SubParagraph{Expressing $\msorcwcra$.}
We now turn our attention to the coherence and atomicity axioms.
These are more challenging because they involve the $\mo$ relation, which is not present in an $\rf$-graph, and thus has to be guessed.
Our main insight behind the MSO definability of $\ramm$ is that there exists an MSO-definable, sound and complete partial modification order for $\ramm$, that can be used instead of $\mo$. 
Next, we outline our approach.

An \emph{RMW chain} is a maximal sequence $\wt_1\LTo{\rf}\rmw_2\LTo{\rf}\rmw_3\LTo{\rf}\cdots \LTo{\rf} \rmw_{m}$.
Any $\rf$-graph that satisfies \ref{eq:wat} (which is expressed independently by $\msowat$ in \cref{eq:mso_wat}) guarantees that all its write and RMW events are decomposable into disjoint RMW chains (i.e., every such event appears in exactly one RMW chain).
First, we observe that given any RMW event $\rmw$, the (unique) write $\wt$ heading the RMW chain of $\rmw$ can be captured in MSO.
This allows us to infer $\mo$-orderings only between write events ---and not between RMWs--- using MSO, in non-fixed-point style.
Second, we prove that the acyclicity of these inferred orderings is nevertheless sufficient and necessary for consistency.

We start with the following (first-order) formula that identifies a conflicting triplet $\tuple{\event_1, \event_2, \event_3}$ for which 
\begin{enumerate*}[label=(\roman*)]
\item $\op(\event_1)\in \{\wt,\rmw\}$, and
\item $\event_2\LTo{\rf}\event_3$,
\end{enumerate*}
since this may imply the existence of an $\mo$ edge.
\[
    \msotriplet(\event_1,\event_2,\event_3) \triangleq \event_1\neq \event_2 \land  \event_1\neq \event_3
     \land \lloc(\event_1)=\lloc(\event_2) \land \op(\event_1) \in \{\wt,\rmw\} \land  \rf(\event_2,\event_3)\numberthis
\]

The following formula infers $\mo$-orderings between write events $\event_1$, $\event_2$.
We infer $\event_1\LTo{\mo}\event_2$ if $\event_1\LTo{\hb}\event_2$, due to \ref{eq:wc}, or there exists a conflicting triplet $\tuple{\event_3, \event_4, \event_5}$ such that
\begin{enumerate*}[label=(\roman*)]
\item $\event_3\LTo{\hb} \event_5$,
\item $\event_3$ (resp., $\event_4$) is in the RMW chain headed by $\event_1$ (resp., $\event_2$), and
\item it is not the case that $\event_3$ and $\event_4$ are in the same RMW chain with $\event_3$ preceding $\event_4$,
\end{enumerate*}
due to \ref{eq:rc} and \ref{eq:at}.
\begin{align}
\msopmora(\event_1,&\event_2) \triangleq \op(e_1) = \wt \land \op(e_2) = \wt \land \lloc(\event_1)=\lloc(\event_2) \land \big[\msohb(\event_1,\event_2)  \lor \exists \event_3, \event_4, \event_5. \nonumber\\
& \big(\msotriplet(\event_3,\event_4,\event_5) \land \msohb(\event_3,\event_5)  \land \lnot \varphi_{\rf^+}(\event_3,\event_4) \land \varphi_{\rf^*}(\event_1, \event_3) \land \varphi_{\rf^*}(\event_2, \event_4) \big) \big]
\numberthis\label{eq:mso_pmora}
\end{align}

Finally, we express coherence and atomicity as the requirement that $\msopmora$ is acyclic, using the transitive closure formula  $\varphi_{R^+}$ from \cref{eq:mso_tc} with $R=\msopmora$,
\[
\msorcwcra \triangleq \lnot \exists \event. \varphi_{R^+}(\event,\event)
\numberthis\label{eq:mso_rc_wc_ra}
\]
and define
\[
\msora \triangleq \msohbirr \land \msowat \land \msorcwcra
\numberthis\label{eq:mso_ra}
\]

The following lemma, together with \cref{lem:rf_graphs_mso_definable}, concludes the MSO-definability of $\RFGraphsMM{\ramm}$.

\smallskip
\begin{restatable}{lemma}{lemramso}\label{lem:ra_mso}
For any $\rf$-graph $G=\tuple{\Events, \seqbef, \rf}$, we have $G\models \msora$ iff $G\models \ramm$.
\end{restatable}

The case of $\ramm$ acts as a template for the other memory models, which admit simpler axiomatizations.
We describe these here at a high level, and refer to \cref{SUBSEC:APP_MSO_DEF} for their full lemmas and MSO formulas. 

\Paragraph{The case of $\rlxmm$.}
The model $\rlxmm$ requires only \ref{eq:scloc}.
Recall that, for a single register, $\rlxmm$ and $\ramm$ coincide.
It thus suffices to tailor the formulas for $\ramm$ to a single register, and conjoin the resulting per-register formulas, together with $\msowat$, over all $x \in \LocationDom$.

\Paragraph{The case of  $\rmwfree{\sramm}$.}
The RMW-free fragment of $\sramm$ requires \ref{eq:swc} and \ref{eq:rc}.
As for $\ramm$, we infer $\mo$ orderings between writes via an MSO-definable partial order $\msopmosra$, here capturing only the orderings due to \ref{eq:rc}.
This is simpler than $\msopmora$ for two reasons:~we do not need to infer $\wt_1\LTo{\mo}\wt_2$ from $\wt_1\LTo{\hb}\wt_2$, as \ref{eq:swc} mixes $\hb$ and $\mo$; and, since there are no RMW chains, an ordering $\wt_1\LTo{\mo}\wt_2$ arises only from a conflicting triplet $\tuple{\wt_1, \wt_2, \rd}$ with $\wt_1\LTo{\hb}\rd$.

\Paragraph{The case of $\wramm$.}
Recall that we have already defined $\msohbirr$ and $\msowat$, respectively expressing \ref{eq:hbirr} and \ref{eq:wat}.
It remains to express \ref{eq:wrc}, which is straightforward, as it is directly translatable to MSO.

%% file: reachability_implications.tex
\subsection{Algorithmic Implications}\label{SUBSEC:BOUNDED_TW_REACHABILITY}

We now state some direct algorithmic implications of \cref{lem:free_rf_graphs_definable} and \cref{thm:mso_definable_mms}, by equipping  Courcelle's theorem to solve a number of fundamental algorithmic problems concerning the behavior of  programs restricted to bounded-treewidth graphs.
Given a memory model $\mm$, a concurrent program $\ConcProg$, and some $k\in \Nats^+$, let $\RFGraphsOfTw{\ConcProg}{\mm}{k}=\{G\in \RFGraphsOf{\ConcProg}{\mm}\colon \tw(G)\leq k\}$.

\Paragraph{Reachability.}
One of the most fundamental algorithmic questions is that of state reachability:~given a concurrent program $\ConcProg=\tuple{\Prog_{t_1},\dots, \Prog_{t_{|\ThreadDom|}}}$ and a sequence of states $\tuple{s_1, \dots, s_{{|\ThreadDom|}}}$,
is there an execution that, for each $i \in [|\ThreadDom|]$, brings $\Prog_{t_i}$ to state $s_i$?
A natural generalization to this question is to decide whether $\ConcProg$ exhibits an execution that satisfies an arbitrary MSO specification.
As the following lemma states, this question is decidable when restricted to graphs of bounded treewidth.

\smallskip
\begin{restatable}{lemma}{lemboundedtwreachability}\label{lem:bounded_tw_reachability}
For any memory model $\mm\in \{ \ramm, \rlxmm, \rmwfree{\sramm}, \wramm \}$, concurrent program $\ConcProg$, natural number $k\in \Nats^+$ and MSO formula $\varphi$, the question of whether there exists an $\rf$-graph $G\in \RFGraphsOfTw{\ConcProg}{\mm}{k}$ such that $G\models \varphi$ is decidable.
\end{restatable}

Note that state reachability for $\ramm$ is undecidable~\cite{Abdulla19}, even when restricted to RMW-free executions~\cite{Conrado26}; furthermore, although state reachability for $\sramm$ and $\wramm$ is decidable~\cite{Lahav2022}, MSO reachability is not, due to the fact that execution graphs under these models include large grid minors (\cref{SUBSEC:TW_SC_TSO}), and MSO satisfiability over grids is undecidable~\cite{Courcelle2012}.

\Paragraph{Verification.}
Another fundamental problem is that of verification: given a concurrent program $\ConcProg$ and a specification $\varphi$, is it the case that all executions of $\ConcProg$ satisfy $\varphi$?
Since MSO is closed under negations, \cref{lem:bounded_tw_reachability} implies that this problem is decidable for $\mm\in \{ \ramm, \rlxmm, \rmwfree{\sramm}, \wramm \}$ when $\varphi$ is an MSO formula and we consider only executions of bounded treewidth.

\Paragraph{Robustness and differentiation.}
Given two memory models $\mm_1$ and $\mm_2$ with $\mm_1\mmorder\mm_2$, the robustness question of a concurrent program $\ConcProg$ wrt $(\mm_1, \mm_2)$ asks whether $\ConcProg$ exhibits executions under $\mm_2$ that are not possible under $\mm_1$.
On the other hand, given two concurrent programs $\ConcProg_1$ and $\ConcProg_2$ under a memory model $\mm$, the differentiation question asks whether $\ConcProg_1$ exhibits behaviors that are not possible under $\ConcProg_2$.
The MSO definability of memory models implies that robustness and differentiation are decidable over graphs of bounded treewidth.

\smallskip
\begin{restatable}{lemma}{lemboundedtwdiff}\label{lem:bounded_tw_diff}
For any $\mm_1, \mm_2\in \{ \ramm, \rlxmm, \rmwfree{\sramm}, \wramm \}$, two concurrent programs $\ConcProg_1, \ConcProg_2$, and natural number $k\in \Nats^+$, the question of whether $\RFGraphsOfTw{\ConcProg_1}{\mm_1}{k}\subseteq \RFGraphsOfTw{\ConcProg_2}{\mm_2}{k}$ is decidable.
\end{restatable}

%% file: non_mso_definable_mm.tex
\subsection{Non-MSO-Definability for $\scmm$, $\tsomm$, $\psomm$ and $\sramm$}\label{SUBSEC:NON_MSO_DEFINABLE_MM}

We now turn our attention to memory models which are not expressible in MSO, under standard beliefs in complexity theory.

\Paragraph{Orthogonal vectors.} 
The orthogonal vectors problem (OV) is stated on two sets $X, Y$, each containing $n$ boolean $d$-dimensional vectors. 
A solution to OV is a pair of vectors $\tuple{x_1,\cdots, x_d}\in X, \tuple{y_1,\cdots, y_d}\in Y$ that is orthogonal, i.e., $\sum_{i\in \IntSet{d}} x_i \cdot y_i = 0$. 
OV is easily solvable in $O(n^2 \cdot d)$ and $O(2^d \cdot n)$ time.
The OV hypothesis states that this problem has no $O(n^{2-\epsilon} \cdot d^c)$ algorithm, for any fixed $\epsilon, c>0$~\cite{Bringmann19}.
We establish the following theorem, which states that a hypothetical MSO definability of the stated memory models would require, not only to break the OV hypothesis, but to actually achieve \emph{nearly-linear} complexity for OV.

\smallskip
\begin{restatable}{theorem}{thmnonmsodefinability}\label{thm:non_mso_definability}
For any memory model $\mm\in \{\rmwfree{\scmm}, \sramm, \rmwfree{\tsomm}, \rmwfree{\psomm} \}$, the set $\RFGraphsMM{\mm}$ is not MSO-definable unless OV is solvable in $O(n\cdot d)$ time.
\end{restatable}

Naturally, the impossibility of \cref{thm:non_mso_definability} extends to memory models with mixed accesses that subsume any of the above, such as the C++ memory model with SC accesses~\cite{Batty2011}, and the POWER architecture for programs compiled from the release/acquire fragment of C++, since those enjoy $\sramm$ semantics~\cite{Lahav2016}.
The proof for each case of \cref{thm:non_mso_definability} is shown in \cref{SUBSEC:APP_NON_MSO_DEFINABLE_MM}.
\looseness=-1

%% file: rf_robustness.tex
\section{Reads-From Robustness}\label{SEC:RFROB}

Robustness against a memory model $\mm$ is the property of a concurrent program $\ConcProg$ of not exhibiting new behaviors under $\mm$ compared to $\scmm$.
A robust program enjoys the best of $\mm$ and $\scmm$:~its weak-memory operations provide efficiency, while its $\scmm$-only behaviors makes it amenable to simpler reasoning and analysis methods.
Robustness comes in different variants, based on what exactly constitutes program behavior in this context. 
In this section we show that the treewidth and MSO, as developed in the previous sections, have appealing applications to robustness.

\Paragraph{State robustness.}
One of the least sensitive notions of robustness is wrt to states:~$\ConcProg$ is state robust if every state of $\ConcProg$ that is reachable via an execution in $\ExecutionGraphsOf{\ConcProg}{\mm}$ is also reachable via an execution in $\ExecutionGraphsOf{\ConcProg}{\scmm}$.
State robustness against $\mm$ is decidable iff state reachability under $\mm$ is decidable~\cite{Derevenetc2015}.
In particular for $\ramm$, this implies the undecidability of state robustness~\cite{Abdulla19}.

\Paragraph{Execution-graph robustness.}
A concurrent program $\ConcProg$ is execution-graph robust against a memory model $\mm$ if $\ExecutionGraphsOf{\ConcProg}{\mm}=\ExecutionGraphsOf{\ConcProg}{\scmm}$.
Execution graph robustness is more sensitive than state robustness, and is decidable against $\tsomm$~\cite{Bouajjani2011} and $\ramm$~\cite{Lahav2019}.

\Paragraph{Reads-from robustness.}
In this work, we introduce the notion of reads-from robustness (or $\rf$-robustness).
A concurrent program $\ConcProg$ is $\rf$-robust against a memory model $\mm$ if $\RFGraphsOf{\ConcProg}{\mm}=\RFGraphsOf{\ConcProg}{\scmm}$.
It is easy to see that execution-graph robustness implies $\rf$-robustness, and this relationship is strict, in the sense that there exist programs that are $\rf$-robust but not execution-graph robust (see \cref{subfig:rob_example1,subfig:rob_example2} for an example).
Hence, this notion is a step closer to state robustness, by being insensitive to $\mo$.

\input{figures/rob_example}

\input{rf_robustness_reachability}

\input{observational_rf_robustness}

%% file: figures/rob_example.tex
\begin{figure}
\centering
\def\xstep{1}
\def\ystep{1}
\def\scalevalue{1}
\begin{subfigure}[b]{0.24\textwidth}
\centering
\scalebox{\scalevalue}{
\begin{tikzpicture}[ line width=1pt, node distance = 1cm] 
            \node[] (wx0) at (1,1.25){\textcolor{colorVARa}{\footnotesize$\wt(x)$}};
            \node[right = 0.5cm of wx0] (wx1) {\textcolor{colorVARa}{\footnotesize$\wt(x)$}}; 

            \node[below of = wx0] (wy0) {\textcolor{colorVARb}{\footnotesize$\wt(y)$}};
            \node[below of = wx1] (wy1) {\textcolor{colorVARb}{\footnotesize$\wt(y)$}}; 

            \node[below of = wy0] (rx0) {\textcolor{colorVARa}{\footnotesize$\rd(x)$}};
            \node[below of = wy1] (ry0){\textcolor{colorVARb}{\footnotesize$\rd(y)$}};

            \draw [seqbef] (wx0) to (wy0);
            \draw [seqbef] (wy0) to (rx0);
            \draw [seqbef] (wx1) to (wy1);
            \draw [seqbef] (wy1) to (ry0);
            \draw [rf] (wx0) to[bend right=40] (rx0);
            \draw [rf] (wy0) to (ry0);
            \draw [mo] (wx0) to (wx1);
            \draw [mo] (wy1) to (wy0);

\end{tikzpicture}
}
\caption{$G_1 \in \ExecutionGraphsOf{\ConcProg}{\ramm}\setminus \ExecutionGraphsOf{\ConcProg}{\scmm}$}
\label{subfig:rob_example1}
\end{subfigure}
\hfill
\begin{subfigure}[b]{0.24\textwidth}
\centering
\scalebox{\scalevalue}{
\begin{tikzpicture}[ line width=1pt, node distance = 1cm] 

            \node[] (wx0) at (1,1.25){\textcolor{colorVARa}{\footnotesize$\wt(x)$}};
            \node[right = 0.5cm of wx0] (wx1) {\textcolor{colorVARa}{\footnotesize$\wt(x)$}}; 

            \node[below of = wx0] (wy0) {\textcolor{colorVARb}{\footnotesize$\wt(y)$}};
            \node[below of = wx1] (wy1) {\textcolor{colorVARb}{\footnotesize$\wt(y)$}}; 

            \node[below of = wy0] (rx0) {\textcolor{colorVARa}{\footnotesize$\rd(x)$}};
            \node[below of = wy1] (ry0){\textcolor{colorVARb}{\footnotesize$\rd(y)$}};

            \draw [seqbef] (wx0) to (wy0);
            \draw [seqbef] (wy0) to (rx0);
            \draw [seqbef] (wx1) to (wy1);
            \draw [seqbef] (wy1) to (ry0);
            \draw [rf] (wx0) to[bend right=40] (rx0);
            \draw [rf] (wy0) to (ry0);
            \draw [mo] (wx1) to (wx0);
            \draw [mo] (wy1) to (wy0);
\end{tikzpicture}
}
\caption{$G_2 \in \ExecutionGraphsOf{\ConcProg}{\ramm}\cap \ExecutionGraphsOf{\ConcProg}{\scmm}$}
\label{subfig:rob_example2}
\end{subfigure}
\hfill
\begin{subfigure}[b]{0.24\textwidth}
\centering
\scalebox{\scalevalue}{
\begin{tikzpicture}[ line width=1pt, node distance = 0.75cm] 

            \node[] (wx0) at (1,1.25){\textcolor{colorVARa}{\footnotesize$\wt(x)$}};
            \node[below of = wx0] (wy0) {\textcolor{colorVARb}{\footnotesize$\wt(y)$}};
            \node[below of = wy0] (rx0) {\textcolor{colorVARa}{\footnotesize$\rd(x)$}};
            \node[below of = rx0] (ckx0) {\textcolor{colorVARa}{\footnotesize$\ck(x)$}};
            
            \node[right = 0.5cm of wx0] (wx1) {\textcolor{colorVARa}{\footnotesize$\wt(x)$}}; 
            \node[below of = wx1] (wy1) {\textcolor{colorVARb}{\footnotesize$\wt(y)$}}; 
            \node[below of = wy1] (ry0){\textcolor{colorVARb}{\footnotesize$\rd(y)$}};
            \node[below of = ry0] (cky0){\textcolor{colorVARb}{\footnotesize$\ck(y)$}};
            \node[below of = cky0, draw=colorVARa] (rx1){\textcolor{colorVARa}{\footnotesize$\rd(x)$}};

            \draw [seqbef] (wx0) to (wy0);
            \draw [seqbef] (wy0) to (rx0);
            \draw [seqbef] (wx1) to (wy1);
            \draw [seqbef] (wy1) to (ry0);
            \draw [seqbef] (rx0) to (ckx0);
            \draw [seqbef] (ry0) to (cky0);
            \draw [seqbef] (cky0) to (rx1);
            \draw [rf] (wx0.200) to[bend right=40] (rx0.160);
            \draw [rf] (wx1) to[bend left=40] (rx1);
            \draw [rf] (wy0) to (ry0);
            \draw [mo] (wx0) to (wx1);
            \draw [mo] (wy1) to (wy0);
\end{tikzpicture}
}
\caption{$G_3 \in \ExecutionGraphsOf{\ConcProgAug}{\ramm}\setminus \ExecutionGraphsOf{\ConcProgAug}{\scmm}$}
\label{subfig:rob_example3}
\end{subfigure}
\hfill
\begin{subfigure}[b]{0.24\textwidth}
\centering
\scalebox{\scalevalue}{
\begin{tikzpicture}[ line width=1pt, node distance = 0.75cm] 

            \node[] (wx0) at (1,1.25){\textcolor{colorVARa}{\footnotesize$\wt(x)$}};
            \node[below of = wx0] (wy0) {\textcolor{colorVARb}{\footnotesize$\wt(y)$}};
            \node[below of = wy0] (rx0) {\textcolor{colorVARa}{\footnotesize$\rd(x)$}};
            \node[below of = rx0] (ckx0) {\textcolor{colorVARa}{\footnotesize$\ck(x)$}};
            
            \node[right = 0.5cm of wx0] (wx1) {\textcolor{colorVARa}{\footnotesize$\wt(x)$}}; 
            \node[below of = wx1] (wy1) {\textcolor{colorVARb}{\footnotesize$\wt(y)$}}; 
            \node[below of = wy1, inner sep=0pt] (ry0){\textcolor{colorVARb}{\footnotesize$\rd(y)$}};
            \node[below of = ry0] (cky0){\textcolor{colorVARb}{\footnotesize$\ck(y)$}};
            \node[below of = cky0, draw=colorVARa] (rx1){\textcolor{colorVARa}{\footnotesize$\rd(x)$}};

            \draw [seqbef] (wx0) to (wy0);
            \draw [seqbef] (wy0) to (rx0);
            \draw [seqbef] (wx1) to (wy1);
            \draw [seqbef] (wy1) to (ry0);
            \draw [seqbef] (rx0) to (ckx0);
            \draw [seqbef] (ry0) to (cky0);
            \draw [seqbef] (cky0) to (rx1);
            \draw [rf] (wx0.200) to[bend right=40] (rx0.160);
            \draw [rf] (wx0.-25) to[out=-10, in=170, looseness=0.5] (rx1.155);
            \draw [rf] (wy0) to (ry0);
            \draw [mo] (wx1) to (wx0);
            \draw [mo] (wy1) to (wy0);
\end{tikzpicture}
}
\caption{$G_4 \in \ExecutionGraphsOf{\ConcProgAug}{\ramm}\cap \ExecutionGraphsOf{\ConcProgAug}{\scmm}$}
\label{subfig:rob_example4}
\end{subfigure}
\caption{A program $\ConcProg$ that only produces the execution graphs shown in  (\subref{subfig:rob_example1}) and (\subref{subfig:rob_example2}) is $\rf$-robust, but not execution-graph robust.
The augmented program $\ConcProgAug$ that only produces the execution graphs shown in  (\subref{subfig:rob_example3}) and (\subref{subfig:rob_example4}) is observationally $\rf$-robust, but not execution-graph robust, $\rf$-robust, or observationally execution-graph robust.
}
\label{fig:rob_example}
\end{figure}
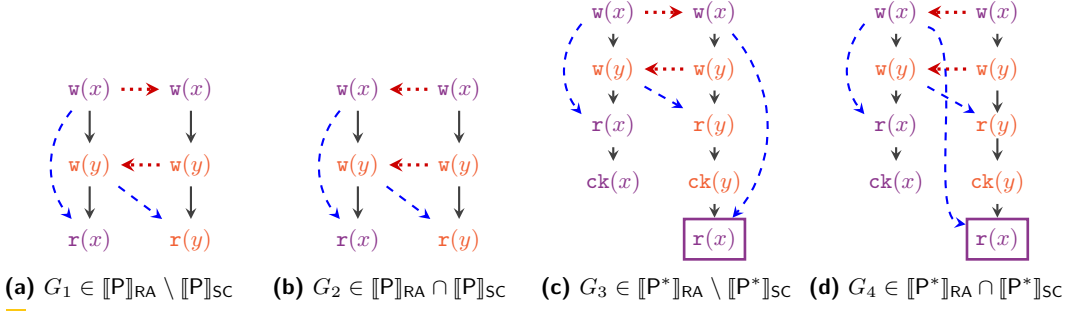

%% file: rf_robustness_reachability.tex
\subsection{Reads-From Robustness and State Reachability}\label{subsec:rf-robustness}

Is $\rf$-robust decidable for common memory models?
This is a challenging question, but,
perhaps surprisingly, we can establish a conceptually stronger result:~for MSO-definable memory models $\mm$ that further exhibit a contiguity property, every program $\ConcProg$ either enjoys decidable state reachability, or it is ``obviously'' not $\rf$-robust; and this distinction is decidable.
Our key insight is that treewidth divergence is both
\begin{enumerate*}[label=(\roman*)]
\item necessary for undecidable reachability/state robustness, due to \cref{lem:bounded_tw_reachability}, and
\item sufficient for $\rf$-non-robustness, due to \cref{lem:sc-tw}.
\end{enumerate*}
In the following, we make these insights formal.

\Paragraph{$(\seqbef \cup\rf)$-contiguity.}
We call a memory model $\mm$ \emph{$(\seqbef\cup \rf)$-contiguous} if 
\begin{enumerate*}[label*=(\roman*)]
\item for every $\rf$-graph $G\in \RFGraphsMM{\mm}$, $(G.\seqbef\cup G.\rf)$ is acyclic, and
\item for every program $\ConcProg$ and $G\in \RFGraphsOf{\ConcProg}{\mm}$, removing an event $e$ from $G$ that is maximal wrt $(\seqbef\cup\rf)^+$ (together with any adjacent $\seqbef$- or $\rf$-edges) yields a graph $G'\in \RFGraphsOf{\ConcProg}{\mm}$.
\end{enumerate*}
All memory models in \cref{tab:main_models}, except $\rlxmm$, are $(\seqbef\cup \rf)$-contiguous.

\Paragraph{Treewidth contiguity.}
Given a natural number $\ell\in \Nats^+$, a memory model $\mm$ is $\ell$-treewidth-contiguous (or $\ell$-tw-contiguous) if for every concurrent program $\ConcProg$ and $\rf$-graph $G\in \RFGraphsOf{\ConcProg}{\mm}$ with $\tw(G)\geq \ell$, there is an $\rf$-graph $G'\in \RFGraphsOf{\ConcProg}{\mm}$ with $\tw(G)-\ell \leq \tw(G') < \tw(G)$.
In words, tw-contiguity excludes large treewidth gaps on the $\rf$-graphs of $\ConcProg$ under $\mm$.
The following lemma follows from the facts that
\begin{enumerate*}[label=(\roman*)]
\item  removing a single vertex from a graph decreases its treewidth by at most $1$, and
\item the treewidth of a graph is upper-bounded by its size.
\end{enumerate*}

\smallskip
\begin{restatable}{lemma}{lemtwcontiguousmodels}\label{lem:tw_contiguous_models}
Every $(\seqbef\cup \rf)$-contiguous memory model (and thus all memory models in \cref{tab:main_models} except $\rlxmm$) is also $1$-tw-contiguous.
\end{restatable}

\Paragraph{Exact-treewidth satisfiability.}
Consider a concurrent program $\ConcProg$ over an MSO-definable memory model $\mm$, and some natural number $k\in \Nats^+$.
The question of whether there exists an $\rf$-graph $G\in \RFGraphsOf{\ConcProg}{\mm}$ with $\tw(G)\leq k$ is decidable,
as a special case of \cref{thm:courcelle} where the corresponding MSO formula is simply $\True$.
In order to be able to detect treewidth divergence in $\ConcProg$, we need to be able to decide whether $\RFGraphsOf{\ConcProg}{\mm}$ exhibits a graph with treewidth \emph{exactly} $k$.
We achieve this based on the fact that optimal tree decompositions of graphs can be computed via MSO transductions~\cite{Mikolaj2016,Mikolaj2017}, which can be checked to have width exactly $k$.
Formally, we obtain the following satisfiability lemma as a strengthening of \cref{cor:mso_satisfiability}.

\smallskip
\begin{restatable}{lemma}{lemcourcelletheoremexacttw}\label{lem:courcelle_theorem_exact_tw}
Given an MSO formula $\varphi$ and some $k\in \Nats^+$, the problem of whether there exists an $\rf$-graph $G=\tuple{\Events, \seqbef, \rf}$ with $\tw(G)= k$ and $G\models\varphi$ is decidable.
\end{restatable}

\Paragraph{An algorithm for reachability and $\rf$-robustness.}
We are now ready to establish the algorithmic interplay between reachability and $\rf$-robustness.
Let $\mm$ be an MSO-definable memory model that is also $\ell$-tw-contiguous, for some $\ell\in \Nats^+$.
Consider a concurrent program $\ConcProg$, and let $k=|\ThreadDom| + |\LocationDom|$.
For each $m\in \{1,\dots, \ell\}$, we can decide whether there exists an $\rf$-graph $G\in \RFGraphsOf{\ConcProg}{\mm}$ with $\tw(G)=k+m$, using \cref{lem:free_rf_graphs_definable} and \cref{lem:courcelle_theorem_exact_tw}.
We distinguish the following cases.
\begin{compactitem}
\item If the answer is negative for all such $m$, since $\mm$ is $\ell$-tw-contiguous, we are guaranteed that $\RFGraphsOfTw{\ConcProg}{\mm}{k}=\RFGraphsOf{\ConcProg}{\mm}$, i.e., all $\rf$-graphs of $\ConcProg$ have treewidth $\leq |\ThreadDom| + |\LocationDom|$.
Then, \cref{lem:bounded_tw_reachability} implies that state reachability (and thus also state robustness) is decidable.
\item If the answer is positive for any such $m$, then the program is not $\rf$-robust, since every $\rf$-graph $G\in \RFGraphsOf{\ConcProg}{\scmm}$ is such that $\tw(G)\leq k$, due to \cref{lem:sc-tw}.
\end{compactitem}

\smallskip
\begin{restatable}{lemma}{lemrobustnessvsreachability}\label{lem:robustness_vs_reachability}
For any $\ell$-tw-contiguous and $\mso$-definable memory model $\mm$, there is an algorithm that receives a program $\ConcProg$ as an input and either
\begin{enumerate*}[label=(\roman*)]
\item solves state/MSO reachability for $\ConcProg$ under $\mm$, or
\item reports that $\ConcProg$ is not $\rf$-robust against $\mm$.
\end{enumerate*}
\end{restatable}

In turn, \cref{lem:robustness_vs_reachability}, \cref{lem:tw_contiguous_models} and \cref{thm:mso_definable_mms} yield the following theorem.

\smallskip
\begin{restatable}{theorem}{thmrobustnessvsreachability}\label{thm:robustness_vs_reachability}
For any memory model $\mm\in \{ \ramm, \rmwfree{\sramm}, \wramm \}$, there is an algorithm that receives a program $\ConcProg$ as an input and either
\begin{enumerate*}[label=(\roman*)]
\item solves state/MSO reachability for $\ConcProg$ under $\mm$, or
\item reports that $\ConcProg$ is not $\rf$-robust against $\mm$.
\end{enumerate*}
\end{restatable}

\cref{thm:robustness_vs_reachability} is particularly interesting for $\ramm$, for which verification is undecidable in general, and also achieves better complexity than state-reachability algorithms for the other models, whose complexity is non-primitive recursive~\cite{Lahav2022}.

%% file: observational_rf_robustness.tex
\subsection{Observational Reads-From Robustness}\label{SUBSEC:AUG_RF}

The robustness notions mentioned thus far are sometimes too sensitive.
One example is when a program executes speculative reads, followed by a validation step that determines whether these reads obtained obsolete data that should be ignored~\cite{Boehm2012}.
Such programs are not even state-robust, yet it is desirable to classify them as robust since the reads responsible for the violation are unused.
In order to deal with such spurious violations, execution graph robustness has been relaxed to \emph{observational} robustness, effectively disregarding unused reads from triggering a violation~\cite{Margalit2021}.
Here we show that reads-from robustness also supports this type of observational coarsening.

\Paragraph{Augmented concurrent programs.}
In order to avoid introducing heavy extra notation, we make some simplifying assumptions on the structure of programs.
In particular, each thread $t$ executes an \emph{augmented LTS} $\ProgAug_t$ where the following hold wlog.
\begin{compactenum}
\item $\ProgAug_t$ starts with a sequence of $\wt(t, x, \bot), \wt(t, y,\bot),\dots $ of writes to each shared register, where $\bot$ is a fresh value not written by any other transition.
\item There is an additional event label $\tm(t)$, which signifies that $t$ reached an error state.
\item There is an additional label $\ck(t,x)$, which signifies that $t$ is about to use its most recent read on $x$.
Given a sequence of labels $\sigma=\rd(t,x,v), \alpha_1, \alpha_2, \dots, \alpha_m$, we say that $\rd(t,x,v)$ is \emph{unused} in $\sigma$ if for all $i\in [m]$, if $\alpha_i=\ck(t,x)$, then there exists some $j<i$ such that $\op(\alpha_j)\in \{\rd, \rmw\}$ and $\lloc(\alpha_j)=x$.
Then $\ProgAug_t$ must satisfy the following \emph{invariance under unused reads} property:
for every sequence $\sigma_1 \cdot \rd(t,x,v) \cdot \sigma_2\in \Lang{\ProgAug_t}$, if  $\rd(t,x,v)$ is unused in $\rd(t,x,v) \cdot \sigma_2$, then for all $v'\in \ValueDom$, we have $\sigma_1\cdot  \rd(t,x,v') \cdot \sigma_2\in \Lang{\ProgAug_t}$.
\end{compactenum}
At the source-code level, these three conditions can be satisfied by simple, syntactic transformations.
The first two are rather trivial.
The third condition can be achieved by 
\begin{enumerate*}[label=(\roman*)]
\item directing every load on $x$ to a dedicated local register $l_x$, and
\item executing a check instruction $\ck(t,x)$ right before reading the contents of $l_x$, or 
right after a read-modify-write instruction on $x$ (regardless of whether it succeeds).
\end{enumerate*}
An \emph{augmented concurrent program} is represented as $\ConcProgAug=\tuple{\ProgAug_{t_1},\dots, \ProgAug_{t_{|\ThreadDom|}}}$.

\Paragraph{Augmented graphs.}
We lift the definitions of execution-/$\rf$-graphs to augmented execution-/$\rf$-graphs naturally,
by also considering $\ck(t,x)$ and $\tm(t)$ events.
We say that an augmented $\rf$-graph $G$ \emph{reaches $\tm$} if $\tm(t)\in G.\Events$, for some thread $t$.

Given an augmented graph $G$, we define the set of its unused reads $\UnusedReadDom(G)$ in the natural way: a read $\rd(t,x,v)$ of $G$ is unused if for every check $\ck(t,x)$ such that $\rd(t,x,v)\LTo{\seqbef^+} \ck(t,x)$, there exists a read/RMW event $\event$ with $\lloc(\event)=x$ and $\rd(t,x,v)\LTo{\seqbef^+} \event\LTo{\seqbef^+} \ck(t,x)$.
We say that $G$ and $G'$ are \emph{equivalent}, written $G\simeq G'$, if they are component-wise equal, except for the values and $\rf$-edges of their unused reads.
Formally, there is a bijective function $f\colon G.E \to G'.E$ satisfying the following.
\begin{compactenum}
\item For all $\event_1, \event_2\in G.\Events$, we have $(\event_1, \event_2)\in G.\seqbef$ iff $(f(\event_1),f( \event_2))\in G'.\seqbef$.
\item For all $\event \in G.E\setminus \UnusedReadDom(G)$, we have $\event = f(\event)$ and $f(\event)\in G'.E\setminus \UnusedReadDom(G')$.
\item For all $\event \in \UnusedReadDom(G)$, we have that $\op(f(\event))=\rd$ and $\lloc(\event)=\lloc(f(\event))$.
\item For all $\event_1, \event_2\in G.\Events$ with $\event_2\not \in \UnusedReadDom(G)$, we have $(\event_1,\event_2)\in G.\rf$ iff $(f(\event_1), f(\event_2))\in G'.\rf$.
\end{compactenum}

\Paragraph{Observational $\rf$-robustness.}
An augmented program $\ConcProgAug$ is \emph{observationally $\rf$-robust} against $\mm$ if for all $G \in\RFGraphsOf{\ConcProgAug}{\mm}$, there exists some  $G' \in \RFGraphsOf{\ConcProgAug}{\scmm}$ s.t. $G\simeq G'$. 
Note that $\rf$-robustness implies observational $\rf$-robustness, but the latter notion is strictly coarser (see \cref{fig:rob_example}).
Finally, observational $\rf$-robustness implies robustness against $\tm$ reachability.

\smallskip
\begin{restatable}{lemma}{lemobsrobust}\label{lem:observational_robustness}
Consider a memory model $\mm$ and a (augmented) concurrent program $\ConcProgAug$ that is observationally $\rf$-robust against $\mm$.
If there exists some $G \in \RFGraphsOf{\ConcProgAug}{\mm}$ reaching $\tm$, then there exists some  $G' \in \RFGraphsOf{\ConcProgAug}{\scmm}$ also reaching $\tm$.
\end{restatable}

\Paragraph{Observational $\rf$-robustness vs reachability.}
Given \cref{lem:observational_robustness}, we would like to have a decision procedure for detecting when an augmented program violates observational $\rf$-robustness, analogously to \cref{thm:robustness_vs_reachability} for $\rf$-robustness.
At a first glance, this seems not possible, as observationally $\rf$-robust programs might generate augmented $\rf$-graphs of unbounded treewidth.
Nevertheless, we show that this is possible by defining a reduced variant of $\rf$-graphs that has bounded treewidth for all observationally $\rf$-robust programs.
The full procedure is detailed in \cref{SUBSEC:APP_AUG_RF}.
This allows to arrive at the following theorem.

\smallskip
\begin{restatable}{theorem}{thmobsrobustnessvsreachability}\label{thm:obs_robustness_vs_reachability}
For any memory model $\mm\in \{ \ramm, \rmwfree{\sramm}, \wramm \}$, there is an algorithm that receives an augmented program $\ConcProgAug$ as an input and either
\begin{enumerate*}[label=(\roman*)]
\item solves $\tm$ reachability for $\ConcProgAug$ under $\mm$, or
\item reports that $\ConcProgAug$ is not observationally $\rf$-robust against $\mm$.
\end{enumerate*}
\end{restatable}

%% file: related_work.tex
\section{Related Work}\label{SEC:RELATED_WORK}

Here we discuss related work on the automated analysis of concurrent systems, and connect the results of this paper to existing research directions.

\Paragraph{MSO and treewidth.}
The study of computational models under the lens of treewidth+MSO has been long and fruitful.
In~\cite{Madhusudan2011}, it was shown to yield efficient algorithms for automata with access to various types of unbounded storage (e.g., multiple stacks/queues), for which verification is undecidable in general.
The executions of multi-stack automata have also been parameterized by splitwidth~\cite{Cyriac2012,Cyriac2014}, which is similar to treewidth and also enables efficient decision procedures based on MSO.
In a distributed setting, MSO was recently shown capable to capture most common message-passing communication protocols~\cite{Bollig2021a,DiGiusto2023}, leading to similar decision procedures.

Our work is inspired by aforementioned literature, but our results are novel and based on distinct insights for weak memory.
For example, in contrast to message passing, we show that several weak-memory models are not MSO-definable, via novel connections to fine-grained complexity (\cref{SUBSEC:NON_MSO_DEFINABLE_MM}).
Moreover, the features that allow axiomatization in MSO can be subtle (e.g., recall that RMWs break the MSO definability of $\sramm$ but not $\ramm$), and hint towards a connection to (non) multi-copy atomicity, a notion existing primarily in shared-memory concurrency.

At a technical level, our axiomatization relies on an MSO-definable \emph{partial} modification order, i.e., a subset of $\mo$ that can be inferred directly from $\rf$ and $\po$, and is sufficient to establish consistency (aka ``saturation''). 
Beyond the models of \cref{SUBSEC:MSO_DEFINABLE_MM}, a partial $\mo$ is known to suffice for Causal Memory~\cite{Bouajjani2017}, as well as for certain weak isolation levels~\cite{Biswas19,Moldrup25}.
We remark, however, that saturability does not always imply MSO-definability; in fact, our proof behind \cref{thm:non_mso_definability} is over execution families where saturation is sound and complete.

\Paragraph{Bounded analyses.}
The analysis of shared-memory concurrent programs has been a long and burgeoning topic.
Typical algorithms perform some sort of bounding of the space of program executions, in order to obtain a tractable, or even decidable, version of the analysis.
The most common bounding concerns the length of the executions, enabling stateless model checking algorithms~\cite{Flanagan2005,Abdulla2014,Abdulla2015,Chalupa2018,Abdulla2018,Kokologiannakis21,Bui2021,Agarwal2021}.
Another popular type of bounding is on context switches between threads, and has strong practical~\cite{Qadeer2005,Atig2011,Musuvathi2007,Abdulla2024} and complexity-theoretic~\cite{Esparza2014,Chini2017} implications.
Although verification under a context bound is typically decidable, this is not the case for programs executing under $\ramm$~\cite{Abdulla19}, which was amended by the significantly stronger notion of view bounding.

This paper develops this direction further, by proposing treewidth as another dimension of bounding.
Bounded-treewidth executions are an appealing approximation mechanism, since
\begin{enumerate*}[label=(\roman*)]
\item they are of unbounded size,
\item they capture very weak behaviors (see \cref{SEC:APP_WEAK}), even involving arbitrary causality cycles, which are challenging to all existing methods, and
\item they may be complete in practice, as programs appear to naturally restrict the treewidth of their executions (see \cref{SEC:APP_EXPERIMENTS}).
\end{enumerate*}

Although treewidth and context-bounding are incomparable, executions with bounded view-switching yield  $\rf$-graphs of bounded treewidth.
A crude bound of  $O(|\ThreadDom||\LocationDom| + k|\LocationDom|)$ for an execution with $k$ view switches follows from a construction similar to that of \cref{lem:sc-tw}.

\Paragraph{Robustness.}
Program robustness has been a long-standing concept in weak-memory concurrency, utilized by both researchers and programmers.
The simplest and most common guarantee of robustness has been data-race freedom (DRF)~\cite{Adve1990,Adve1993,Gharachorloo1992}, variants of which hold for $\tsomm$~\cite{Owens2010} and the C/C++ memory model~\cite{Batty2015,Lahav2017}.
Although DRF implies state robustness, it is generally considered too sensitive, in the sense that simple (state-) robust programs contain races.
To this avail, researchers have proposed the coarser notion of execution graph robustness,
and proved its decidability against $\tsomm$~\cite{Bouajjani2011,Bouajjani2013} and $\ramm$~\cite{Lahav2019}.
This direction of ever coarser robustness criteria was recently extended with observational robustness~\cite{Margalit2021} and induced subgraph robustness~\cite{Nagar2024}.

The notion of reads-from robustness (as well as its observational relaxation) introduced in this work is a further coarsening in this direction.
This is the first algorithmically useful robustness notion that reasons purely about program-observable behaviors, disentangling them from lower-level architectural details captured in the modification order.

%% file: conclusion.tex
\section{Conclusion}\label{SEC:CONCLUSION}
We have studied weak-memory concurrency under the lens of treewidth and MSO, giving a systematic account of whether and under what conditions various memory models are MSO-axiomatizable, and characterizing the treewidth of their reads-from semantics.
MSO and Courcelle's theorem are usually the first approach to tackling a number of challenging problems, by revealing  insights to be exploited by more refined methods.
In similar spirit, our work suggests a number of interesting future directions, such as establishing precise complexity bounds for the studied problems (analogously to complexity refinements made for message-passing~\cite{Cyriac2014}), possibly exploring different width parameters of $\rf$-graphs, and practical tooling.
\looseness=-1

%% file: app_tw_mm.tex
\section{Proofs of \cref{SEC:TWMM}}\label{SEC:APP_TWMM}

\input{figures/dec_example}

\input{figures/sc_td}

\lemsctw*
\begin{proof}
Consider an $\rf$-graph $G=\tuple{ \Events, \seqbef, \rf }$ such that $G\models \scmm$. 
Since $G\models \scmm$, there exists an $\mo$ such that the execution graph $G'=\tuple{ \Events, \seqbef, \rf, \mo }\models \scmm$,
thus the corresponding relation $\po \cup \rf \cup \mo \cup \fr$ is acyclic.
Let $\tau=\event_1,\cdots, \event_n$ be any total extension of $\po \cup \rf \cup \mo \cup \fr$.
We construct a tree decomposition $\tuple{\{ \Bag_i\colon i\in I  \}, \Tree=(I, F)}$ of $G$ based on $\tau$, as follows.
The vertex set is $I=\IntSet{n}$, and the edge set is $F=\{(i, i+1)\colon i\in \IntSet{n-1} \}$ (i.e., $\Tree$ in fact forms a path).
For each $i\in \IntSet{n}$, the bag $\Bag_i$ contains the following vertices.
\begin{compactenum}[label=(\roman*)]
\item\label{item:td_sc_1} The event $\event_i$.
\item\label{item:td_sc_2} For each thread $t\in \ThreadDom$, the event $\event_j$ that is the last event of thread $t$ with $j<i$, if such an event exists.
\item\label{item:td_sc_3} For each register $x\in \LocationDom$, the event $\event_j$ that is the last event writing to $x$ with $j<i$, if such an event exists.
\end{compactenum} 
It is clear that $|\Bag_i|\leq |\ThreadDom|+|\LocationDom|+1$, yielding the desired bound on the treewidth.
It remains to argue that $\tuple{\{ \Bag_i\colon i\in I  \}, \Tree=(I, F)}$ is indeed a tree decomposition of $G$.
\begin{compactenum}
\item \cref{item:td_sc_1} implies that $\bigcup_{i \in I} \Bag_i = \Events$.
\item For any edge $(\event_j, \event_i)\in \seqbef$, we have $\event_j, \event_i\in \Bag_i$, due to \cref{item:td_sc_2}.
Moreover, for any edge $(\event_j, \event_i)\in \rf$, we have $\event_j, \event_i\in \Bag_i$, due to \cref{item:td_sc_3},
since $\event_j$ must be the most recent write/RMW event on $\lloc(\event_j)$.
\item Finally, every event $\event_i$ occupies a contiguous interval, starting in bag $\Bag_i$ and ending in $\Bag_j$, with $j$ being the earliest index in which $\event_i$ is no longer the last event of $\tid(\event_i)$ nor the the last event writing to $\lloc(\event_i)$.
\end{compactenum}
The desired result follows.
\end{proof}

\smallskip
\begin{restatable}{lemma}{lemsctwlowerbound}\label{lem:sc_tw_lowerbound}
For all $\ThreadDom$ and $\LocationDom$ ($\ThreadDom, \LocationDom \geq 1$), there exists an $\rf$-graph $G$ such that $G\models \scmm$ and $\tw(G)\geq |\ThreadDom|+|\LocationDom|-1$.
\end{restatable}
\begin{proof}
We construct such a graph $G$ that has a clique of size $|\ThreadDom|+|\LocationDom|$ as a minor, using the fact that treewidth is preserved under taking minors~\cite{Robertson1986a}. 
We label the threads $t_1, \cdots, t_{|\ThreadDom|}$ and the locations $x_1, \dots, x_{|\LocationDom|}$.
The clique minor of $G$ is formed by vertices $v_{t_1}, \dots, v_{t_{|\ThreadDom|}}$ and $v_{x_1}, \cdots, v_{x_{|\LocationDom|}}$.
In turn, each vertex is obtained by contracting various events in $G$.

We list the events of $G$ in a linear order $\sigma=\sigma_1\cdot \sigma_2\cdot \sigma_3$, which implicitly specifies $G.\seqbef$ and $G.\rf$ according to $\scmm$ semantics (hence $G\models \scmm$).
We use a superscript $\event^i$ to refer to events of $\sigma_i$.
At a high level,
\begin{enumerate*}[label=(\roman*)]
\item $\sigma_1$ encodes a clique amongst $v_{x_1}, \dots, v_{x_{|\LocationDom|}}$,
\item $\sigma_2$ encodes a complete bipartite graph where the two partites consist of $v_{x_1}, \dots, v_{x_{|\LocationDom|}}$ and $v_{t_1}, \dots, v_{t_{|\ThreadDom|}}$, and
\item $\sigma_3$ encodes a clique amongst $v_{t_1}, \dots, v_{t_{|\ThreadDom|}}$.
\end{enumerate*}
Note that these three subgraphs already have a clique over $v_{t_1}, \cdots, v_{t_{|\ThreadDom|}}$ and $v_{x_1}, \cdots, v_{x_{|\LocationDom|}}$ as a minor, as desired.

\SubParagraph{The sequence $\sigma_1$.}
We construct the first sequence as
$
\sigma_1=\sigma_1^0\cdot \sigma_1^1\cdots \sigma_1^{|\LocationDom|-1}
$
where
$\sigma_1^0=\wt^1(t_1, x_1), \wt^1(t_1, x_2), \dots, \wt^1(t_1, x_{|\LocationDom|})$,
and for each $i\in \IntSet{|\LocationDom|-1}$, we have
\begin{align*}
\sigma_1^i=&\rd^1(t_1, x_{i}), \wt^1(t_1, x_{i}), \rd^1(t_1, x_{i+1}), \wt^1(t_1, x_{i+1}), \\
&\rd^1(t_1, x_{i}), \wt^1(t_1, x_{i}), \rd^1(t_1, x_{i+2}), \wt^1(t_1, x_{i+2}),\\
&\vdots\\
&\rd^1(t_1, x_{i}), \wt^1(t_1, x_{i}), \rd^1(t_1, x_{|\LocationDom|}), \wt^1(t_1, x_{|\LocationDom|})\ .    
\end{align*}
Note that all events of $\sigma_1$ belong to the same thread, hence pairwise successive events in $\sigma_1$ are connected by $\seqbef$.
This also induces an $\rf$ relation between each read from some location $x_i$ and the previous write to $x_i$.
In order to obtain the first clique minor, we contract all $\seqbef$- and $\rf$- edges between events of $\sigma_1$ in the same location $x_i$, to obtain the sought vertex $v_{x_i}$.
Note that this is possible because all such events form a sequence of the form
\[
\wt^1(t_1,x_i)\LTo{\rf}\rd^1(t_1, x_i)\LTo{\seqbef}\wt^1(t_1, x_i)\LTo{\rf}\cdots \LTo{\seqbef}\wt^1(t_1, x_i)\ .
\]
In order to see that the contracted vertices form a clique, note that for each $i, j\in \IntSet{|\LocationDom|}$ with $i<j$, the subsequence $\sigma_1^i$ connects some $\wt^1(t_1, x_i)$ to the unique $\rd^1(t_1, x_j)$ of $\sigma_1^i$ via $\seqbef$.

\SubParagraph{The sequence $\sigma_2$.}
We construct the second sequence as $
\sigma_2= \sigma_2^1\cdot \sigma_2^2\cdots \sigma_2^{|\ThreadDom|}
$
where
\[\sigma_2^i=\rd^2(t_i, x_1), \rd^2(t_i, x_2), \dots, \rd^2(t_i, x_{|\LocationDom|})\ .\]
For each $i\in [|\ThreadDom|]$ and $j\in [|\LocationDom|]$,  $\rd^2(t_i, x_j)$ will read from the latest event $\wt^1(t_1, x_j)$ in $\sigma_1$.

Vertices in the first side of the sought partition, $v_{x_1}, \dots, v_{x_{|\LocationDom|}}$, have been defined in $\sigma_1$. $\forall i \in [|\ThreadDom|]$, the vertex $v_{t_i}$ will be obtained by contracting $\seqbef$-edges between events in $\sigma_2$ (and, later, $\sigma_3$) that belong to the same thread $t_i$. In order to see that this forms the desired complete bipartition, notice that $\forall i\in [|\ThreadDom|], j\in [|\LocationDom|]$, $\wt^1(t_1, x_j)\LTo{\rf}\rd^2(t_i, x_j)$. 

\SubParagraph{The sequence $\sigma_3$.}Finally, we construct the third sequence as $
\sigma_3= \sigma_3^1\cdot \sigma_3^2\cdots \sigma_3^{|\ThreadDom|}
$
where
\[\sigma_3^i=\wt^3(t_i, x_1), \rd^3(t_{i+1}, x_1), \dots, \rd^3(t_{|\ThreadDom|}, x_1)\ .\]
For every $\sigma_3^i$, each $\rd^3(t_j, x_1)$ reads from the unique $\wt^3(t_i, x_1)$.

As previously mentioned, vertices $v_{t_i}, \forall i \in [|\ThreadDom|]$ will be obtained by contracting $\seqbef$-edges between events in $\sigma_2$ and $\sigma_3$ that belong to the same thread $t_i$. In order to see that the contracted vertices form a clique, note that for each $i, j\in \IntSet{|\ThreadDom|}$ with $i<j$, the subsequence $\sigma_3^i$ connects $\wt^3(t_i, x_1)$ to $\rd^3(t_j, x_1)$ via $\rf$.

Notice that since the linear order of events $\sigma$ follows $\seqbef$ and each read reads from the latest write, this $\rf$-graph is consistent in $\scmm$.
    
\end{proof}

\smallskip
\lemtsotw*
\begin{proof}

Given some $n\in \Nats^+$, the grid graph $\grid{n}$ consists of $n^2$ vertices $\{ \tuple{i,j}\colon i,j\in\IntSet{n} \}$ and edges $\{ (\tuple{i_1, j_1}, \tuple{i_2, j_2})\colon (i_1=i_2 \text{ and } |j_2-j_1|=1) \text{ or } (j_1=j_2 \text{ and } |i_2-i_1|=1)\}$.
It is known that $\tw(\grid{n})=n$~\cite{Robertson1986a}, hence grids have unbounded treewidth (i.e. not bounded by a fixed constant independent of $n$).

We prove the lemma by constructing an $\rf$-graph $G_n$ that has $\grid{n}$ as a minor, since then $\tw(G_n)\geq \tw(\grid{n}) =n$,
following the fact that treewidth bounds are closed under minors~\cite{Robertson1986a}.
Let $\List_n$ be a sequence of the vertices of $\grid{n}$ ordered lexicographically.
For example $\List_3=\tuple{1,1}, \tuple{1,2}, \tuple{1,3}, \tuple{2,1}, \tuple{2,2}, \tuple{2,3}, \tuple{3,1}, \tuple{3,2}, \tuple{3,3}$.
We construct $G$ with two threads $t_1$, $t_2$ and $2$ registers $x$, $y$.
Intuitively, $t_1$ contains all the grid vertices and its horizontal edges, while $t_2$ encodes the vertical edges.

Thread $t_1$ consists of $n^2$ segments $s^1_{\tuple{i,j}}$, one for each vertex $\tuple{i,j}$ of $\grid{n}$.
The segments are $\seqbef$-ordered following $\List_n$.
Each segment $s^1_{\tuple{i,j}}$ is either a single event or two events ordered by $\seqbef$.
In particular
\begin{compactitem}
\item If $i=1$, we have a single event $\wt_{\tuple{i,j}}(x)$.
\item If $i=n$, we have a single event $\rd_{\tuple{i,j}}(y)$.
\item In all other cases, we have two events with the corresponding ordering $\rd_{\tuple{i,j}}(y) \LTo{\seqbef} \wt_{\tuple{i,j}}(x)$.
\end{compactitem}
Thread $t_2$ consists of $n(n-1)$ segments $s^2_{\tuple{i,j}}$, one for each vertex $\tuple{i,j}$ of $\grid{n}$ with $i<n$.
The segments are $\seqbef$-ordered following $\List_n$.
Each segment $s_{\tuple{i,j}}$ consists of two events with the corresponding ordering $\rd_{\tuple{i,j}}(x) \LTo{\seqbef} \wt_{\tuple{i+1,j}}(y)$.
The reads from relation is such that $(\wt_{\tuple{i,j}}(z), \rd_{\tuple{i,j}}(z))\in \rf$ for each $z\in \{x, y\}$.
See \cref{fig:grid_4} for an illustration of $G_4$.

We first argue that $G_n\models \tsomm$.
Note that there can be only one possible witnessing $\mo$, satisfying $\mo\subseteq \po$, since every register is exclusively written by a single thread.
This implies that
\begin{enumerate*}[label=(\roman*)]
\item $\moe=\emptyset$, and
\item $\fr=\fre$.
\end{enumerate*}
Moreover, observe that all $\ppo$ edges are monotone:~if $\event_{\tuple{i_1, j_1}}\LTo{\ppo}\event'_{\tuple{i_2, j_2}}$, then $\tuple{i_2, j_2}$ does not appear before $\tuple{i_1, j_1}$ in $\List_n$.
Since every segment contains at most one write, this makes $\fr$ edges strictly monotone.
Since there are clearly no $(\ppo\cup \rfe)$ cycles, this also excludes the presence of $(\ppo\cup \rfe \cup \moe \cup \fre)$ cycles, guaranteeing \ref{eq:sce}.
Finally, to verify that \ref{eq:scloc} also holds, observe that every register $z\in \{x,y\}$ is exclusively written or read by one of the threads, which makes $\po_z\subseteq \ppo$.
Since $(\ppo\cup \rfe \cup \moe \cup \fre)$ is acyclic, we also have that $(\po_z\cup \rf_z\cup \mo_z \cup \fr_z)$ is acyclic.

We now show that $\grid{n}$ is a minor of $G_n$, as $\tw(\grid{n})=n$~\cite{Robertson1986a}.
We obtain $\grid{n}$ from $G_n$ via edge contractions. 
To do so, we contract all edges of the form $e_{\tuple{i,j}} \to e'_{\tuple{i,j}}$, i.e., edges between events that refer to the same vertex $\tuple{i,j}$ of $\grid{n}$. 
The new vertex created after applying all contractions for $\tuple{i,j}$ is the vertex $\tuple{i,j}$ of the minor $\grid{n}$.
Let $H_n$ be the contracted graph, and we argue that $H_n=\grid{n}$, as it contains all the edges of $\grid{n}$.

First, all horizontal edges of $\grid{n}$, of the form $(\tuple{i_1, j_1}, \tuple{i_2, j_2})$ with $(i_1=i_2 \text{ and } |j_2-j_1|=1)$ are present in $H_n$, as these are adjacent in $\List_n$, which defines the $\seqbef$-order of events in $t_1$.  The vertical edges of $\grid{n}$, of the form $(\tuple{i_1, j_1}, \tuple{i_2, j_2})$ with $(|i_2-i_1|=1 \text{ and } j_1=j_2)$ are also present in $H_n$, through the edges $\rd_{\tuple{i,j}}(x) \LTo{\seqbef} \wt_{\tuple{i+1,j}}(y)$ in $t_2$. 

\end{proof}

%% file: figures/dec_example.tex
\begin{figure}
\begin{subfigure}[t]{.32\textwidth}

    \begin{tikzpicture}[ line width=1pt, node distance = 0.8cm and 0.4cm] 

            \node (1) {\textcolor{colorVARa}{\footnotesize$\rd_1(x)$}};
            \node[below of = 1] (2){\textcolor{colorVARb}{\footnotesize$\wt_2(y)$}};

            \node[right = of 1] (3) {\textcolor{colorVARa}{\footnotesize$\wt_3(x)$}};
            \node[below of = 3] (4){\textcolor{colorVARb}{\footnotesize$\rd_4(y)$}};
            \node[below of = 4] (5){\textcolor{colorVARa}{\footnotesize$\wt_5(y)$}};

            \node[right = of 3] (6) {\textcolor{colorVARa}{\footnotesize$\rd_6(x)$}};
            \node[below of = 6] (7){\textcolor{colorVARa}{\footnotesize$\rmw_7(y)$}};
            \node[below of = 7] (8){\textcolor{colorVARa}{\footnotesize$\rd_8(y)$}};

            \draw [seqbef] (1) to (2);
            \draw [seqbef] (3) to (4);
            \draw [seqbef] (4) to (5);
            \draw [seqbef] (6) to (7);
            \draw [seqbef] (7) to (8);
            \draw [rf] (3) to (1);
            \draw [rf] (2) to (4);
            \draw [rf] (3) to (6);
            \draw [rf] (3) to (7);
            \draw [rf] (7) to (5);
            \draw [rf] (5) to (8);
        
        \end{tikzpicture}
\caption{$G$.}
\label{SUBFIG:DECEXAMPLE1}
\end{subfigure}%
\begin{subfigure}[t]{.23\textwidth}

    \begin{tikzpicture}[ line width=1pt, node distance = 0.7cm,inner sep=5pt] 
            \node[rectangle, draw=black, rounded corners] (t12) at (5,0){\footnotesize$\textcolor{colorVARa}{\rd_1},\textcolor{colorVARb}{\wt_2},\textcolor{colorVARa}{\wt_3},\textcolor{colorVARb}{\rd_4}$};
            \node[rectangle, rounded corners, draw=black, below of =t12] (t13){\footnotesize$\textcolor{colorVARa}{\wt_3},\textcolor{colorVARb}{\rd_4},\textcolor{colorVARa}{\rd_6},\textcolor{colorVARa}{\rmw_7}$};
            \node[rectangle, rounded corners, draw=black, below of =t13] (t14){\footnotesize$\textcolor{colorVARb}{\rd_4},\textcolor{colorVARa}{\wt_5},\textcolor{colorVARa}{\rmw_7},\textcolor{colorVARa}{\rd_8}$};

            \draw [] (t12) tonode [below=2pt]{} (t13);
            \draw [] (t13) tonode [below=2pt]{} (t14);
        
        \end{tikzpicture}
\caption{Width $=3$.}
\label{SUBFIG:DECEXAMPLE2}
\end{subfigure}%
\begin{subfigure}[t]{.45\textwidth}

\begin{tikzpicture}[ line width=1pt, node distance = 0.7cm,inner sep=5pt]

            \node[rectangle, rounded corners, draw=black] (t12) at (5,0){\footnotesize$\textcolor{colorVARa}{\rd_1},\textcolor{colorVARb}{\wt_2},\textcolor{colorVARa}{\wt_3}$};
            \node[rectangle, rounded corners, draw=black, right =0.1 cm of t12] (t13){\footnotesize$\textcolor{colorVARb}{\wt_2},\textcolor{colorVARa}{\wt_3},\textcolor{colorVARb}{\rd_4}$};
            \node[rectangle, rounded corners, draw=black, below of =t13] (t14){\footnotesize$\textcolor{colorVARa}{\wt_3},\textcolor{colorVARb}{\rd_4},\textcolor{colorVARa}{\rmw_7}$};
            \node[rectangle, rounded corners, draw=black,  left = 0.1cm of t14] (t15){\footnotesize$\textcolor{colorVARa}{\wt_3},\textcolor{colorVARa}{\rd_6},\textcolor{colorVARa}{\rmw_7}$};
            \node[rectangle, rounded corners, draw=black,  right = 0.1cm of t14] (t16){\footnotesize$\textcolor{colorVARb}{\rd_4},\textcolor{colorVARa}{\rmw_7},\textcolor{colorVARa}{\wt_5}$};
            \node[rectangle, rounded corners, draw=black, below of = t16] (t17){\footnotesize$\textcolor{colorVARa}{\wt_5},\textcolor{colorVARa}{\rmw_7},\textcolor{colorVARa}{\rd_8}$};

            \draw [] (t12) tonode [below=2pt]{} (t13);
            \draw [] (t13) tonode [below=2pt]{} (t14);
            \draw [] (t14) tonode [below=2pt]{} (t15);
            \draw [] (t14) tonode [below=2pt]{} (t16);
            \draw [] (t16) tonode [below=2pt]{} (t17);
        
        \end{tikzpicture}
\caption{Width $=2$.}
\label{SUBFIG:DECEXAMPLE3}
\end{subfigure}%
\caption{
An $\rf$-graph $G$ (\subref{SUBFIG:DECEXAMPLE1}), together with a sub-optimal (\subref{SUBFIG:DECEXAMPLE2}) and an optimal (\subref{SUBFIG:DECEXAMPLE3}) tree decomposition of $G$.
}
\label{FIG:DECEXAMPLE}
\end{figure}
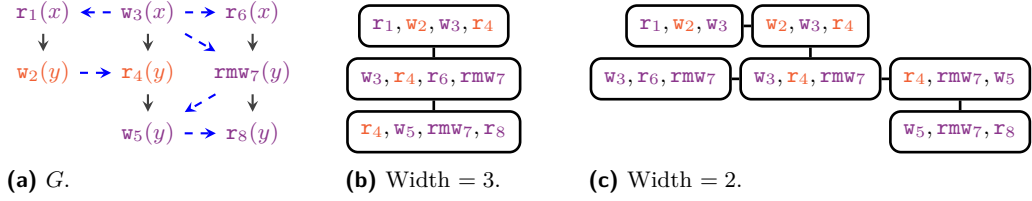

%% file: figures/sc_td.tex
\begin{figure}
\begin{subfigure}[b]{.35\textwidth}

\begin{tikzpicture}[ line width=1pt, node distance = 0.7cm] 

            \node (wy1) {\textcolor{colorVARb}{\footnotesize$\wt_1(y)$}};
            \node[below of = wy1, inner sep=2pt] (rx0){\textcolor{colorVARa}{\footnotesize$\rd_2(x)$}};

            \node[right = 0.4cm of wy1, inner sep=2pt] (wx1) {\textcolor{colorVARa}{\footnotesize$\wt_4(x)$}};
            \node[above of = wx1, inner sep=2pt] (wx0) {\textcolor{colorVARa}{\footnotesize$\wt_3(x)$}};
            \node[below of = wx1, inner sep=2pt] (ry0) {\textcolor{colorVARb}{\footnotesize$\rd_5(y)$}};

            \draw [seqbef] (wx0) to (wx1);
            \draw [seqbef] (wy1) to (rx0);
            \draw [seqbef] (wx1) to (ry0);
            \draw [rf] (wx1) to (rx0);
            \draw [rf] (wy1) to (ry0);
\end{tikzpicture}
\caption{An  $\rf$-graph $G$ consistent in $\scmm$.}
\label{subfig:sc_rf}
\end{subfigure}%
\hfill
\begin{subfigure}[b]{.6\textwidth}

\begin{tikzpicture}[ line width=1pt, node distance = 0.8cm]

    \node[anchor=west] at (-1.2,0.7){\footnotesize$\tau = \ \textcolor{colorVARa}{\wt_3}, \ \ \ \ \ \ \ \textcolor{colorVARa}{\wt_4}, \ \ \ \ \ \ \ \ \ \textcolor{colorVARb}{\wt_1}, \ \ \ \ \ \ \ \ \ \textcolor{colorVARa}{\rd_2}, \ \ \ \ \ \ \ \ \ \ \textcolor{colorVARb}{\rd_5} $};

    \node[draw, align=left, rounded corners] at (0,0) (a3){\footnotesize$\textcolor{colorVARa}{\wt_3}$\textcolor{white}{
    $\wt_3$}\\\textcolor{white}{a}\textcolor{white}{a}};
    
    \node[draw, align=left, rounded corners, right = 0.3cm of a3] (a4){\footnotesize$\textcolor{colorVARa}{\wt_3}$, $\textcolor{colorVARa}{\wt_4}$\\\textcolor{white}{a}\textcolor{white}{a}};
    
    \node[draw, align=left, rounded corners, right = 0.3cm of a4] (a5){\footnotesize$\textcolor{colorVARb}{\wt_1}$, $\textcolor{colorVARa}{\wt_4}$\\\textcolor{white}{a}\textcolor{white}{a}};
    
    \node[draw, align=left, rounded corners, right = 0.3cm of a5] (a6){\footnotesize$\textcolor{colorVARb}{\wt_1}$, $\textcolor{colorVARa}{\rd_2}$,\\ $\textcolor{colorVARa}{\wt_4}$\textcolor{white}{a}};

    \node[draw, align=left, rounded corners, right = 0.3cm of a6] (a7){\footnotesize$\textcolor{colorVARb}{\wt_1}$, $\textcolor{colorVARa}{\rd_2}$,\\ $\textcolor{colorVARa}{\wt_4}$, $\textcolor{colorVARb}{\rd_5}$};

    \draw (a3) to (a4);
    \draw (a4) to (a5);
    \draw (a5) to (a6);
    \draw (a6) to (a7);

\end{tikzpicture}
\caption{A tree decomposition of (\subref{subfig:sc_rf}).}
\label{subfig:sc_tw}
\end{subfigure}%
\caption{
An instance of the tree decomposition behind \cref{lem:sc-tw}.
}
\label{fig:sc_tw}
\end{figure}
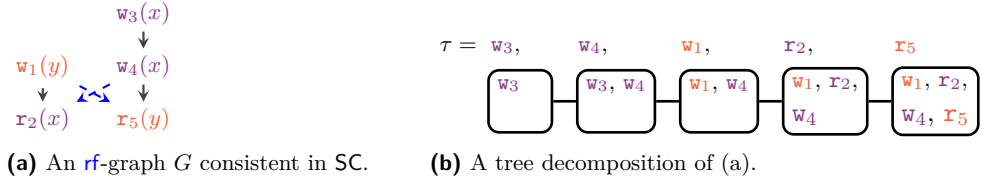

%% file: app_weak.tex
\section{Bounded Treewidth and Weak Behaviours}\label{SEC:APP_WEAK}

We first show that high-treewidth $\rf$-graphs must exhibit weak behaviors. 

\smallskip
\lemhightwweak*
\begin{proof}
    We prove the contrapositive: if there exists a total extension $\sigma$ of $(G.\seqbef\cup G.\rf)^+$ such that \emph{every} pair $\event_1\LTo{\rf}\event_2$ has fewer than $k$ writes/RMWs on their same location between $\event_1$ and $\event_2$ in $\sigma$, then $\tw(G) < |\ThreadDom| + k \cdot |\LocationDom|$.

    We construct a tree decomposition $\tuple{\{ \Bag_i\colon i\in I  \}, \Tree=(I, F)}$ of $G$ based on $\sigma$, in a similar manner to the construction of \cref{lem:sc-tw}.
    
    The vertex set is $I=\IntSet{n}$, and the edge set is $F=\{(i, i+1)\colon i\in \IntSet{n-1} \}$.
    For each $i\in \IntSet{n}$, the bag $\Bag_i$ contains.
    \begin{enumerate*}[label=(\roman*)]
    \item\label{item:td_sc_1} The event $\event_i$,
    \item\label{item:td_sc_2} for each thread $t\in \ThreadDom$, the event $\event_j$ that is the last event of thread $t$ with $j<i$, if such an event exists, and
    \item\label{item:td_sc_3} for each register $x\in \LocationDom$, the (at most) $k$ most recent events among $\{\event_j \mid j<i\}$ that write to $x$.
    \end{enumerate*} 
    
It is clear that $|\Bag_i|\leq |\ThreadDom|+k \cdot |\LocationDom|+1$, yielding the desired bound on the treewidth.
It is also the case that $\tuple{\{ \Bag_i\colon i\in I  \}, \Tree=(I, F)}$ is a tree decomposition of $G$, since
 \begin{enumerate*}[label=(\arabic*)]
\item $\bigcup_{i \in I} \Bag_i = \Events$ due to  \cref{item:td_sc_1},
\item for any edge $(\event_j, \event_i)\in \seqbef$, we have $\event_j, \event_i\in \Bag_i$ due to \cref{item:td_sc_2},
\item for any edge $(\event_j, \event_i)\in \rf$, we have $\event_j, \event_i\in \Bag_i$ due to \cref{item:td_sc_3} (since $\event_j$ must be one of the $k$ most recent write/RMW events on $\lloc(\event_j)$), and
\item every event $\event_i$ occupies a contiguous interval.
\end{enumerate*} 
\end{proof}
We now show examples of executions with very weak behavior that, nevertheless, enjoy small treewidth.

\Paragraph{Store Buffer.}
Consider (a scaled version of) the standard Store Buffer litmus test (\cref{SUBFIG:WEAKTW1}).
Observe that any total extension of $(G.\seqbef \cup G.\rf)^+$ separates $n$ reads from their writer by $n-1$ intermediate writes on the same register.
Nevertheless, this graph has treewidth only $2$.

\Paragraph{Extended Store Buffer.}
Next, consider an Extended Store Buffer example (\cref{SUBFIG:WEAKTW2}), where now every read observes a different write.
In every total extension of $(G.\seqbef\cup G.\rf)^+$, there are $n-1$ stale writes being read (as opposed to just $1$ in the previous example).
Nevertheless, this graph has treewidth only $3$, with the corresponding tree decomposition shown in \cref{FIG:PATHDECESB}.

\Paragraph{Extended Load Buffer.}
Finally, consider an Extended Load Buffer example (\cref{SUBFIG:WEAKTW3}).
Such behaviors are generally considered very weak, because now $(G.\seqbef\cup G.\rf)$ is not acyclic, and in fact, every causality cycle has length $\geq n$.
Nevertheless, this graph also has treewidth only $3$.
This illustrates another important aspect of treewidth:~the definition of tree decompositions is not sensitive to the directionality of edges, and thus the graphs in \cref{SUBFIG:WEAKTW2} and \cref{SUBFIG:WEAKTW3} have, save for event relabeling, the same set of tree decompositions (and thus the same treewidth).

\Paragraph{Implications.}
We show in \cref{SUBSEC:BOUNDED_TW_REACHABILITY} that, under certain memory models, the set of behaviors of a program as witnessed by $\rf$-graphs of bounded treewidth can be explored systematically and efficiently.
The above examples indicate that this set contains non-trivial behaviors, that have traditionally been a challenge to both model checking and verification based on program logics.

%% file: app_experiments.tex
\section{Experiments}\label{SEC:APP_EXPERIMENTS}

\Paragraph{Setup.}
In order to test the hypothesis that executions of programs under weak memory tend to have low treewidth, we use the GenMC model checker~\cite{Kokologiannakis21} to explore all (bounded) executions of a large class of programs and measure their treewidth.
Our benchmark set consists of all programs accompanying GenMC under the ``correct'' category.
We use the options \texttt{-disable-ipr} and \texttt{-disable-race-detection} to allow each benchmark to execute fully, and we set the loop-unroll bound to \texttt{unroll=200}.
In order to obtain indications of programs which have weak behaviors, we run each benchmark both under $\scmm$ and $\ramm$.
In the latter case, GenMC uses the memory order of each memory access, and thus might mix accesses of different types. 
Out of $399$ benchmarks originally, we exclude $27$ for reasons such as not compiling because of dependence on external libraries and not producing $\scmm$ executions (e.g., due to being blocked by the verifier).
We run each of the remaining $372$ benchmarks until it either exceeds $10$ minutes, or produces 10K execution graphs.
This executes $342$ benchmarks to completion, while for the remaining $30$, we keep the execution graphs produced by GenMC thus far.
We use FlowCutter~\cite{flowcutter} to calculate the treewidth of each $\rf$-graph. 
FlowCutter performs iterative over-approximations of treewidth until it exhausts a user-defined time bound ($2$s in our case).
Hence, the reported treewidth values are not tight, but still informative when they are sufficiently small to support our hypothesis.

The full commands used to generate the $\scmm$ and weak-memory executions, are, respectively

\begin{compactitem}
    \item \texttt{./genmc/RelWithDebInfo/genmc --print-exec-graphs --v0}
    
    \texttt{--disable-race-detection --disable-ipr --disable-estimation --unroll=200 -sc}
    
    \texttt{input\_file\_path}
    \item \texttt{./genmc/RelWithDebInfo/genmc --print-exec-graphs --v0 }
    
    \texttt{--disable-race-detection --disable-ipr --disable-estimation --unroll=200 -ra}
    
    \texttt{input\_file\_path}
\end{compactitem}

\input{figures/ex_graphs}

\begin{figure}
\begin{subfigure}[b]{0.49\textwidth}
\centering
\includegraphics[scale=0.37,trim= 0 0 0 0,clip]{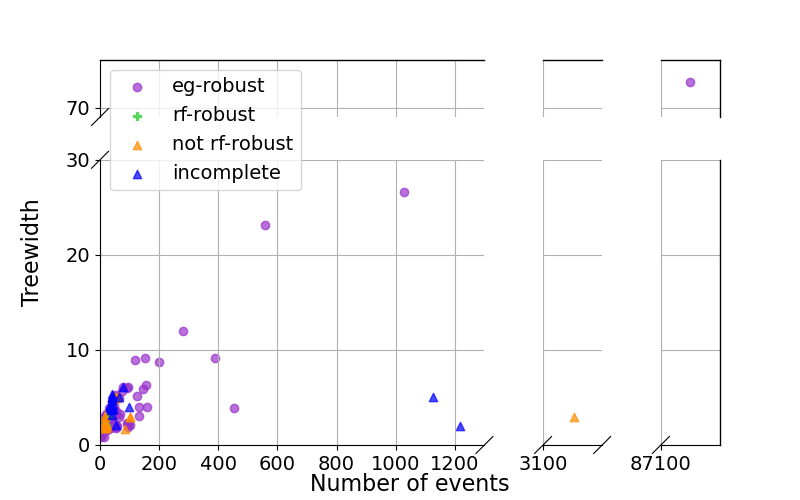}
\subcaption{\label{subfig:graphs_twn}
Overall dataset.
}
\end{subfigure}
\hfill
\begin{subfigure}[b]{0.49\textwidth}
\centering
\includegraphics[scale=0.37,trim= 0 0 0 0,clip]{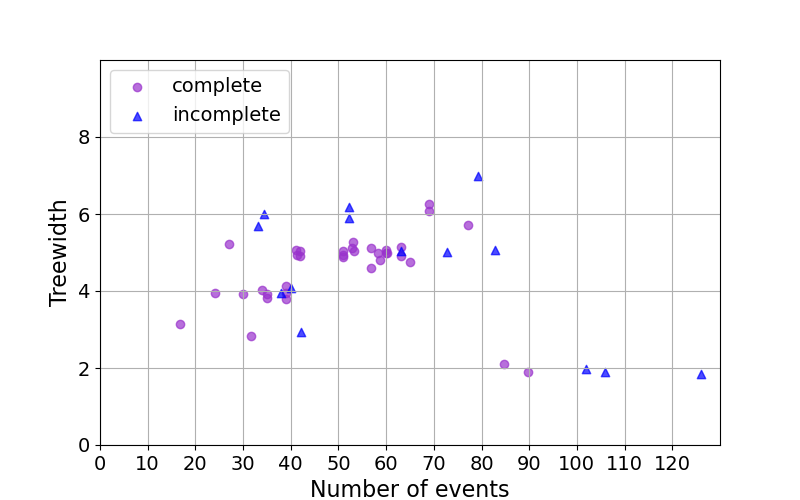}
\subcaption{\label{subfig:graphs_twnxp}
Extended benchmarks with weak behavior.
}
\end{subfigure}
\caption{\label{fig:graphs_n}
Experimental results on the treewidth of our dataset versus execution graph size.
Every benchmark program contributes a data point $(x,y)$, where $y$ is the maximum treewidth among all its $\rf$-graphs explored under weak-memory, and $x$ is the number of events of an execution graph with maximal treewidth size.
The labels ``eg-robust''/``rf-robust'' denote that, between the $\scmm$ and weak-memory runs, the respective benchmark exhibits the same number of execution-/rf- graphs, while ``not rf-robust'' means that there are more $\rf$-graphs under weak memory, and ``incomplete'' means that the benchmark did not execute to completion due to limits on time and output size.
}
\end{figure}

\Paragraph{Results.}
\cref{subfig:graphs_twmax} shows the obtained upper bounds on the treewidth of each benchmark.
\cref{subfig:graphs_twn} shows the relationship between maximum observed treewidth and execution-graph size.
Treewidth tends to stay small and well within the treewidth bound under $\scmm$ of \cref{lem:sc-tw}, even for benchmarks that exhibit weak behavior.
Out of the benchmarks that are executed to completion, $303$ are $\rf$-robust and $282$ are execution-graph robust. 
Moreover, among $39$ benchmarks that are not $\rf$-robust, $36$ exhibit the same maximum treewidth under $\scmm$ and weak memory, while for the remaining $3$, their weak executions saw a treewidth increase by at most $2$.
These findings suggest that verification based on treewidth-based under-approximation may be fairly precise.

\Paragraph{Extensions to benchmarks with weak behavior.}
In order to further stress-test our above conclusions, we select the ``not $\rf$-robust'' benchmarks that have scaling parameters, and scale them further to see if they show a pattern of treewidth divergence.
This results in 10 distinct benchmarks (we exclude simple variants of the same benchmark), and for each, we generate 5 variations by scaling different parameters.
The benchmarks selected for further stress testing are \texttt{buf-ring}, \texttt{chase-lev}, \texttt{dekker2-fixed}, \texttt{dq}, \texttt{dq-opt}, \texttt{eratosthenes}, \texttt{mutex}, \texttt{qu}, \texttt{qu-opt} and \texttt{stc-opt}.
The results are shown in \cref{subfig:graphs_twmaxxp}, confirming that treewidth remains well-behaved.
\cref{subfig:graphs_twnxp} shows the relationship between maximum observed treewidth and execution-graph size.
Out of the $36$ benchmarks executed to completion, only $6$ benchmarks exhibit a higher treewidth under weak memory (by at most $2$), $5$ of which are variations of \texttt{dekker-2}.

%% file: figures/ex_graphs.tex
\begin{figure}
\begin{subfigure}[b]{0.49\textwidth}
\centering
\includegraphics[scale=0.4,trim= 0 0 0 35,clip]{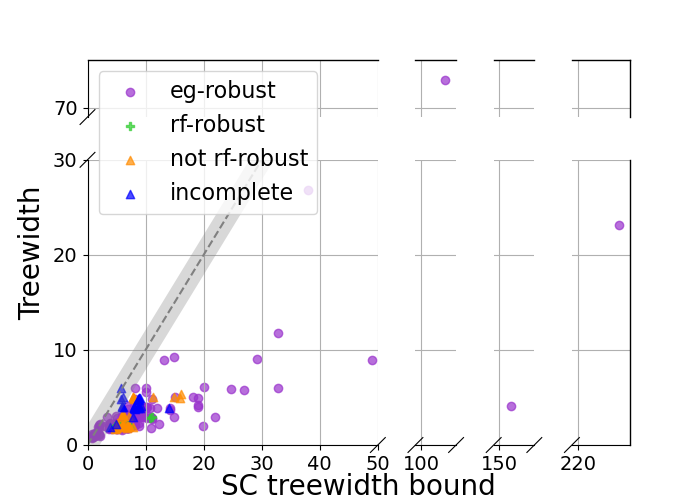}
\subcaption{\label{subfig:graphs_twmax}
Overall dataset.
}
\end{subfigure}
\hfill
\begin{subfigure}[b]{0.49\textwidth}
\centering
\includegraphics[scale=0.4,trim= 0 0 0 35,clip]{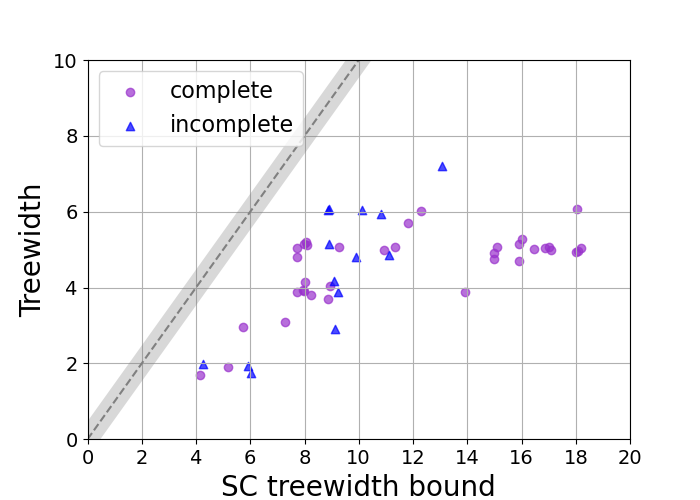}
\subcaption{\label{subfig:graphs_twmaxxp}
Extended benchmarks with weak behavior.
}
\end{subfigure}
\caption{\label{fig:graphs}
Experimental results on the treewidth of our dataset.
Every benchmark program contributes a data point $(x,y)$, where $y$ is the maximum treewidth among all its $\rf$-graphs explored under weak-memory, and $x$ is the theoretical upper bound of its treewidth based on \cref{lem:sc-tw}.
The labels ``eg-robust''/``rf-robust'' denote that, between the $\scmm$ and weak-memory runs, the respective benchmark exhibits the same number of execution-/rf- graphs, while ``not rf-robust'' means that there are more $\rf$-graphs under weak memory, and ``incomplete'' means that the benchmark did not execute to completion due to limits on time and output size.
}
\end{figure}

%% file: app_mso.tex
\section{Proofs of \cref{SEC:MSO}}\label{SEC:APP_MSO}

\subsection{Proofs of \cref{SUBSEC:MSO}}

\smallskip
\lemrfgraphsmsodefinable*
\begin{proof}
We provide two MSO formulas $\msosb$ and $\msorf$ that express the well-formedness conditions on $\seqbef$ and $\rf$, respectively, so that $\tuple{\Events,\seqbef,\rf}$ defines an $\rf$-graph.
\begin{align*}
\msosb \triangleq & \forall \event_1,\event_2 . (\seqbef(\event_1,\event_2)\implies\lnot\exists \event_3.(\event_2\neq \event_3 \land \seqbef(\event_1,\event_3))\numberthis \\
\land& \ \forall \event_1,\event_2 . (\seqbef(\event_2,\event_1)\implies\lnot\exists \event_3.(\event_2\neq \event_3 \land \seqbef(\event_3,\event_1))\\
\land& \ \forall \event_1,\event_2 . (\seqbef(\event_2,\event_1)\implies \tid(\event_1)=\tid(\event_2))\\
\land& \ \forall \event_1, \event_2. ((\lnot \exists \event_3. (\seqbef(\event_3,\event_1) \lor \seqbef(\event_3,\event_2))) \implies \tid(\event_1)\neq\tid(\event_2))\\
\end{align*}
\begin{align*}
\msorf \triangleq \forall \event_1, \event_2. &(\rf(\event_1,\event_2)\implies \op(\event_1)\in \{\wt,\rmw\} \land  \op(\event_2)\in \{\rd,\rmw\}\numberthis\\
& \land \lloc(\event_1)=\lloc(\event_2)\\
& \land \val_{\wt}(\event_1)=\val_{\rd}(\event_2)\\
& \land  \lnot \exists \event_3. (\event_3\neq \event_1 \land \rf(\event_3,\event_2)))\\
\land  \forall \event_1. &(\op(\event_1)\in \{\rd,\rmw\} \implies \exists \event_2. \rf(\event_2,\event_1))
\end{align*}
Observe that neither formula contains second-order quantifiers, which implies that the set of $\rf$-graphs is, in fact, first-order definable.
\end{proof}

\smallskip
\lemfreerfgraphsdefinable*
\begin{proof}
Let $\ConcProg=\tuple{\Prog_{t_1},\dots, \Prog_{t_{|\ThreadDom|}}}$ be the concurrent program, where $\Prog_{t}=$

$\tuple{\States_{t}, \Alphabet_{t}, \Trans_{t}, \InitState_{t}}$ is each sequential program of $\ConcProg$.
Let $\Trans = \bigcup_{t\in \ThreadDom}\Trans_t$. 
For $\delta = (p,\alpha,q) \in \Trans$, we denote $p(\delta)=p$,  $\alpha(\delta)=\alpha$ and $q(\delta)=q$. 
The following formula, conjuncted with $\msosb$, expresses that $\tuple{E, \seqbef} \in \Lang{\ConcProg}$.
\begin{align*}
\msoprogsb \triangleq& \exists (X_\delta)_{\delta \in \Delta}.\\
\Biggl[ \forall \event_1. &\left(\left(\bigvee_{\delta\in \Delta} \event_1 \in X_\delta \land \alpha(\delta) = \llab(\event_1) \right) \land \left(\bigwedge_{\delta_1,\delta_2\in \Delta} \event_1 \in X_{\delta_1} \implies \lnot e_1 \in X_{\delta_2}\right)\right)\\
\land \forall \event_1,\event_2. &\left(\bigwedge_{\delta_1,\delta_2\in \Delta} \left(\event_1\in X_{\delta_1} \land \event_2\in X_{\delta_2} \land  \seqbef(\event_1,\event_2)\right) \implies q(\delta_1) = p(\delta_2)\right)\\
\land \forall \event_1. &\left(\bigwedge_{\delta \in \Delta} \left(\event_1\in X_{\delta} \land \lnot\exists \event_2. \seqbef(\event_2,\event_1)\right) \implies p(\delta)=\InitState_{\tid(\event_1)}\right)\Biggl]
\end{align*}
\end{proof}

\subsection{Proof of \cref{thm:mso_definable_mms}}\label{SUBSEC:APP_MSO_DEF}

We treat each memory model individually. 
An overview of the case of $\ramm$, and its MSO formulas, are presented in the main text.
\smallskip
\lemramso*
\begin{proof}
First, assume that $G\models \msora$ and we argue that $G\models \ramm$ by constructing a total modification order $\mo_x$, for each register $x$, such that $\tuple{\Events, \seqbef, \rf, \mo}\models \ramm$, where $\mo=\bigcup_{x\in \LocationDom} \mo_x$.
Observe that since $G\models \msohbirr$, the \ref{eq:hbirr} axiom holds independently of this $\mo$.
Moreover, since $G\models \msowat$, the write/RMW events of $G$ are partitioned into RMW chains.
For each register $x$, we construct the desired total $\mo_x$ by
\begin{enumerate*}[label=(\roman*)]
\item obtaining a total order on all write events on $x$ such that if $\msopmora(\wt_1, \wt_2)$ then $\wt_1\LTo{\mo_x}\wt_2$, and
\item appending right after each write $\wt$ the events of RMW chain of $\wt$ (in the same order).
\end{enumerate*}
The acyclicity of $\msopmora$ (\cref{eq:mso_rc_wc_ra}) guarantees that this process is well-defined, and that the resulting $\mo$ satisfies the \ref{eq:at} axiom.
Moreover, since $\msopmora(\event_1,\event_2)$ whenever $\event_1\LTo{\hb}\event_2$, we have that the \ref{eq:wc} axiom is satisfied.
Finally, we argue that \ref{eq:rc} also holds.
Consider any conflicting triplet $\tuple{\event_3, \event_4, \event_5}$ such that $\event_3\LTo{\hb}\event_5$,
and let $\event_1$ and $\event_2$ be the heads of of the RMW chains of $\event_3$ and $\event_4$.
Note that  $\event_4\not \LTo{\rf^+}\event_3$, as in this case $\event_1=\event_2$ and we would have $\msopmora(\event_1, \event_2)$ due to \cref{eq:mso_pmora}, contradicting the fact that $G\models \msorcwcra$.
If $\event_3\LTo{\rf^+}\event_4$, we have $\event_1=\event_2$ and thus $\event_3\LTo{\mo_x}\event_4$ since the events in each RMW chain are $\mo_x$ ordered according to $\rf^+$.
The only remaining case is that $\event_1\neq \event_2$.
Observe that in this case we have $\msopmora(\event_1,\event_2)$, and by construction, all events of the RMW chain of $\event_1$ are $\mo_x$-ordered before all events of the RMW chain of $\event_2$, therefore $\event_3\LTo{\mo_x}\event_4$.

For the opposite direction, assume that $G\models \ramm$ and we argue that $G\models \msora$.
It is straightforward to verify that $G$ satisfies \ref{eq:hbirr} and \ref{eq:wat} and thus $G\models \msohbirr$ and $G\models \msowat$.
Consider any $\mo$ that witnesses the consistency of $G$.
It is straightforward to verify that for any $\event_1$, $\event_2$ such that $\msopmora(\event_1, \event_2)$, we have $\event_1\LTo{\mo_x}\event_2$, due to \ref{eq:wc}, \ref{eq:rc} and the totality of $\mo_x$.
The acyclicity of $\mo_x$ implies that $G\models \msorcwcra$, thus $G\models \msora$.

The desired result follows.
\end{proof}

\Paragraph{The case of $\rlxmm$.}
To define $\RFGraphsMM{\rlxmm}$ in MSO, we must express the \ref{eq:scloc} axiom.
To do this in a direct way, we first recall that, for a single register, $\rlxmm$ and $\ramm$ coincide.
It thus suffices to tailor the formulas for $\ramm$ to a single register.
In particular, given a register $x$, we let
\begin{enumerate*}[label=(\roman*)]
\item  $\msohbirrx$ be identical to $\msohbirr$, and
\item $\msorcwcrlx$ be identical to $\msorcwcra$,
\end{enumerate*}
except that, in both cases, $\hb$ is replaced by $(\po_x\cup \rf_x)^+$.
Then
\[
\msorlx \triangleq  \bigwedge_{x \in \LocationDom} \msohbirrx \land \msowat \land \msorcwcrlx
\numberthis\label{eq:mso_rlx}
\]
expresses \ref{eq:scloc}, as stated in the following lemma.

\smallskip
\begin{restatable}{lemma}{lemrlxmso}\label{lem:rlx_mso}
For any $\rf$-graph $G=\tuple{\Events, \seqbef, \rf}$, we have $G\models \msorlx$ iff $G\models \rlxmm$.
\end{restatable}
\begin{proof}
The proof follows easily from \cref{lem:ra_mso} and the fact that, on a single register, $\rlxmm$ and $\ramm$ coincide.
\end{proof}

\Paragraph{The case of $\rmwfree{\sramm}$.}
The RMW-free fragment of $\sramm$ requires \ref{eq:swc} and \ref{eq:rc}.
We begin with the following formula $\msopmosra$, which captures $\mo$ orderings between write events due to \ref{eq:rc}.
This is similar, in spirit, to $\msopmora$ for $\ramm$, with two main differences:
\begin{enumerate*}[label=(\roman*)]
\item we do not need to infer orderings of the form $\wt_1\LTo{\mo}\wt_2$ when $\wt_1\LTo{\hb}\wt_2$, as \ref{eq:swc} mixes $\hb$ and $\mo$ and this allows to defer such $\mo$ orderings when we check for \ref{eq:swc} later, and
\item in the absence of RMW-chains, orderings of the form $\wt_1\LTo{\mo}\wt_2$ only occur when there is a conflicting tripplet $\tuple{\wt_1, \wt_2, \rd}$ with $\wt_1\LTo{\hb}\rd$.
\end{enumerate*}
\[
\msopmosra(\event_1,\event_2) \triangleq \event_1 \neq \event_2 \land \op(\event_1) = \wt \land \lloc(\event_1) = \lloc(\event_2) \land \exists \event_3. \left(\varphi_\rf(\event_2, \event_3) \land \msohb(\event_1,\event_3)\right)\numberthis
\]
The corresponding coherence formula $\msorcswcra$ for $\rmwfree{\sramm}$ expresses the requirement that $\msopmosra$ mixed with $\hb$ is acyclic, using the transitive closure formula $\varphi_{R^+}$ from \cref{eq:mso_tc} with $R=\msopmosra \lor \msohb$.
This is also the final formula $\msosranormw$ expressing consistency in $\rmwfree{\sramm}$.
\[
\msosranormw  \triangleq \msorcswcra \triangleq \lnot \exists \event. \varphi_{R^+}(\event,\event)
\numberthis\label{eq:mso_rc_swc_ra}
\]

\smallskip
\begin{restatable}{lemma}{lemsramso}\label{lem:sra_mso}
For any $\rf$-graph $G=\tuple{\Events, \seqbef, \rf}$, we have $G\models \msosranormw$ iff $G\models \rmwfree{\sramm}$.
\end{restatable}
\begin{proof}
First, assume that $G\models \msosranormw$ and we argue that $G\models \rmwfree{\sramm}$ by constructing a total modification order $\mo_x$, for each register $x$, such that $\tuple{\Events, \seqbef, \rf, \mo}\models \rmwfree{\sramm}$, where $\mo=\bigcup_{x\in \LocationDom} \mo_x$.
We construct a total order $\sigma$ on all write events (across all registers) such that if $\msopmosra(\wt_1, \wt_2)\lor \msohb(\wt_1, \wt_2)$ then $\wt_1\LTo{\mo_x}\wt_2$.
The acyclicity of this relation (\cref{eq:mso_rc_swc_ra}) guarantees that this process is well-defined.
The existence of $\sigma$ implies that $\mo$ satisfies \ref{eq:swc}, where $\mo_x$ orders all write events on register $x$ according to $\sigma$.
Finally, \ref{eq:rc} follows by a simplified version of the proof of this axiom for \cref{lem:ra_mso}.

For the opposite direction, assume that $G\models \rmwfree{\sramm}$ and we argue that $G\models \msosranormw$.
Consider any $\mo$ that witnesses the consistency of $G$.
It is straightforward to verify that for any $\event_1$, $\event_2$ such that $\msopmosra(\event_1, \event_2)$, we have $\event_1\LTo{\mo_x}\event_2$, due to \ref{eq:rc} and the totality of $\mo_x$.
Then, \ref{eq:swc} implies that $G\models \msorcswcra$, thus $G\models \msosranormw$.

The desired result follows.
\end{proof}

\Paragraph{The case of $\wramm$.}
Recall that we have already defined the MSO formulas $\msohbirr$ (\cref{eq:mso_hbirr}) and $\msowat$ (\cref{eq:mso_wat}), respectively expressing \ref{eq:hbirr} and \ref{eq:wat}.
The following formula expresses \ref{eq:wrc}.
\begin{align}
\msowrc \triangleq & \ \lnot \exists \event_1,\event_2,\event_3. (\event_1\neq \event_2 \land \event_2\neq \event_3 \land \event_1\neq \event_3\\&\land \lloc(\event_1) = \lloc(\event_2) \land \lloc(\event_2) = \lloc(\event_3) \nonumber\\
&\land \op(\event_2)\in \{\wt,\rmw\} \land \rf(\event_1,\event_3) \land \msohb(\event_1,\event_2) \land \msohb(\event_2,\event_3))\nonumber
\end{align}
and the final formula for $\RFGraphsMM{\wramm}$ is
\[
\msowra \triangleq \msohbirr \land \msowrc \land \msowat
\numberthis\label{eq:mso_wra}
\]

\smallskip
\begin{restatable}{lemma}{lemwramso}\label{lem:wra_mso}
For any $\rf$-graph $G=\tuple{\Events, \seqbef, \rf}$, we have $G\models \msowra$ iff $G\models \wramm$.
\end{restatable}
\begin{proof}
The proof is straightforward as all axioms are directly translatable to MSO.
\end{proof}

\subsection{Proofs of \cref{SUBSEC:BOUNDED_TW_REACHABILITY}}

\lemboundedtwreachability*
\begin{proof}
Since the axioms of $\mm$ are MSO-definable, \cref{lem:rf_graphs_mso_definable} and \cref{lem:free_rf_graphs_definable} imply that the set $\RFGraphsOf{\ConcProg}{\mm}$ is also MSO-definable, via some formula  $\Phi$.
Hence, the problem of whether there exists an $\rf$-graph $G\in \RFGraphsOf{\ConcProg}{\mm}$ with $\tw(G)\leq k$ and such that $G\models \varphi$ is equivalent to deciding whether $\Phi\land \varphi$ is satisfiable over graphs of treewidth $\leq k$, 
which is decidable due to \cref{cor:mso_satisfiability}.
\end{proof}

\smallskip
\lemboundedtwdiff*
\begin{proof}
Let $\Phi_1$ and $\Phi_2$ be the MSO formulas for the sets $\RFGraphsOf{\ConcProg_1}{\mm_1}$ and $\RFGraphsOf{\ConcProg_2}{\mm_2}$.
We have $\RFGraphsOfTw{\ConcProg_1}{\mm_1}{k}\subseteq \RFGraphsOfTw{\ConcProg_2}{\mm_2}{k}$ iff the MSO formula $\Phi_1\implies \Phi_2$ is satisfiable over graphs of treewidth $\leq k$, which is decidable due to \cref{cor:mso_satisfiability}.
\end{proof}

\subsection{Proof of \cref{thm:non_mso_definability}}\label{SUBSEC:APP_NON_MSO_DEFINABLE_MM}

We treat each case individually. 

\Paragraph{The case of $\scmm$.}
The work of~\cite{Mathur2020} shows that there is a linear-time reduction from OV to checking the $\scmm$-consistency of RMW-free $\rf$-graphs consisting of $2$ threads and $7$ registers.
If $\RFGraphsMM{\scmm}$ was MSO definable via a formula $\varphi$, we would be able to solve OV in nearly linear time by handling an OV-instance $\OVInstance$ as follows:
\begin{enumerate*}[label*=(\roman*)]
\item use the reduction of~\cite{Mathur2020} to obtain an $\rf$-graph $G$,
\item use \cref{lem:sc-tw} to show that if $G\models \scmm$ then $\tw(G)\leq 2+7=9$, and 
\item use \cref{thm:courcelle} on $\varphi$, $G$ and $9$ to decide whether $G$ is such that $\tw(G)\leq 9$ and $G\models \varphi$.
\end{enumerate*}

\Paragraph{The case of $\sramm$.}
We now consider $\sramm$, and show that the presence of RMWs breaks the MSO definability of \cref{lem:sra_mso}.
Our proof strategy has two steps.
First, we observe that the reduction of \cite{Mathur2020} constructs $\rf$-graphs $G$ for which $G.\rf$ is \emph{functional}, i.e., every write is observed at most once.
The importance of this property, for our purposes, is highlighted in the following lemma.

\smallskip
\begin{restatable}{lemma}{lemsctosra}\label{lem:sc_to_sra}
Consider any RMW-free $\rf$-graph $G$ where $G.\rf$ is functional.
Let $G'$ be the $\rf$-graph obtained from $G$ by replacing each read event by an RMW event.
Then $G\in \RFGraphsMM{\scmm}$ iff $G'\in\RFGraphsMM{\sramm}$.
\end{restatable}
\begin{proof}
We first prove that $G\models \scmm$ iff $G'\models \scmm$ when every event $\rd_i$ in $G$ is replaced by $\rmw_i$ in $G'$. 
Each such $\rmw_i$ writes some arbitrary value.
Let $G=\tuple{E, \seqbef, \rf}$ and $G'=\tuple{E', \seqbef', \rf'}$.
Since the transformation only relabels each read $\rd_i$ as an RMW $\rmw_i$ on the same location and value, $E$ and $E'$ are in bijection and we identify each event with its image.
In the scope of this proof, we use unprimed relations to refer to relations defined over $G$, and primed relations to refer to relations defined over $G'$ (e.g., $\po\triangleq\seqbef^+$ and $\po'\triangleq (\seqbef')^+$, respectively).
By construction, $\po'= \po$ and $\rf' = \rf$.

Let $\mo$ be the modification order that witnesses $G\models \scmm$.
We define $\mo'$ to be such that $\wt_i \LTo{\mo'} \wt_j$ iff $\wt_i \LTo{\mo} \wt_j$, and for every $\wt_i \LTo{\rf'} \rmw_i$, $\rmw_i$ immediately succeeds $\wt_i$ in $\mo'$. 
This is only possible because $\rf$ is functional. 
Let $R\triangleq\po\cup\rf\cup\mo\cup\fr$ and $R'\triangleq\po'\cup\rf'\cup\mo'\cup\fr'$.
We show $R^+=(R')^+$; consequently, $G\models\scmm$ iff $G'\models\scmm$.

First, we have $R\subseteq R'$, since $\mo\subseteq\mo'$ and $\fr\subseteq\fr'$, and thus $R^+\subseteq (R')^+$.
For the opposite direction, we argue that $R'\subseteq R^+$.
Indeed, this follows from the fact that $\fr' \subseteq \mo$ and $\mo' \subseteq \mo \cup (\mo;\rf) \cup \fr \cup (\fr;\rf)$. 
These represent $\mo'$ connection among event pairs of type $(\wt,\wt)$, $(\wt, \rmw)$, $(\rmw, \wt)$ and $(\rmw, \rmw)$, respectively. 

It remains to prove that $G'\models \scmm$ iff $G'\models \sramm$.
Since $ \scmm \mmorder \sramm$, it suffices to show that $G'\models \sramm \implies G'\models \scmm$. 
Since there are no read events in $G'$, $\fr'\subseteq \mo'$. 
So the satisfiability of \ref{eq:swc} implies \ref{eq:sc}, which means $G'\models \scmm$.
\end{proof}

\input{figures/g_to_gp_pso}

We can now prove \cref{thm:non_mso_definability} for $\sramm$. 
If $\RFGraphsMM{\sramm}$ was MSO definable via a formula $\varphi$, we would be able to solve OV in nearly linear time by
handling an OV-instance $\OVInstance$ as follows.
First, we use the reduction of~\cite{Mathur2020} to obtain an $\rf$-graph $G$ and transform it to $G'$ as above.
Second we argue that if $G'\models \sramm$ then $\tw(G')\leq 9$.
By \cref{lem:sc_to_sra}, we have $G\models \scmm$, and thus \cref{lem:sc-tw} yields  $\tw(G)\leq 2+7=9$.
However, $G$ and $G'$ are isomorphic, thus $\tw(G')=\tw(G)\leq 9$.
Finally, we use \cref{thm:courcelle} on $\varphi$, $G'$ and $9$ to decide whether $G'$ is such that $\tw(G')\leq 9$ and $G'\models \varphi$.  

\Paragraph{The case of $\tsomm$ and $\psomm$.} 
Finally, we deal with the case of $\tsomm$ and $\psomm$.
Given an $\rf$-graph $G$ obtained by the reduction of~\cite{Mathur2020}, we construct a graph $G'$ over $4$ threads and $9$ registers such that $G\models \scmm$ iff $G'\models \scmm, \tsomm, \psomm$.
To do so, we pair every write event $\wt$ of $G$ with a gadget in $G'$ that makes $\wt$ visible to all threads before the next event in its own thread is executed.
In particular, consider any RMW-free $\rf$-graph $G=\tuple{\Events,\seqbef,\rf}$, and we construct $G'=\tuple{\Events',\seqbef',\rf'}$ as follows.

\begin{compactenum}
\item For each thread $t$ of $G$, we introduce an additional thread $t'$ in $G'$ and an additional register $v_{t}$.
\item We scan the events of $t$ in $\seqbef$-order, and for every write $\wt(t, x)$, we insert a gadget by
\begin{enumerate*}[label*=(\roman*)]
\item appending a read event $\rd(t', x)$ followed by a write event $\wt(t', v_{t})$ at the end of $t'$, and
\item inserting a read event $\rd(t,v_{t})$ exactly after $\wt(t,x)$.
\end{enumerate*}
We further set $\wt(t,x) \LTo{\rf} \rd(t', x)$ and $\wt(t',v_{t}) \LTo{\rf} \rd(t, v_{t})$.
The values of the new events are not important, as long as they are consistent with $\rf$.
\end{compactenum}
This transformation guarantees that if $\wt(t,x)\LTo{\seqbef^+} e$ in $G$, for some event $\event$, then
$\wt(t,x)\LTo{\rfe'} \rd(t', x)\LTo{\pppo'} \wt(t',v_{t}) \LTo{\rfe'}\rd(t, v_{t}) \LTo{\pppo'} \event$ in $G'$ (where $\rfe'$ and $\pppo'$ are defined over $\seqbef'$ and $\rf'$).
We have the following lemma.

\smallskip
\begin{restatable}{lemma}{lemsctotso}\label{lem:sc_to_tso}
Consider any RMW-free $\rf$-graph $G$. 
Let $G'$ be the graph constructed from $G$ by the gadgetizing process. 
Then $G \in \RFGraphsMM{\scmm}$ iff  $G' \in \RFGraphsMM{\scmm}$ iff $G' \in \RFGraphsMM{\tsomm}$ iff $G' \in \RFGraphsMM{\psomm}$.
\end{restatable}
\begin{proof}
Let $G=\tuple{E,\seqbef,\rf}$ and $G'=\tuple{E',\seqbef',\rf'}$.
In the scope of this proof, we use unprimed relations to refer to relations defined over $G$, and primed relations to refer to relations defined over $G'$ (e.g., $\po\triangleq\seqbef^+$ and $\po'\triangleq (\seqbef')^+$, respectively).

First, assume that $G\models \scmm$, and we argue that $G'\models \scmm$ (which also implies that $G'\models \tsomm, \psomm$).
Since $G\models \scmm$, there exists a modification order $\mo$ and a corresponding $\fr$ such that $R=(\seqbef\cup \rf \cup \mo \cup \fr)$ is acyclic.
Let $\sigma$ be any total extension of $R^+$, and construct a sequence of events $\sigma'$ by inserting after every write event $\wt$ the gadget sequence of $\wt$ in $G'$ (i.e., the $\seqbef$-successor of $\wt$ in $G'$, followed by the corresponding read and write event in the primed thread of $\wt$).
It is easy to verify that $\sigma'$ witnesses the $\scmm$-consistency of $G'$, via an $\mo'$ that orders all write events to the same register according to $R^+$.

For the opposite direction, assume that $G'\models \psomm$ (which is also implied by $G'\models \tsomm$ or $G'\models \scmm$), and we argue that $G\models \scmm$.
Since $G'\models \psomm$, there exists a modification order $\mo'$ and a corresponding $\fr'$ such that 
$R'=(\pppo'\cup \rfe' \cup \moe' \cup \fre')$ is acyclic.
Consider any two events $\event_1, \event_2$ of $G$ such that $\event_1\LTo{\seqbef}\event_2$.
If $\event_1$ is a read event, we have $\event_1\LTo{\pppo'} \event_2$, while if $\event_1$ is a write, we have $\event_1\LTo{\rfe';\pppo';\rfe';\pppo'} \event_2$.
In both cases, we have that $\event_1\LTo{(R')^+}\event_2$.
This implies that the relation $\identity{E};( \seqbef \cup \rfe' \cup \moe' \cup \fre'); \identity{E}$ is acyclic.
Since $G'$ also satisfies \ref{eq:scloc}, we have that $\identity{E};\rf';\identity{E}\subseteq \identity{E};(\rfe'\cup \seqbef);\identity{E}$, and similarly for $\moe'$ and $\fre'$.
This implies that  $\identity{E};( \seqbef \cup \rf' \cup \mo' \cup \fr'); \identity{E}$ is acyclic.
The $\scmm$-consistency of $G$ then follows easily, witnessed by $\mo=\identity{E};\mo'; \identity{E}$.
\end{proof}

We can now prove \cref{thm:non_mso_definability} for $\tsomm$ and $\psomm$. 
For any $\mm\in \{ \tsomm, \psomm \}$,
if $\RFGraphsMM{\mm}$ was MSO definable via a formula $\varphi$, we would be able to solve OV in nearly linear time by
handling an OV-instance $\OVInstance$ as follows.
First, we use the reduction of~\cite{Mathur2020} to obtain an $\rf$-graph $G$, and the above (linear-time) transformation of $G$ to an $\rf$- graph $G'$.
Second, we argue that if $G'\models \mm$ for $\mm \in \{\tsomm, \psomm\}$ then $\tw(G')\leq 13$.
This is because, by \cref{lem:sc_to_tso}, we have $G'\models \scmm$, and thus \cref{lem:sc-tw} yields  $\tw(G')\leq 4+9=13$.
Finally, we use \cref{thm:courcelle} on $\varphi$, $G'$ and $13$ to decide whether $G'$ is such that $\tw(G')\leq 13$ and $G'\models \varphi$.  

%% file: figures/g_to_gp_pso.tex
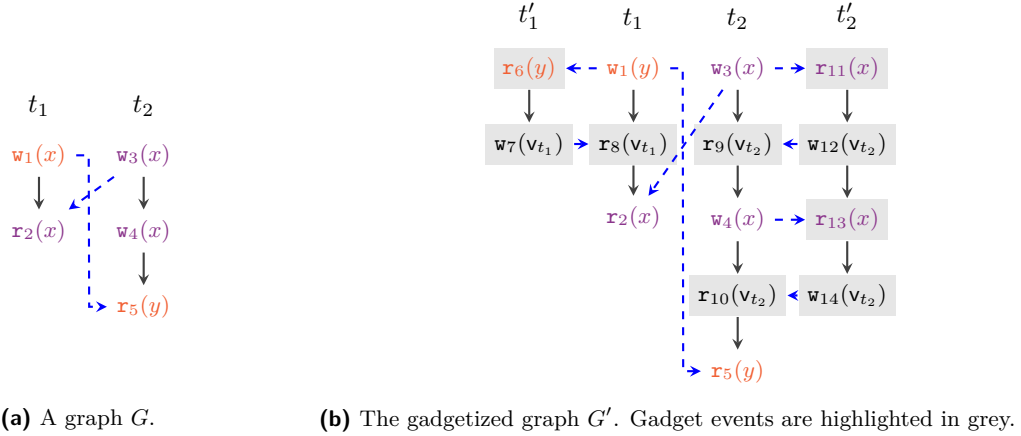
\begin{figure}
\begin{subfigure}[t]{.17\textwidth}
\centering
\begin{tikzpicture}[ line width=1pt, node distance = 1cm] 

            \node (wy1) {\textcolor{colorVARb}{\footnotesize$\wt_1(x)$}};
            \node[below of = wy1] (rx2){\footnotesize\textcolor{colorVARa}{$\rd_2(x)$}};

            \node[right = 0.4cm of wy1] (wx3) {\footnotesize\textcolor{colorVARa}{$\wt_3(x)$}};
            \node[below of = wx3] (wx4) {\footnotesize\textcolor{colorVARa}{$\wt_4(x)$}};
            \node[below of = wx4] (ry5) {\footnotesize\textcolor{colorVARb}{$\rd_5(y)$}}; 

            \node[below of = ry5, white] (aaa) {}; 

            \node[draw=none, left = 0.1cm of wx3] (h1) {};
            \node[draw=none, left = 0.1cm of ry5] (h2) {};

            \node[above = 0.1cm of wy1] (t1){$t_1$};
            \node[above = 0.1cm of wx3] (t2){$t_2$};

            \draw [seqbef] (wy1) to (rx2);
            \draw [seqbef] (wx3) to (wx4);
            \draw [seqbef] (wx4) to (ry5);
            \draw [rf] (wx3) to (rx2);
            \draw [rf] (wy1) -- (h1.center)  -- (h2.center) -- (ry5); 
            
\end{tikzpicture}
\caption{A graph $G$.}
\label{subfig:g_to_gp_pso_g}
\end{subfigure}%
\hspace{5em}
\begin{subfigure}[t]{.7\textwidth}
\centering
\begin{tikzpicture}[ line width=1pt, node distance = 1cm] 

            \node[fill=gray!20, left = 0.4cm of wy1] (ry1) {\footnotesize\textcolor{colorVARb}{$\rd_{6}(y)$}};
            \node[fill=gray!20, below of = ry1] (wt11){\footnotesize{$\wt_{7}(\var_{t_1})$}};
            
            \node (wy1) {\textcolor{colorVARb}{\footnotesize$\wt_{1}(y)$}};
            \node[fill=gray!20, below of = wy1] (rt11){\footnotesize{$\rd_{8}(\var_{t_1})$}};
            \node[below = 0.4cm of rt11] (rx2){\footnotesize\textcolor{colorVARa}{$\rd_{2}(x)$}};

            \node[right = 0.4cm of wy1] (wx3) {\footnotesize\textcolor{colorVARa}{$\wt_{3}(x)$}};
            \node[fill=gray!20,below of = wx3] (rt23) {\footnotesize{$\rd_{9}(\var_{t_2})$}};
            \node[below of = rt23] (wx4) {\footnotesize\textcolor{colorVARa}{$\wt_{4}(x)$}};
            \node[fill=gray!20,below of = wx4] (rt24) {\footnotesize{$\rd_{10}(\var_{t_2})$}};
            \node[below of = rt24] (ry5) {\footnotesize\textcolor{colorVARb}{$\rd_{5}(y)$}}; 

            \node[fill=gray!20, right = 0.4cm of wx3] (rx3) {\footnotesize\textcolor{colorVARa}{$\rd_{11}(x)$}};
            \node[fill=gray!20,below of = rx3] (wt23) {\footnotesize{$\wt_{12}(\var_{t_2})$}};
            \node[fill=gray!20, below of = wt23] (rx4) {\footnotesize\textcolor{colorVARa}{$\rd_{13}(x)$}};
            \node[fill=gray!20,below of = rx4] (wt24) {\footnotesize{$\wt_{14}(\var_{t_2})$}};

            \node[draw=none, left = 0.1cm of wx3] (h1) {};
            \node[draw=none, left = 0.1cm of ry5] (h2) {};

            \node[above = 0.1cm of ry1] (t1p){$t'_1$};
            \node[above = 0.1cm of wy1] (t1){$t_1$};
            \node[above = 0.1cm of wx3] (t2){$t_2$};
            \node[above = 0.1cm of rx3] (t2){$t'_2$};

            \draw [seqbef] (ry1) to (wt11);

            \draw [seqbef] (wy1) to (rt11);
            \draw [seqbef] (rt11) to (rx2);
            
            \draw [seqbef] (wx3) to (rt23);
            \draw [seqbef] (rt23) to (wx4);
            \draw [seqbef] (wx4) to (rt24);
            \draw [seqbef] (rt24) to (ry5);

            \draw [seqbef] (rx3) to (wt23);
            \draw [seqbef] (wt23) to (rx4);
            \draw [seqbef] (rx4) to (wt24);

            \draw [rf] (wy1) to (ry1);
            \draw [rf] (wt11) to (rt11);
            \draw [rf] (wx3) to (rx3);
            \draw [rf] (wt23) to (rt23);
            \draw [rf] (wx4) to (rx4);
            \draw [rf] (wt24) to (rt24);
            
            \draw [rf] (wx3) to (rx2);
            \draw [rf] (wy1) -- (h1.center)  -- (h2.center) -- (ry5); 
\end{tikzpicture}
\caption{
The gadgetized graph $G'$. 
Gadget events are highlighted in grey. 
}
\label{subfig:g_to_gp_pso_gp}
\end{subfigure}%
\caption{
Illustration of transformation used in \cref{lem:sc_to_tso}. 
}
\label{fig:g_to_gp_pso}
\end{figure}

%% file: app_robustness.tex
\section{Proofs of \cref{SEC:RFROB}}\label{SEC:APP_RFROB}

\subsection{Proofs of \cref{subsec:rf-robustness}}

\lemtwcontiguousmodels*
\begin{proof}
Consider any program $\ConcProg$ and some $\rf$-graph $G\in \RFGraphsOf{\ConcProg}{\mm}$, let $\event$ be a $(\seqbef\cup \rf)^+$ maximal event of $G$.
Let $G'$ be the $\rf$-graph obtained by removing $\event$ from $G$.
Note that $G'\in \RFGraphsOf{\ConcProg}{\mm}$, and $\tw(G')\geq \tw(G)-1$.
If $\tw(G')=\tw(G)-1$, we are done.
Otherwise, we repeat the process on $G'$.
Since every step reduces the size of the graph by $1$ event, and the treewidth of the graph is bounded by its size, we are guaranteed to reach a graph $G''$ such that $\tw(G'')=\tw(G)-1$.
The desired result follows.
\end{proof}

\smallskip
\lemcourcelletheoremexacttw*
\begin{proof}
We make use of MSO transductions, as used in \cite{Mikolaj2016,Mikolaj2017}, which are MSO sentences transforming input structures of one type to output structures of another.
In our case, the input structures are $\rf$-graphs, while the output structures are (an encoding of) tree decompositions.
As stated in~\cite{Mikolaj2016}, if $\Phi$ is an MSO transduction and $\phi$ is an MSO sentence over its output vocabulary, then the set of $u \in U$ for which at least one element of $\Phi(u)$ satisfies $\phi$ is MSO-definable.

As shown in~\cite{Mikolaj2016,Mikolaj2017} that there exists an MSO transduction $\Phi_k$ from graphs to tree decompositions such that
\begin{enumerate*}[label=(\roman*)]
\item if $\tw(G)\leq k$ then there exists least one output tree decomposition, and
\item every output tree decomposition has width $\tw(G)$.
\end{enumerate*}
Moreover, checking whether a tree decomposition has width exactly $k$ is easily expressible in MSO.
This implies that the set of graphs $G$ such that a decomposition in $\Phi_k(G)$ has width exactly $k$ is MSO-definable, via some MSO formula $\varphi_{\tw=k}$.
We can thus decide whether there exists a graph $G$ with $\tw(G)= k$ and $G\models\varphi$ by applying \cref{cor:mso_satisfiability} on the MSO formula $\varphi\land \varphi_{\tw=k}$.
\end{proof}

\subsection{Proofs of \cref{SUBSEC:AUG_RF}}\label{SUBSEC:APP_AUG_RF}

\lemobsrobust*
\begin{proof}
Since $G$ reaches $\tm$, there exists some thread $t$ such that $\tm(t)\in G.\Events$.
Since $\ConcProgAug$ is observationally $\rf$-robust against $\mm$, there exists some $G' \in \RFGraphsOf{\ConcProgAug}{\scmm}$ such that $G\simeq G'$.
It follows that $\tm(t)\in G'.\Events$, thus $G'$ also reaches $\tm$.
\end{proof}

\smallskip
\thmobsrobustnessvsreachability*

Below, we detail the procedure used to prove the theorem.

\SubParagraph{1.~Reduced $\rf$-semantics.}
The treewidth divergence of observationally $\rf$-robust programs is purely due to $\rf$-edges of unused reads.
Our first insight is that the invariance under unused-reads property allows us to capture the behavior of $\ConcProgAug$ by syntactically omitting $\rf$-edges to unused reads.
In particular, we define a \emph{reduced} $\rf$-graph $H$ as an augmented $\rf$-graph where unused reads do not have $\rf$-edges, i.e., the reads-from relation of $H$ satisfies $\rf \subseteq (\bigcup_{x\in \LocationDom}(\locx{\WriteDom} \cup \locx{\RMWDom}) \times ((\locx{\ReadDom}\setminus \UnusedReadDomX(H)) \cup \locx{\RMWDom}))$. 
For a reduced $\rf$-graph $H$ and an augmented $\rf$-graph $G$, we write $G\Reduces H$ to denote that $H$ occurs from $G$ after removing the $\rf$-edges of unused reads, i.e.,
\begin{enumerate*}[label=(\roman*)]
\item $H.\Events = G.\Events$,
\item $H.\seqbef = G.\seqbef$, and
\item $H.\rf \subseteq G.\rf$.
\end{enumerate*}
The \emph{reduced $\rf$-semantics} of $\ConcProgAug$ is the set
\[
\RedGraphsOf{\ConcProgAug}{\mm}=\{ H \colon H \text{ is a reduced $\rf$-graph and there exists $G\in \RFGraphsOf{\ConcProgAug}{\mm}$ s.t. $G\Reduces H$} \}\ .
\]
Reduced $\rf$-graphs serve as an abstraction of program executions, that is nevertheless sound and complete wrt $\tm$ reachability, i.e., if there is some $H\in \RedGraphsOf{\ConcProg}{\mm}$ reaching $\tm$, then there is some $G\in \RFGraphsOf{\ConcProgAug}{\mm}$ also reaching $\tm$ (in particular, that $G$ for which $G\Reduces H$).
Moreover, the set $\RedGraphsOf{\ConcProgAug}{\scmm}$ has treewidth bounded by $|\ThreadDom|+|\LocationDom|$, much like $\RFGraphsOf{\ConcProg}{\scmm}$.

\smallskip
\begin{restatable}{lemma}{lemreducedsemanticsMSO}\label{lem:reduced_semantics_MSO}
For any augmented concurrent program $\ConcProgAug$ and memory model $\mm\in \{\ramm, \rmwfree{\sramm}, \rlxmm, \wramm \}$, the set $\RedGraphsOf{\ConcProgAug}{\mm}$ is MSO-definable.
\end{restatable}
\begin{proof}

This proof follows the same lines of \cref{lem:free_rf_graphs_definable}.  We may check whether  $\tuple{E, \seqbef} \in \Lang{\ConcProgAug}$ using $\msoprogsb$. Now we must verify whether some $\tuple{E, \seqbef, \rf}$ is a reduced $\rf$-graph (similarly to how $\msorf$ expresses that $\tuple{E, \seqbef, \rf}$ is an $\rf$-graph). 
We do so through $\msorfaug$, defined below.
First, $\varphi_\text{reach}^x(\event_1, \event_2)$ models $\seqbef$, but excluding incoming edges to $\rd$ or $\rmw$ events on $x$.
This is useful as we later want to exclude paths that contain such events.
\begin{align}
\varphi_\text{reach}^x(\event_1, \event_2) \triangleq \seqbef(\event_1,\event_2) \land \lnot \left(\lloc(\event_2)=x \land \op(\event_2)\in \{\rd,\rmw\}\right)
\end{align}
Next, $\varphi_\text{used}$ determines whether a read event $\event_1$ on $x$ is used, by verifying whether it can reach some event $\ck(x)$ through an $\seqbef$-path that does not contain a $\rd$ or $\rmw$ event on the same location.
Here, $\varphi_{\text{reach}^+}^x (\event_1, \event_2) = \varphi_{R^+}(\event_1, \event_2)$ with $R(\event_1, \event_2) = \varphi_\text{reach}^x(\event_1, \event_2)$, using \cref{eq:mso_tc}.
\begin{align}
\varphi_\text{used}(\event_1) \triangleq \ &\op(\event_1)=\rd \land \\ &\bigvee_{x\in \LocationDom} (\lloc(\event_1) = x \land \exists \event_2.
\left(\varphi_{\text{reach}^+}^x(\event_1,\event_2) \land \lloc(\event_2)=x \land \op(\event_2)=\ck\right))\nonumber
\end{align}
Finally, $\msorfaug$ is similar to $\msorf$, with additional checks that assure that every $\rf$ points to an used read event, and unused read events have no incoming $\rf$-edges. 
\begin{align*}
\msorfaug \triangleq \ (\forall \event_1, \event_2. (\rf(\event_1,\event_2)\implies & \op(\event_1)\in \{\wt,\rmw\} \land  \op(\event_2)\in \{\rd,\rmw\})\numberthis\\
\land \ & (\op(\event_2)= \rd \implies \varphi_\text{used}(\event_2))\\
\land \ &\lloc(\event_1)=\lloc(\event_2) \land \val_{\wt}(\event_1)=\val_{\rd}(\event_2)\\
\land \ & \lnot \exists \event_3. (\event_3\neq \event_1 \land \rf(\event_3,\event_2)))\\
\land \ & \forall \event_1. (\op(\event_1)=\rmw \implies \exists \event_2. \rf(\event_2,\event_1))\\
\land \ & \forall \event_1. ((\op(\event_1)=\rd \land \varphi_\text{used}(\event_1))\implies \exists \event_2. \rf(\event_2,\event_1))
\end{align*}

$\mm$-consistency checking works the same as in non-augmented $\rf$-graphs, therefore we may use the formulas $\Phi_{\mm}$ presented in \cref{SUBSEC:MSO_DEFINABLE_MM}. 
The set  $\RedGraphsOf{\ConcProg}{\mm}$ can thus be expressed by the conjunction of the provided formulas.
\end{proof}

\SubParagraph{2.~MSO-definability of reduced $\rf$-semantics.}
At this point, we may attempt to define the reduced $\rf$-semantics of an augmented program $\ConcProgAug$ in MSO, following our process in \cref{SEC:MSO} for the $\rf$-semantics.
This entails the new challenge of ensuring that the $\rf$ component of a reduced $\rf$-graph $H$ omits \emph{only the $\rf$-edges of unused reads}.
Our second insight is that the unused reads of $H$ can be captured in MSO which, in turn, allows us to express the well-formedness of $H.\rf$ in MSO.
The conditions that
\begin{enumerate*}[label=(\roman*)]
\item $H$ represents an execution of $\ConcProgAug$, and
\item $H\models \mm$, for some MSO-definable model $\mm$
\end{enumerate*}
remain essentially the same in MSO as in the case of $\rf$-graphs.

\SubParagraph{3.~Treewidth contiguity of reduced $\rf$-semantics.}
Finally, \cref{thm:robustness_vs_reachability} depends on the $1$-tw-contiguity of the corresponding memory models.
Extending this notion to reduced $\rf$-graphs, we call a memory model \emph{$\ell$-tw-contiguous wrt reduced graphs} if for every augmented program $\ConcProgAug$, for every reduced $\rf$-graph $H\in \RedGraphsOf{\ConcProgAug}{\mm}$ with $\ell\leq \tw(H)$, there exists some $H'\in \RedGraphsOf{\ConcProgAug}{\mm}$ with $\tw(H)-\ell\leq \tw(H')<\tw(H)$.
Our third insight is that $(\seqbef\cup\rf)$-contiguity is also a sufficient condition for $1$-tw-contiguity wrt reduced graphs,
although the proof is less straightforward.

\smallskip
\begin{restatable}{lemma}{lemreducedtwcontiguousmodels}\label{lem:reduced_tw_contiguous_models}
Every $(\seqbef\cup \rf)$-contiguous memory model (and thus all memory models in \cref{tab:main_models} except $\rlxmm$) is also $1$-tw-contiguous wrt reduced $\rf$-graphs.
\end{restatable}
\begin{proof}
Consider any augmented program $\ConcProgAug$ and some reduced $\rf$-graph $H\in \RedGraphsOf{\ConcProgAug}{\mm}$, let $\event$ be a $(\seqbef\cup \rf)^+$ maximal event of $H$.
\begin{compactitem}
\item If $\op(\event)\neq \ck$, let $H'$ be the reduced $\rf$-graph obtained by removing $\event$ from $H$.
Note that $H'\in \RedGraphsOf{\ConcProgAug}{\mm}$, and $\tw(H')\geq \tw(H)-1$.
\item Otherwise, let $\event=\ck(t,x)$, for some thread $t$ and shared register $x$.
We construct a reduced $\rf$-graph $H'$ by removing from $H$
\begin{enumerate*}[label=(\roman*)]
\item $\event$, and
\item the $\rf$-edge of $\seqbef$-maximal read event $\rd(t,x)$.
\end{enumerate*}
Removing $\event$ does not decrease the treewidth of the underlying graph, since $\event$ has exactly one edge.
Then, removing the $\rf$-edge can only decrease the treewidth by $1$, and thus $\tw(H')\geq \tw(H)-1$.
\end{compactitem}
If $\tw(H')=\tw(H)-1$, we are done.
Otherwise, we repeat the process on $H'$.
Since every step reduces the size of the graph by $1$ event, and the treewidth of the graph is bounded by its size, we are guaranteed to reach a graph $H''$ such that $\tw(H'')=\tw(H)-1$.
The desired result follows.
\end{proof}

We can now combine the above insights into an algorithm.
Given any memory model $\mm\in \{ \ramm, \rmwfree{\sramm}, \wramm \}$ and an augmented concurrent program $\ConcProgAug$,
we first use \cref{lem:reduced_semantics_MSO} and \cref{lem:courcelle_theorem_exact_tw} for $k=|\ThreadDom|+|\LocationDom|$ to check if $\RedGraphsOf{\ConcProgAug}{\mm}$ contains a reduced $\rf$-graph of treewidth exactly $k+1$.
If not, then \cref{lem:reduced_tw_contiguous_models} guarantees that $\RedGraphsOf{\ConcProgAug}{\mm}$ has treewidth $\leq k$,
at which point \cref{cor:mso_satisfiability} implies that $\tm$ reachability for $\ConcProgAug$ is decidable.
Otherwise, $\ConcProgAug$ is not observationally $\rf$-robust, since the set $\RedGraphsOf{\ConcProgAug}{\scmm}$ has treewidth $\leq k$.

%% file: app_rc20.tex
\section{Extensions to $\rcmm$}\label{SEC:APP_RC20}

In this section, we formally define the $\rcmm$ memory model~\cite{Margalit2021}, and show that it is 1-$\tw$-contiguous and MSO-definable. 
As a result, \cref{thm:mso_definable_mms,lem:bounded_tw_reachability,lem:bounded_tw_diff,thm:robustness_vs_reachability,thm:obs_robustness_vs_reachability} can be extended to also include this model.

\subsection{The $\rcmm$ memory model}\label{SUBSEC:APP_RC20_MM}

The $\rcmm$ memory model is a variant of the C11 model that mixes \emph{release/acquire} and \emph{relaxed} accesses and additionally synchronizes through memory fences.
We phrase it in the notation of \cref{SUBSEC:PRELIM_MM}; the only additions are access modes on events and a new event type for fences.

\Paragraph{Access modes and fences.}
Unlike the models of \cref{tab:main_models}, $\rcmm$ distinguishes events by an access mode.
The set of access modes is $\ModDom = \{\mrlx, \mrel, \macq, \mra\}$ (relaxed, release, acquire, and acquire-release), partially ordered by the relation $\sqsubseteq$, which is the smallest reflexive-transitive relation such that $\mrlx \sqsubseteq \macq \sqsubseteq \mra$ and $\mrlx \sqsubseteq \mrel \sqsubseteq \mra$.
Each event $e$ carries a mode $\acc(e) \in \ModDom$, constrained by its operation: reads use $\{\mrlx,\macq\}$, writes use $\{\mrlx,\mrel\}$, and RMWs use any mode.
We further add \emph{fence} events, given by a label $\fc(t,o)$ with $t \in \ThreadDom$ and $o \in \{\macq,\mrel,\mra\}$.
Fences carry neither a location nor a value; for such events $\op(e) = \fc$
and $\acc(e) = o$. 
We let $\FenceDom = \{e \in \Events : \op(e) = \fc\}$, and, for a mode $o$, we write
$\Events^{\sqsupseteq o} = \{e \in \Events : \acc(e) \sqsupseteq o\}$; in particular
$\Events^{\sqsupseteq\mrel}$ and $\Events^{\sqsupseteq\macq}$ represent the release and acquire events, respectively. 
Fences participate in neither $\rf$ nor $\mo$.

\Paragraph{Synchronization and happens-before.}
$\rcmm$ uses a \emph{synchronizes-with} relation $\sw$, defined as follows.
\[\sw \triangleq \identity{\Events^{\sqsupseteq\mrel}} ; (\identity{\FenceDom};\po)^? ; \rf^+ ; (\po;\identity{\FenceDom})^? ; \identity{\Events^{\sqsupseteq\macq}}\]
This relation is used to define the $\rcmm$-equivalent of the happens-before relation: $\hbrc \triangleq (\po \cup \sw)^+$.

\Paragraph{Consistency.}
The axiomatic definition of $\rcmm$ uses the from-reads relation $\fr \triangleq \rf^{-1};\mo$ of \cref{SUBSEC:PRELIM_MM}, and is defined by the following four axioms.

{
\footnotesize
\begin{center}
\begin{tabular}{ll|ll}
$\irreflexive(\mo;\rf^?;\hbrc^?)$ & \customlabel{eq:rcwc}{(rc-write-coherence)} & $\irreflexive(\fr;\rf^?;\hbrc)$ & \customlabel{eq:rcrc}{(rc-read-coherence)} \\
$\irreflexive(\mo;\mo;\rf^{-1})$ & \ref{eq:at} & $\irreflexive(\hb)$ & \ref{eq:hbirr} \\
\end{tabular}
\end{center}
}

Given an execution graph $G = \tuple{\Events, \seqbef, \rf, \mo}$, we say that $G$ is \emph{consistent in $\rcmm$}, denoted $G \models \rcmm$, if $G$ satisfies \ref{eq:rcwc}, \ref{eq:rcrc}, \ref{eq:at}, and \ref{eq:hbirr}. 
The release/acquire fragment of $\rcmm$ (in which every read is $\macq$, every write is $\mrel$, and there are no relaxed accesses or fences) recovers $\ramm$~\cite{Margalit2021}.

\Paragraph{1-$\tw$-contiguity.} 
Note, since $\rcmm$ requires \ref{eq:hbirr}, the model is also $(\seqbef \cup \rf)$-contiguous, and thus 1-$\tw$-contiguous by \cref{lem:tw_contiguous_models}.

\subsection{The MSO-definability of $\rcmm$}\label{SUBSEC:APP_RC20_MSO}

Here we show that, for any concurrent program $\ConcProg$, the set $\RFGraphsOf{\ConcProg}{\rcmm}$ is MSO-definable, therefore extending \cref{thm:mso_definable_mms} to $\rcmm$.

We follow the structure of the MSO formula $\msora$ for $\ramm$.
The formula $\msorc$ is a conjunction of three formulas: $\msohbirr$ (\cref{eq:mso_hbirr}), enforcing \ref{eq:hbirr}; $\msowat$ (\cref{eq:mso_wat}), enforcing \ref{eq:wat}; and $\msorcwcrc$, enforcing \ref{eq:rcwc}, \ref{eq:rcrc}, and \ref{eq:at} under the assumption that \ref{eq:wat} already holds.
The formulas $\msohbirr$ and $\msowat$ can be lifted verbatim from \cref{SUBSEC:MSO_DEFINABLE_MM}.

Compared to $\msora$, two changes are required to $\msorcwcrc$.
First, not every $\rf$-edge induces synchronization in $\rcmm$; instead, synchronization is captured by the relation $\sw$, which we express in MSO and use to define a happens-before formula $\msohbrc$ for $\hbrc$.
Second, the coherence formula must infer $\mo$-orderings from the weaker
premises $\rf^?;\hbrc$ of \ref{eq:rcwc} and \ref{eq:rcrc}, instead of $\hb$.

\SubParagraph{Expressing $\msohbrc$.}
We extend the vocabulary of MSO formulas in the natural way: events may carry
the additional operation $\fc$, and the label predicate $\acc(\event) = o$
exposes the access mode $o \in \ModDom$ of an event $\event$.
The relation $\sw$ is expressed by the following formula, where the existentially
quantified events $\event_3, \event_4$ capture the endpoints of the $\rf^+$ path,
which coincide with $\event_1$ and $\event_2$ unless the latter are fences.
\begin{align}
\msosw(\event_1,\event_2) \triangleq & \ \acc(\event_1) \sqsupseteq \mrel \land \acc(\event_2) \sqsupseteq \macq \ \land \label{eq:mso_sw}\\
& \exists \event_3,\event_4. \big( \nonumber\\
&\quad (\event_1 = \event_3 \lor (\op(\event_1)=\fc \land \varphi_{\seqbef^+}(\event_1,\event_3))) \ \land \nonumber\\
&\quad (\event_4 = \event_2 \lor (\op(\event_2)=\fc \land \varphi_{\seqbef^+}(\event_4,\event_2))) \ \land \nonumber\\
&\quad \varphi_{\rf^+}(\event_3,\event_4)\big)\nonumber
\end{align}
We then express that $\event_1 \LTo{\hbrc} \event_2$ in the formula
$\msohbrc = \varphi_{R^+}$, where $\varphi_{R^+}$ is the transitive-closure
formula (\cref{eq:mso_tc}) and $R(\event_1, \event_2) = \msosw(\event_1,\event_2) \lor \seqbef(\event_1,\event_2)$.

\SubParagraph{Expressing $\msorcwcrc$.}
Similarly to $\ramm$, the coherence and atomicity axioms of $\rcmm$ involve the $\mo$ relation, which is not present in an $\rf$-graph and thus has to be guessed.
We use the same approach as \cref{SUBSEC:MSO_DEFINABLE_MM}: assuming \ref{eq:wat},
the write and RMW events decompose into disjoint RMW chains, and it suffices to
infer $\mo$-orderings between the writes heading these chains and require their
acyclicity.

The difference is that in $\rcmm$, the axioms \ref{eq:rcwc} and \ref{eq:rcrc} order writes already when the conflicting triplet $\tuple{\event_3, \event_4, \event_5}$ of $\msotriplet$ satisfies the weaker condition $\event_3 \LTo{\rf^?;\hbrc} \event_5$, rather than $\event_3 \LTo{\hb} \event_5$.
Given the MSO-definability of $\hbrc$, this condition is clearly expressible
in MSO.
\[
\varphi_{(\rf^?;\hbrc)}(\event_1,\event_2) \triangleq \msohbrc(\event_1,\event_2)\lor \exists \event_3.\left(\rf(\event_1, \event_3) \land \msohbrc(\event_3,\event_2)\right)
\numberthis\label{eq:mso_rf_hbrc}
\]
The formula inferring $\mo$-orderings between write events follows $\msopmora$
(\cref{eq:mso_pmora}), but adapted to the $\rf^?;\hbrc$ conditions of \ref{eq:rcwc} and \ref{eq:rcrc}.
In particular, since $\rf \not\subseteq \hbrc$, an ordering between the heads $\event_1, \event_2$ may now also arise from a constraint of the form $\event_6 \LTo{\rf^?;\hbrc} \event_2$, where $\event_6$ is an event in the RMW-chain headed by $\event_1$.
\begin{align}
\msopmorc(\event_1,&\event_2) \triangleq  \ \op(\event_1) = \wt \land \op(\event_2) = \wt \land \lloc(\event_1)=\lloc(\event_2) \ \land \nonumber\\
& \big[\exists \event_6. \big(\varphi_{\rf^*}(\event_1, \event_6) \land \varphi_{(\rf^?;\hbrc)}(\event_6,\event_2)\big)
\lor \exists \event_3, \event_4, \event_5. \big(\msotriplet(\event_3,\event_4,\event_5) \land \nonumber\\
& \varphi_{(\rf^?;\hbrc)}(\event_3,\event_5)  \land \lnot \varphi_{\rf^+}(\event_3,\event_4) \land \varphi_{\rf^*}(\event_1, \event_3) \land \varphi_{\rf^*}(\event_2, \event_4) \big) \big]\numberthis\label{eq:mso_pmorc}
\end{align}

Finally, we express coherence and atomicity as the requirement that $\msopmorc$
is acyclic, using the transitive-closure formula $\varphi_{R^+}$ from
\cref{eq:mso_tc} with $R=\msopmorc$,
\[
\msorcwcrc \triangleq \lnot \exists \event. \varphi_{R^+}(\event,\event)
\numberthis\label{eq:mso_rc_wc_rc}
\]
and define
\[
\msorc \triangleq \msohbirr \land \msowat \land \msorcwcrc
\numberthis\label{eq:mso_rc}
\]

\smallskip
\begin{restatable}{lemma}{lemrcmso}\label{lem:rc_mso}
For any $\rf$-graph $G=\tuple{\Events, \seqbef, \rf}$, we have $G\models \msorc$ iff $G\models \rcmm$.
\end{restatable}
\begin{proof}
The proof follows the lines of \cref{lem:ra_mso}, replacing $\hb$ by $\hbrc$
and the premise $\event_3 \LTo{\hb} \event_5$ of each conflicting triplet by
$\event_3 \LTo{\rf^?;\hbrc} \event_5$, as per \ref{eq:rcwc} and \ref{eq:rcrc}.
\end{proof}